\documentclass[12pt]{article}
\usepackage{amsfonts,amsmath,graphicx,amssymb,geometry,arydshln,textcomp,amsthm,hyperref,framed,color,appendix,lscape}
\usepackage[round]{natbib}
\usepackage[onehalfspacing]{setspace}
\usepackage{ragged2e}
\usepackage{xcolor,xspace,colortbl,rotating}
\usepackage{tabulary}
\usepackage{booktabs}
\usepackage{multirow}
\usepackage{bm}
\usepackage{cleveref}
\usepackage{threeparttable}
\usepackage{subfig}
\usepackage{appendix}
\usepackage{etoc}
\DeclareGraphicsExtensions{.pdf,.svg,.eps,.ps,.png,.jpg,.jpeg}
\theoremstyle{definition}
\newtheorem {theorem}{Theorem}

\newtheorem{assumption}{Assumption}

\newtheorem {corollary}{Corollary}

\newtheorem {lemma}{Lemma}

\newtheorem {proposition}[theorem]{Proposition}
\newtheorem {remark}[theorem]{Remark}

\crefname{subequation}{subequation}{subequations}
\Crefname{subequation}{Assumption}{assumptions}
\crefname{assumption}{Assumption}{Assumptions}
\hypersetup{colorlinks,linkcolor={blue},citecolor={blue},urlcolor={blue}}

\begin{document}

\author{Xiao Huang\thanks{ Corresponding author. Department of Economics, Finance, and Quantitative Analysis, Coles College of Business, Kennesaw State University, GA 30144, USA. Email: xhuang3@kennesaw.edu.} \ \ \ \ \ \ \  Zhaoguo Zhan\thanks{Department of Economics, Finance, and Quantitative Analysis, Coles College of Business, Kennesaw State University, GA 30144, USA. Email: zzhan@kennesaw.edu.} }
\title{\Large Local Composite Quantile Regression for Regression Discontinuity}
\date{ \today}
\maketitle

\doublespace

\begin{abstract}
We introduce the local composite quantile regression (LCQR) to causal inference in regression discontinuity (RD) designs. \cite{kai2010local}  study the efficiency property of LCQR, while we show that its nice boundary performance  translates to accurate estimation of treatment effects in RD under a variety of data generating processes. Moreover, we propose a bias-corrected  and standard error-adjusted \textit{t}-test for inference, which leads to confidence intervals with good coverage probabilities. A bandwidth selector   is also discussed.  For illustration, we conduct a simulation study and revisit a classic example from \cite{lee2008randomized}. A companion R package \texttt{rdcqr} is  developed.
\end{abstract}

\bigskip


\bigskip


\newpage  

\normalsize

\doublespace

\section{Introduction} \label{intro}

Over the past few decades, regression discontinuity (RD) has become a popular quasi-experimental method to identify the local average treatment effect. In its simplest form, the sharp RD design, a unit $i$ receives treatment  if and only if an underlying variable $X_i$ attains a pre-specified cutoff. Under some smoothness assumptions, units in a small neighborhood of the cutoff share similar characteristics so that the difference between the outcomes  above and below the cutoff can be interpreted as the local average treatment effect. For recent discussions of the literature and background references, see, e.g., \citeauthor{cattaneo2019practical} (\citeyear{cattaneo2019practical}, \citeyear{cattaneo2020practical}).

Central to a large portion of the empirical and theoretical work on RD is the use of local linear regression (LLR, see  \cite{fan1996local}), while the goal of this paper is to introduce another nonparametric smoother to the estimation and inference in sharp as well as fuzzy RD designs. Although LLR is the best linear smoother, it is possible to use an alternative estimator in RD if we consider the larger class of nonlinear smoothers. One such example is the local composite quantile regression (LCQR) method in \cite{kai2010local}, who show that LCQR can have some efficiency gain against LLR when data are non-normal. The literature related to LCQR is growing. Several recent examples include  \cite{kai2011new},  \cite{zhao2014efficient}, \cite{li2016local}, and \cite{huanglin2020}, although none of these studies focuses on RD. To use LCQR, the researcher chooses a finite number of quantile positions, say $5$, and uses a local polynomial to estimate the quantile (the intercept) at each of the 5 quantile positions. Averaging the 5 quantile estimates gives an estimate for the conditional mean, whose values above and below the pre-specified cutoff are the key components in sharp and fuzzy RD designs.

Our paper makes the following contributions to the large and growing literature on RD. First, we introduce the LCQR method  to sharp and fuzzy RD designs. Numerical evidence for the efficiency gain of using LCQR instead of LLR in RD  is also provided. Second, similar to the robust $t$-test in \cite{calonico2014robust} for conducting inference on causal effects, we  propose a \textit{t}-test that adjusts both the bias and the standard error  of the LCQR estimator, and the resulting confidence intervals are shown to have good coverage probabilities. Third, we further discuss a new bandwidth selector that is based on an adjusted mean squared error (MSE). As a byproduct of our research, we also develop an R package \texttt{rdcqr} that implements the LCQR method in this paper. The \texttt{rdcqr} package can be downloaded from \url{https://github.com/xhuang20/rdcqr}.

The rest of the paper unfolds as follows. Section 2 sets up the notation for RD designs and introduces the LCQR method. Section 3 presents the bias correction and standard error adjustment for inference on the boundary. Section 4 briefly discusses several extensions, such as allowing for covariates and kink RD designs.  Section 5 presents a simulation experiment. An empirical illustration of LCQR for RD is provided in Section 6 using the data from \cite{lee2008randomized}. Section 7 concludes. The Online Supplement contains all technical details and proofs, as well as additional figures and tables.

\section{Regression discontinuity and local composite quantile regression}

We introduce the notation for RD  and the application of LCQR in RD in this section.

\subsection{Regression discontinuity: setup and notation}

Consider the triplet $(Y_i,X_i,T_i)$, $i=1,\cdots,n$ in a standard RD setup that nests both the sharp and fuzzy RD designs. $Y_i$ is the observed outcome for individual $i$, $X_i$ is the variable that determines the assignment of treatment, and $T_i$ equals $1$ if individual $i$ receives treatment and $0$ otherwise. For individual $i$ to receive treatment, a threshold value for $X_i$ is set. Without loss of generality, we assume the threshold is $0$, and the treatment assignment rule becomes that  individual $i$ is assigned to the treatment group if $X_i \geq 0$. Throughout the paper, the signs $+$ and $-$  will be used to denote data or quantities associated with $X_i \geq 0$ and $X_i < 0$, respectively. 

We use a general nonparametric model to describe the relationship between $Y_i$ and $X_i$:
\begin{align}
	Y_{+,i} &= m_{Y_+}(X_{+,i}) + \sigma_{\epsilon_{Y_+}}(X_{+,i})\epsilon_{Y_+,i} \label{eq:Y plus model}, \\
	Y_{-,i} &= m_{Y_-}(X_{-,i}) + \sigma_{\epsilon_{Y_-}}(X_{-,i})\epsilon_{Y_-,i} \label{eq:Y minus model},
\end{align}
where $m_{Y_+}, m_{Y_-}, \sigma_{\epsilon_{Y_+}}, \sigma_{\epsilon_{Y_-}}, \epsilon_{Y_+}$ and $\epsilon_{Y_-}$ are the conditional mean functions, conditional standard deviation functions, and error terms with unit variance, respectively.

In a sharp RD design, whether individual $i$ receives treatment is determined by whether $X_i \geq 0$, so $T_i=I\{X_i \geq 0\}$. One is interested in measuring the average treatment effect at the threshold $0$: 
\begin{equation} \label{eq:sharp treatment}
	\tau_{\text{sharp}} = m_{Y_+}(0) - m_{Y_-}(0),
\end{equation}
where
\begin{equation} \label{eq:limit m_Y}
	m_{Y_+}(0) = \lim\limits_{x_+\downarrow 0}m_{Y_+}(x_+) \text{ and } m_{Y_-}(0) = \lim\limits_{x_-\uparrow 0}m_{Y_-}(x_-).  
\end{equation}

A \textit{t}-statistic for testing a hypothesized treatment effect $\tau_0$ reads:
\begin{equation} \label{eq:sharp t}
	t_{\text{sharp}} = \frac{\hat{\tau}_{\text{sharp}} - \tau_0}{\sqrt{\text{Var}(\hat{\tau}_{\text{sharp}})}},
\end{equation}
where, under the \textit{i.i.d.}  assumption, we have

\begin{equation} \label{eq:sharp nume and deno}
     \hat{\tau}_{\text{sharp}} = \hat{m}_{Y_+}(0) - \hat{m}_{Y_-}(0),  \ \ 
     \text{Var}(\hat{\tau}_{\text{sharp}}) = \text{Var}(\hat{m}_{Y_+}(0)) + \text{Var}(\hat{m}_{Y_-}(0)).
\end{equation}

In a fuzzy RD design, the treatment assignment rule remains the same, but $T_i$ is not necessarily equal to $1$ when $X_i \geq 0$ or $0$ when $X_i < 0$. It is customary to use a nonparametric function to describe the probability of receiving treatment:

\begin{align}
	T_{+,i} &= m_{T_+}(X_{+,i}) + \sigma_{\epsilon_{T_+}}(X_{+,i})\epsilon_{T_+,i} \label{eq:T plus model}, \\
    T_{-,i} &= m_{T_-}(X_{-,i}) + \sigma_{\epsilon_{T_-}}(X_{-,i})\epsilon_{T_-,i} \label{eq:T minus model}.
\end{align}
The treatment effect measurement becomes
\begin{equation} \label{eq:fuzzy treatment}
	\tau_{\text{fuzzy}} = \frac{m_{Y_+}(0) - m_{Y_-}(0)}{m_{T_+}(0) - m_{T_-}(0)},
\end{equation}
where $m_{Y_+}(0)$, $m_{Y_-}(0)$ are as defined in (\ref{eq:limit m_Y}) above, and similarly,
\begin{equation} \label{eq:limit m_T}
    m_{T_+}(0) = \lim\limits_{x_+\downarrow 0}m_{T_+}(x_+) \text{ and } m_{T_-}(0) = \lim\limits_{x_-\uparrow 0}m_{T_-}(x_-). 
\end{equation}

The corresponding \textit{t}-statistic is given by 
\begin{equation} \label{eq:fuzzy t}
	t_{\text{fuzzy}} = \frac{\hat{\tau}_{\text{fuzzy}} - \tau_0}{\sqrt{\text{Var}(\hat{\tau}_{\text{fuzzy}})}},
\end{equation}
where
\begin{equation}\label{eq:fuzzy nume and deno}
    \hat{\tau}_{\text{fuzzy}} = \frac{\hat{m}_{Y_+}(0) - \hat{m}_{Y_-}(0)}{\hat{m}_{T_+}(0) - \hat{m}_{T_-}(0)}, 
\end{equation}
and $\text{Var}(\hat{\tau}_{\text{fuzzy}})$ depends on the variances and covariances of $\hat{m}_{Y_+}(0)$, $\hat{m}_{Y_-}(0)$, $\hat{m}_{T_+}(0)$ and $\hat{m}_{T_-}(0)$, which we provide in the Online Supplement (see (A.16) in Section S.2) to conserve space.

To conduct the \textit{t}-test, we thus need nonparametric estimates for all quantities in \cref{eq:sharp nume and deno,eq:fuzzy nume and deno} at the boundary point $0$.

\subsection{Local composite quantile regression}

LLR is the leading method to estimate quantities in \cref{eq:sharp nume and deno,eq:fuzzy nume and deno}. LCQR is introduced in \cite{kai2010local} as an alternative to the LLR method due to its potential efficiency gains with non-normal errors. The probable advantage under non-normality is the key reason we propose to use LCQR in RD as normality  can be easily violated in practice.

The same LCQR method can be applied to \cref{eq:Y plus model,eq:Y minus model,eq:T plus model,eq:T minus model}. Because of this similarity, we describe the LCQR method based on \cref{eq:Y plus model}. For $k=1,\cdots,q$, let $\tau_k = k/(q+1)$ be the equally spaced $q$ quantile positions and $\rho_{\tau_k}(r) = \tau_k r -r I(r < 0)$ be the $q$ check loss functions in quantile regression. The loss function for LCQR at the point $x$ is defined as
\begin{equation} \label{eq:lcqr obj}
	\sum_{k=1}^{q} \sum_{i=1}^{n_+}\rho_{\tau_k} \left(Y_{+,i} - a_k - b (X_{+,i} - x)\right) K\left(\frac{X_{+,i} - x}{h_{Y+}}\right),
\end{equation}
where $a_k$ is the \textit{k}-th quantile estimand, $b$ is a slope restricted to be the same across quantiles, $h_{Y+}$ is the bandwidth, and $K$ is a kernel function.   Here we use the linear approximation to  the conditional mean function in \cref{eq:Y plus model}. Let  $F_{Y|X}$ be the cumulative distribution function of $Y$ given $X$ so that $a_k = F_{Y|X = x}^{-1}(\tau_k)$ is the \textit{k}-th quantile,  and $b$ is $m_{Y_+}^{(1)}(x)$, the first derivative of the conditional mean function at $x$. Using a single slope $b$ allows us to combine information across quantiles, since a specification allowing for $b_k$ in \cref{eq:lcqr obj} would be equivalent to estimating each quantile separately. Our point of interest in \cref{eq:lcqr obj} is $x = 0$.


The parameter $q$ plays an important role here. Once $q$ is chosen, there will be $q$ quantiles corresponding to  the $q$ quantile positions $\tau_k$. Since we use equally spaced quantile positions, the median will always be selected as the middle quantile if $q$ is odd. For example, if $q=3$, the quantile positions are $(0.25,0.5,0.75)$. The goal is to spread the quantile positions on the interval $(0,1)$ so that LCQR can combine information in multiple quantiles. The value of $q$ is decided by the researcher. Results in \cite{kai2010local} and our discussion later on show that a small number such as 5, 7 or 9 may be adequate for many common non-normal errors. It can be viewed as a hyperparameter and can be tuned.

Minimizing \cref{eq:lcqr obj} w.r.t. $a_k$ and $b$ yields $(\hat{a}_1,\cdots,\hat{a}_q,\hat{b})$. The LCQR estimator for $m_{Y+}(x)$ is defined as
\begin{equation} \label{eq:lcqr estimate}
	\hat{m}_{Y_+}(x) = \frac{1}{q}\sum_{k=1}^{q} \hat{a}_k
\end{equation}
and $\hat{b}$ is the estimator for $m_{Y_+}^{(1)}(x)$, the first derivative of $m_{Y_+}(x)$. In a similar fashion, $\hat{m}_{Y_-}(x)$, $\hat{m}_{T_+}(x)$, and $\hat{m}_{T_-}(x)$ can be obtained at $x = 0$. The average in \cref{eq:lcqr estimate} combines information across different quantile estimators. The hope is to possibly improve the relative efficiency of individual nonparametric quantile estimator to the nonparametric mean estimator based on LLR.

\cite{kai2010local} discuss the asymptotic bias and variance of the LCQR estimator at an interior point $x$. The asymptotic properties at a boundary point $x$ are discussed in \cite{kai2010local_supp}. Given the results in \cite{kai2010local_supp}, one can immediately implement the \textit{t}-test for the sharp RD in \cref{eq:sharp t} with  bias correction if needed. To implement the \textit{t}-test for the fuzzy RD in \cref{eq:fuzzy t}, additional covariance expressions  for \cref{eq:fuzzy nume and deno} need to be estimated and we provide these results in Lemma 3 in the Online Supplement. The asymptotic results for the bias and variance of $\hat{m}_{Y_+}(x)$, $\hat{m}_{Y_-}(x)$, $\hat{m}_{T_+}(x)$, and $\hat{m}_{T_-}(x)$ follow those in \cite{kai2010local_supp}, and are given in Lemma 2 in the Online Supplement.

We make the following assumptions. Let $c$ be a positive constant of supp($K$).

\begin{assumption} \label{assumption:m continuity}
	$m_{Y_+}(\cdot)$, $m_{T_+}(\cdot)$, $m_{Y_-}(\cdot)$, and $m_{T_-}(\cdot)$ are at least three times continuously differentiable at $x$.  
\end{assumption}
\begin{assumption} \label{assumption:sigma continuity}
	$\sigma_{\epsilon_{Y_+}}(\cdot)$, $\sigma_{\epsilon_{T_+}}(\cdot)$, $\sigma_{\epsilon_{Y_-}}(\cdot)$, and $\sigma_{\epsilon_{T_-}}(\cdot)$ are right or left continuous and differentiable at $x$.  
\end{assumption}
\begin{assumption} \label{assumption:kernel}
	The kernel function $K$ is bounded and symmetric.
\end{assumption}
Following \Cref{assumption:kernel}, we define, for $j=0$, $1$, $2$, $\cdots$, 
	\begin{align*}
		\mu_{+,j}(c) &= \int_{-c}^{\infty} u^j K(u) du, \qquad \nu_{+,j} = \int_{-c}^{\infty} u^j K^2(u) du,\\
		\mu_{-,j}(c) &= \int_{-\infty}^{c} u^j K(u) du, \qquad \nu_{-,j} = \int_{-\infty}^{c} u^j K^2(u) du.
	\end{align*}
\begin{assumption} \label{assumption:density x}
	The marginal density $f_{X_+}(\cdot)$ is right continuous, differentiable and positive at $x = 0$, and $f_{X_-}(\cdot)$ is left continuous, differentiable and positive at $x = 0$.
\end{assumption}
\begin{assumption} \label{assumption:error distribution}
	All error distributions for $\epsilon_{Y_+}$, $\epsilon_{T_+}$, $\epsilon_{Y_-}$, and $\epsilon_{T_-}$ are symmetric and have positive density. All errors are \textit{i.i.d.}
\end{assumption}
\begin{assumption} \label{assumption:bandwidth}
	$h_{Y_+} \rightarrow 0$, $n_+h_{Y_+} \rightarrow \infty$, and $n_+h_{Y_+}^7 \rightarrow 0$ as $n_+ \rightarrow \infty$. The same condition also holds for $(n_+,h_{T_+}),(n_-,h_{Y_-})$, and $(n_-,h_{T_-})$.
\end{assumption}

\Cref{assumption:m continuity,assumption:sigma continuity} are used for Taylor series expansions in the proof. The bias expressions require only second-order differentiability of the conditional mean function. Higher-order differentiability permits the study of small order terms in the expansion. The support of the kernel is assumed to be $(-\infty,\infty)$ in \Cref{assumption:kernel}. If the kernel has a bounded support such as $[-1,1]$, we would have $\mu_{+,j}(c) = \int_{-c}^{1} u^j K(u) du$. Similar changes can be made to the $\nu$ variables. The differentiability in \Cref{assumption:density x} is also used for Taylor  expansions in the proof,  and we  require  the positivity of the density at the boundary as it appears in the denominator of the bias expression. The symmetric error distribution assumption in \Cref{assumption:error distribution} helps to remove an extra term in the bias. When the error distribution is asymmetric, this extra term can be removed during bias correction. In \Cref{assumption:bandwidth},  $h_{Y_+} \rightarrow 0$ and $n_+h_{Y_+} \rightarrow \infty$ help establish the consistency of $\hat{m}_{Y_+}$. The assumption $n_+h_{Y_+}^7 \rightarrow 0$ is used to  establish the asymptotic normality of the bias-corrected and s.e.-adjusted $t$-statistic. \Cref{assumption:bandwidth} translates to $h \rightarrow 0$, $nh \rightarrow \infty$, and $nh^7 \rightarrow 0$ as $n \rightarrow \infty$ if we use a single bandwidth $h$ for the data above and below  the cutoff, and $n$ is the total number of observations. Similar assumptions can be found in \cite{calonico2014robust} for LLR.


Next we provide two theorems on the LCQR estimator for sharp and fuzzy treatment effects defined in \cref{eq:sharp nume and deno} and \cref{eq:fuzzy nume and deno}, respectively. Let $\mathbf{X}$ be the $\sigma$-field  generated by all $X_i$.
\begin{theorem} \label{thm:sharp results}
	Under \cref{assumption:m continuity,assumption:sigma continuity,assumption:kernel,assumption:density x,assumption:error distribution,assumption:bandwidth}, as both $n_+$ and $n_- \rightarrow \infty$, we have
	\begin{align}
		\text{Bias}(\hat{\tau}_{\text{sharp}}|\mathbf{X}) &= \frac{1}{2} a_{Y_+}(c) m_{Y+}^{(2)}(0) h_{Y_+}^2 - \frac{1}{2} a_{Y_-}(c) m_{Y-}^{(2)}(0)  h_{Y_-}^2 + o_p(h_{Y_+}^2 + h_{Y_-}^2),\\
		\text{Var}(\hat{\tau}_{\text{sharp}}|\mathbf{X}) &=\frac{b_{Y_+}(c)\sigma_{\epsilon_{Y_+}}^2(0)}{n_+h_{Y_+}f_{X_+}(0)} + \frac{b_{Y_-}(c)\sigma_{\epsilon_{Y_-}}^2(0)}{n_-h_{Y_-}f_{X_-}(0)} + o_p(\frac{1}{n_+h_{Y_+}} + \frac{1}{n_-h_{Y_-}}),
	\end{align}
	where 
	\begin{align}
		a_{Y_+}(c) &= a_{Y_-}(c) = \frac{\mu_{+,2}^2(c) - \mu_{+,1}(c)\mu_{+,3}(c)}{\mu_{+,0}(c)\mu_{+,2}(c) - \mu_{+,1}^2(c)},\label{eq:ac}\\
		b_{Y_+}(c) &= e_q^T (S_{Y_+}^{-1}(c) \Sigma_{Y_+}(c) S_{Y_+}^{-1}(c))_{11} e_q/q^2, \label{eq:bYc plus}\\
		b_{Y_-}(c) &= e_q^T (S_{Y_-}^{-1}(c) \Sigma_{Y_-}(c) S_{Y_-}^{-1}(c))_{11} e_q/q^2, \label{eq:bYc minus}
	\end{align}
	and the matrices $S_{Y_+}(c)$, $\Sigma_{Y_+}(c)$ and similarly, $S_{Y_-}(c)$, $\Sigma_{Y_-}(c)$, are given in the Online Supplement (Section S.1, Equation (A.1)). $e_q$ is the $q$-dimensional vector of ones.
\end{theorem}

Proof of \Cref{thm:sharp results} is a straightforward application of Theorem 2.1 in \cite{kai2010local_supp}. As long as a symmetric kernel is used, the value of $a_{Y_+}(c)$ remains the same if kernel moments on the other side of the cutoff are used. Generally, we have $b_{Y_+}(c) \neq b_{Y_-}(c)$ unless error distributions on both sides of the cutoff are assumed to be the same.

When the bandwidths above and below  the cutoff are assumed to be the same, \Cref{thm:sharp results} reduces to the following corollary. 
\begin{corollary} \label{corollary:sharp single h}
	Under \cref{assumption:m continuity,assumption:sigma continuity,assumption:kernel,assumption:density x,assumption:error distribution,assumption:bandwidth}, if we further assume $h_{Y_+}=h_{Y_-}=h$, then we have
	\begin{align}
	\text{Bias}(\hat{\tau}_{\text{sharp}}|\mathbf{X}) &= \frac{1}{2} \left[ a_{Y_+}(c) m_{Y+}^{(2)}(0) - a_{Y_-}(c) m_{Y-}^{(2)}(0)\right] h^2 + o_p(h^2),\\
	\text{Var}(\hat{\tau}_{\text{sharp}}|\mathbf{X}) &=  \frac{b_{Y_+}(c)\sigma_{\epsilon_{Y_+}}^2(0)}{n_+hf_{X_+}(0)} +  \frac{b_{Y_-}(c)\sigma_{\epsilon_{Y_-}}^2(0)}{n_-hf_{X_-}(0)} + o_p(\frac{1}{n_+h} + \frac{1}{n_-h}).
	\end{align}
\end{corollary}

The asymptotic results for the fuzzy RD are similarly provided in the theorem below.
\begin{theorem} \label{thm:fuzzy results}
	Under \cref{assumption:m continuity,assumption:sigma continuity,assumption:kernel,assumption:density x,assumption:error distribution,assumption:bandwidth}, as both $n_+$ and $n_- \rightarrow \infty$, we have
	\begin{align}
		\text{Bias}(\hat{\tau}_{\text{fuzzy}}|\mathbf{X}) &= \frac{1}{m_{T_+}(0) - m_{T_-}(0)}\left[\frac{1}{2} a_{Y_+}(c) m_{Y+}^{(2)}(0) h_{Y_+}^2 - \frac{1}{2} a_{Y_-}(c) m_{Y-}^{(2)}(0)  h_{Y_-}^2\right] \nonumber \\
		&\quad - \frac{m_{Y_+}(0) - m_{Y_-}(0)}{\left[m_{T_+}(0) - m_{T_-}(0)\right]^2}\left[\frac{1}{2} a_{T_+}(c) m_{T+}^{(2)}(0) h_{T_+}^2 - \frac{1}{2} a_{T_-}(c) m_{T-}^{(2)}(0)  h_{T_-}^2\right]\nonumber\\
		&\quad + o_p(h_{Y_+}^2 + h_{Y_-}^2 + h_{T_+}^2 + h_{T_-}^2),
	\end{align}
and $\text{Var}(\hat{\tau}_{\text{fuzzy}}|\mathbf{X})$ is provided in the proof; see the Online Supplement (A.16) in Section S.2.
\end{theorem}

From \Cref{thm:sharp results,thm:fuzzy results} and using a single bandwidth $h$, we see a lot of similarities between  LCQR and  LLR: both biases have order $O(h^2)$, and both variances have order $O((nh)^{-1})$. In fact, the asymptotic bias expression is identical for the two methods. The variance from LCQR, however, is smaller than that from LLR under many non-normal error distributions, for which we  provide  numerical results in the next subsection.

In addition, in the special case of $q = 1$, we can show that 
\begin{equation} \label{eq:q1}
b_{Y_+}(c) = \frac{\mu_{+,2}^2 \nu_{+,0} - 2 \mu_{+,1}\mu_{+,2}\nu_{+,1} + \mu_{+,1}^2\nu_{+,2}}{(\mu_{+,0}\mu_{+,2}-\mu_{+,1}^2)^2} \frac{1}{q^2}\sum_{k=1}^{q}\sum_{k'=1}^{q} \frac{\tau_{kk'}}{f_{\epsilon_{Y_+}}(c_k)f_{\epsilon_{Y_+}}(c_k')},
\end{equation}
where $\tau_{kk'}=\tau_{k}\wedge\tau_{k'}-\tau_{k}\tau_{k'}$,   $k$, $k'=1$, $2$, ..., $q$, $c_k=F_{\epsilon_{Y_+}}^{-1}(\tau_k)$,  and $F_{\epsilon_{Y_+}}$,  $f_{\epsilon_{Y_+}}$ are the c.d.f.  and p.d.f. of the error distribution, respectively.  This result  similarly holds for $b_{Y_-}(c)$. However,  \cref{eq:q1} does not hold for $q \geq 2$ in general.

\subsection{Efficiency gains on the boundary}

We provide some numerical evidence in support of using LCQR in RD. To facilitate the presentation of the results, we consider \Cref{corollary:sharp single h} with a single bandwidth and replace the conditional density $f_{X_+}(0)$ with the unconditional density $f_{X}(0)$ so that the product $n_+ f_{X_+}(0)$ becomes $nf_{X}(0)$. Similarly, $n_- f_{X_-}(0)$ is also replaced by $nf_{X}(0)$. 

Under standard assumptions, the MSE for the LLR estimator of $\tau_{sharp}$ is given by (see, e.g., \cite{imbens2012optimal})
\begin{equation} \label{eq:MSE LLR}
	\begin{split}
	\text{MSE}_{\text{LLR}} &= \frac{1}{4} \left[ a_{Y_+}(c) m_{Y+}^{(2)}(0) - a_{Y_-}(c) m_{Y-}^{(2)}(0)\right]^2 h^4\\
	&\quad + \frac{1}{nh}\left[ \frac{b(c)\sigma_{\epsilon_{Y_+}}^2(0)}{f_{X}(0)} +  \frac{b(c)\sigma_{\epsilon_{Y_-}}^2(0)}{f_{X}(0)} \right] + o_p(h^4 + \frac{1}{nh}),
	\end{split}
\end{equation}
where 
\begin{equation} \label{eq:bc}
	b(c) = \frac{\mu_{+,2}^2 \nu_{+,0} - 2 \mu_{+,1}\mu_{+,2}\nu_{+,1} + \mu_{+,1}^2\nu_{+,2}}{(\mu_{+,0}\mu_{+,2}-\mu_{+,1}^2)^2}.
\end{equation}

Let $a(c) = a_{Y_+}(c)= a_{Y_-}(c)$. By assuming errors on both sides of the cutoff have the same distribution and letting $b_{Y}(c) = b_{Y_+}(c) = b_{Y_-}(c)$, we can further simplify the variance expression in \Cref{corollary:sharp single h} to have
\begin{align} \label{eq:MSE LCQR}
\text{MSE}_{\text{LCQR}}
           &= \frac{1}{4} \left[ a(c) m_{Y+}^{(2)}(0) - a(c) m_{Y-}^{(2)}(0)\right]^2 h^4 \nonumber \nonumber\\
           &\quad + \frac{1}{nh}\left[ \frac{b(c)\sigma_{\epsilon_{Y_+}}^2(0)}{f_{X}(0)} +  \frac{b(c)\sigma_{\epsilon_{Y_-}}^2(0)}{f_{X}(0)} \right] \frac{b_Y(c)}{b(c)} + o_p(h^4 + \frac{1}{nh}),
\end{align}
where the two MSEs differ by a factor  $\frac{b_Y(c)}{b(c)}$ for the variance term. Minimizing these two MSEs gives optimal bandwidths. Substituting the bandwidths into the MSEs leads to the optimal values, $\text{MSE}^{\text{opt}}_{\text{LLR}}$ and $\text{MSE}^{\text{opt}}_{\text{LCQR}}$. Similar to \cite{kai2010local}, we define the asymptotic relative efficiency (ARE) of LCQR with respect to LLR on the boundary as
\begin{equation} \label{eq:are}
	\text{ARE} = \frac{\text{MSE}^{\text{opt}}_{\text{LLR}}}{\text{MSE}^{\text{opt}}_{\text{LCQR}}} \rightarrow \left[\frac{b_Y(c)}{b(c)}\right]^{-4/5}.
\end{equation}

When $q = 1$, the ratio in (\ref{eq:are}) simplifies to  $\left[\frac{1}{q^2}\sum_{k=1}^{q}\sum_{k'=1}^{q} \frac{\tau_{kk'}}{f_{\epsilon_{Y_+}}(c_k)f_{\epsilon_{Y_+}}(c_k')}\right]^{-4/5}.$ When $q \geq 2$, this simplification does not hold. Yet we can evaluate the ratio in \cref{eq:are} for many common distributions. Using the triangular kernel as an example, we calculate the ratio for the five error distributions listed in \Cref{tab:are}. Using the Epanechnikov kernel gives similar results.

\bigskip

\begin{table}[htp]
	\centering
	\caption{Asymptotic relative efficiency (ARE) of LCQR in sharp RD designs}
	\begin{tabular}{lrrrrr}
		\toprule
		\text{Error distribution}& $q$ = 1  & $q$ = 5 & $q$ = 9 & $q$ = 19 & $q$ = 99 \\
		\midrule
		1. $N(0,1)$& 0.6968 & 0.9290 & 0.9569 & 0.9728 & 0.9819 \\
		2. \text{Laplace with $\mu=0$ and $\sigma = 1$}  & 1.7411 & 1.3315 & 1.2920 & 1.2616 & 1.2303 \\
		3. \textit{t}-distribution with 3 degrees of freedom& 1.4718 & 1.6401 & 1.6144 & 1.5703 & 1.4854 \\
		4. $0.95N(0,1) + 0.05N(0,3^2)$& 0.8639 & 1.1271 & 1.1511 & 1.1579 & 1.1327 \\
		5. $0.95N(0,1) + 0.05N(0,10^2)$& 2.6960 & 3.4578 & 3.4986 & 3.4590 & 2.2632 \\
		\bottomrule
	\end{tabular}%
	\label{tab:are}%
\end{table}%

The ARE for boundary points in \Cref{tab:are} largely follows the same pattern for interior points as shown in Table 1 of \cite{kai2010local}. When errors are normally distributed (Distribution 1), increasing $q$ will quickly bring ARE close to one, reflecting a very small efficiency loss of LCQR with respect to LLR. For non-normal errors (Distributions 2 - 5), there can be large efficiency gains. The  column with $q=1$ in \Cref{tab:are} is identical for both interior and boundary points, a result of using a symmetric kernel function. When $q \geq 2$, the numbers differ from those in Table 1 of \cite{kai2010local}. Overall, \Cref{tab:are} provides the numerical evidence in support of using LCQR in RD when data are non-normal, i.e., the corresponding ARE mostly exceeds 1.

\section{Inference on the boundary}

The goal of inference is related to but different from that of estimation. While one still needs a precise $\hat{\tau}_{\text{sharp}}$ or $\hat{\tau}_{\text{fuzzy}}$, much of the effort is spent on making $t$-statistics in \cref{eq:sharp t,eq:fuzzy t} behave like a standard normal random variable. Common approaches include reducing the impact of bias by under-smoothing, directly adjusting the bias, combining bias correction and standard error (s.e.) adjustment, etc. This section studies the bias-corrected and s.e.-adjusted $t$-statistics  based on LCQR  for conducting inference in sharp and fuzzy RD designs.

\subsection{The sharp case}

From \Cref{thm:sharp results}, the bias-corrected estimator is given by
\begin{equation} \label{eq:sharp bc}
	\hat{\tau}_{\text{sharp}}^{\text{bc}} = \hat{\tau}_{\text{sharp}} - \widehat{\text{Bias}}(\hat{\tau}_{\text{sharp}}),
\end{equation}
where $\widehat{\text{Bias}}(\hat{\tau}_{\text{sharp}}) = \widehat{\text{Bias}}(\hat{m}_{Y_+}) - \widehat{\text{Bias}}(\hat{m}_{Y_-})$ with
\begin{equation}\label{eq:m+ bias hat}
	\widehat{\text{Bias}}(\hat{m}_{Y_+})= \frac{1}{2} a_{Y_+}(c) \hat{m}_{Y_+}^{(2)} h_{Y_+}^2 ,   \text{ \ }
	\widehat{\text{Bias}}(\hat{m}_{Y_-})= \frac{1}{2} a_{Y_-}(c) \hat{m}_{Y_-}^{(2)} h_{Y_-}^2,
\end{equation}
and $\hat{m}_{Y_+}^{(2)}$ and $\hat{m}_{Y_-}^{(2)}$ are the estimators for the second derivatives $m_{Y_+}^{(2)}$ and $m_{Y_-}^{(2)}$, respectively, which can be computed using LCQR with a second-order polynomial.

Simply replacing $\hat{\tau}_{\text{sharp}}$ in \cref{eq:sharp t} with $\hat{\tau}_{\text{sharp}}^{\text{bc}} $ may still lead to undercoverage of the resulting confidence intervals.   \cite{calonico2014robust} observed that the key to fix this problem is to take into consideration  the additional variability introduced by bias correction. Similar idea is applied to inference based on LCQR, i.e., instead of using $\text{Var}(\hat{\tau}_{\text{sharp}})$ in the denominator of \cref{eq:sharp t}, we use $\text{Var}(\hat{\tau}_{\text{sharp}} - \widehat{\text{Bias}}(\hat{\tau}_{\text{sharp}}))$. The additional variability due to bias correction increases the adjusted variance, and the s.e./variance  adjustment naturally improves the coverage of the resulting confidence intervals.  

By the \textit{i.i.d.} assumption, we have
\begin{align}
	\text{Var}(\hat{\tau}_{\text{sharp}} - \widehat{\text{Bias}}(\hat{\tau}_{\text{sharp}})) &= \text{Var}(\hat{\tau}_{\text{sharp}}) + \text{Var}(\widehat{\text{Bias}}(\hat{\tau}_{\text{sharp}})) -2 \text{Cov}(\hat{\tau}_{\text{sharp}},\widehat{\text{Bias}}(\hat{\tau}_{\text{sharp}})) \nonumber\\
	&= \text{Var}(\hat{m}_{Y_+}) + \text{Var}(\hat{m}_{Y_-}) + \text{Var}(\widehat{\text{Bias}}(\hat{m}_{Y_+})) + \text{Var}(\widehat{\text{Bias}}(\hat{m}_{Y_-})) \nonumber \\
	& \quad - 2 \text{Cov}(\hat{m}_{Y_+},\widehat{\text{Bias}}(\hat{m}_{Y_+})) - 2 \text{Cov}(\hat{m}_{Y_-},\widehat{\text{Bias}}(\hat{m}_{Y_-})), \label{eq:sharp var adjusted}
\end{align}
where the expressions for $\text{Var}(\hat{m}_{Y_+})$ and $\text{Var}(\hat{m}_{Y_-})$ are available in \cite{kai2010local_supp, kai2010local} and Lemma 2. We contribute to the literature by deriving the explicit forms of the variances of bias and the covariances in \cref{eq:sharp var adjusted}. 

\begin{theorem} \label{thm:sharp t adjusted dist}
	Under \cref{assumption:m continuity,assumption:sigma continuity,assumption:kernel,assumption:density x,assumption:error distribution,assumption:bandwidth}, as both $n_+ \text{ and }n_- \rightarrow \infty$, the adjusted $t$-statistic for the sharp RD follows an asymptotic normal distribution
	\begin{equation} \label{eq:sharp t adjusted dist}
		t_{\text{sharp}}^{\text{adj.}} = \frac{\hat{\tau}_{\text{sharp}} - \widehat{\text{Bias}}(\hat{\tau}_{\text{sharp}})-\tau_0}{\sqrt{\text{Var}(\hat{\tau}_{\text{sharp}} - \widehat{\text{Bias}}(\hat{\tau}_{\text{sharp}}))}} \overset{d}{\to} N(0,1).
	\end{equation}
	The exact form of the adjusted variance is provided in the proof; see the proof of \Cref{thm:sharp t adjusted dist} in the Online Supplement Section S.2. 
\end{theorem}

Following \Cref{thm:sharp t adjusted dist}, the confidence interval for the nominal $95\%$ coverage becomes $\hat{\tau}_{\text{sharp}} - \widehat{\text{Bias}}(\hat{\tau}_{\text{sharp}}) \pm 1.96 \sqrt{\text{Var}(\hat{\tau}_{\text{sharp}} - \widehat{\text{Bias}}(\hat{\tau}_{\text{sharp}}))}$, which incorporates the variability introduced by bias correction and will thus improve the coverage probability.

\setcounter{theorem}{0}
\begin{remark}
	\Cref{thm:sharp t adjusted dist} assumes that the researcher uses the same bandwidths  ($h_{Y_+}$ and $h_{Y_-}$) for estimation, bias correction, and s.e. adjustment. This simplifies the presentation of the result and the proof. It is also  empirically relevant and appealing  since using a single bandwidth throughout estimation and inference greatly reduces the complexity of implementation, though the researcher might want to use a different bandwidth for every required nonparametric estimation. Moreover, \citeauthor{calonico2018jasa} (\citeyear{calonico2018jasa}, \citeyear{calonico2020optimalbandwidth})
	discuss how using the same bandwidths is optimal in some well-defined senses for the robust $t$-test of  \cite{calonico2014robust}.
\end{remark}

\begin{remark}
	We assume two different bandwidths, $h_{Y_+}$ and $h_{Y_-}$, in \Cref{thm:sharp t adjusted dist}. If they are further assumed to be the same, $h_{Y_+} = h_{Y_-} = h$, then by replacing the conditional density, $f_{X_+}(0)$ and $f_{X_-}(0)$, with the unconditional one, $f_{X}(0)$, it can be shown that 
	\begin{equation*}
		\text{Var}(\hat{\tau}_{\text{sharp}} - \widehat{\text{Bias}}(\hat{\tau}_{\text{sharp}})) = \frac{1}{nh}C_1 + o_p(\frac{1}{nh}),
	\end{equation*}
	where  the constant $C_1$ can be straightforwardly inferred from the proof of \Cref{thm:sharp t adjusted dist} in the Online Supplement.
\end{remark}

\begin{remark}
Instead of bias correction, one may consider under-smoothing ($h=o(n^{-1/5})$), so that the bias is negligible relative to the square root of the variance, which leads to the conventional 95\% confidence interval $\hat{\tau}_{\text{sharp}}  \pm 1.96 \sqrt{\text{Var}(\hat{\tau}_{\text{sharp}})}$. Yet \citeauthor{calonico2018jasa} (\citeyear{calonico2018jasa}, \citeyear{calonico2020optimalbandwidth}) show that robust bias correction can offer high-order refinements in the sense that its resulting confidence interval tends to have smaller coverage errors than that based on under-smoothing. Their proof is based on coverage error expansions for confidence intervals resulting from local polynomial estimation. It  would thus be value-added to this paper if one similarly extends the proof to  LCQR-based confidence intervals. 
\end{remark}
 
\subsection{The fuzzy case}

Similar to the sharp case above, we also propose the adjusted $t$-statistic for the fuzzy RD design. Following \cite{feir2016weak}, we use a null-restricted \textit{t}-statistic to help eliminate the size distortion due to possibly weak identification in the fuzzy RD. Use \cref{eq:fuzzy treatment} to rewrite the the null $H_0: \tau_{\text{fuzzy}} = \tau_0$ as
\begin{equation}
	H_0: [m_{Y_+}(0) - \tau_0 m_{T_+}(0)] - [m_{Y_-}(0) - \tau_0 m_{T_-}(0)]= 0.
\end{equation}

Define $\tilde{\tau}_{\text{fuzzy}} = (\hat{m}_{Y_+} - \tau_0 \hat{m}_{T_+}) - (\hat{m}_{Y_-} - \tau_0 \hat{m}_{T_-})$. The bias-corrected $\tilde{\tau}_{\text{fuzzy}}$ is given by
\begin{equation}
	\tilde{\tau}^{\text{bc}}_{\text{fuzzy}} = \tilde{\tau}_{\text{fuzzy}} - \widehat{\text{Bias}}(\tilde{\tau}_{\text{fuzzy}}),
\end{equation}
where
\begin{equation} \label{eq:fuzzy biases}
	\widehat{\text{Bias}}(\tilde{\tau}_{\text{fuzzy}}) = [\widehat{\text{Bias}}(\hat{m}_{Y_+}) - \tau_0 \widehat{\text{Bias}}(\hat{m}_{T_+})] - [\widehat{\text{Bias}}(\hat{m}_{Y_-}) - \tau_0 \widehat{\text{Bias}}(\hat{m}_{T_-})].
\end{equation}

We note that $\widehat{\text{Bias}}(\hat{m}_{Y_+})$ and $\widehat{\text{Bias}}(\hat{m}_{Y_-})$ are defined in \cref{eq:m+ bias hat}, while $\widehat{\text{Bias}}(\hat{m}_{T_+})$ and $\widehat{\text{Bias}}(\hat{m}_{T_-})$ are defined similarly by replacing $Y$ with $T$. All the bias terms can be estimated using the result in Lemma 2 in the Online Supplement. To illustrate the components in the adjusted variance, it is helpful to consider first the data above the cutoff (the expressions for the data below the cutoff result by replacing $+$ with $-$):
\begin{align} \label{eq:fuzzy adj var}
	&\text{Var}( (\hat{m}_{Y_+} - \tau_0 \hat{m}_{T_+}) -(\widehat{\text{Bias}}(\hat{m}_{Y_+}) - \tau_0 \widehat{\text{Bias}}(\hat{m}_{T_+})) )\nonumber\\
	&= \text{Var}(\hat{m}_{Y_+}) + \tau_0^2\text{Var}(\hat{m}_{T_+}) - 2 \tau_0 \text{Cov}(\hat{m}_{Y_+},\hat{m}_{T_+})\nonumber\\
	&\quad + \text{Var}(\widehat{\text{Bias}}(\hat{m}_{Y_+})) + \tau_0^2 \text{Var}(\widehat{\text{Bias}}(\hat{m}_{T_+})) - 2 \tau_0 \text{Cov}(\widehat{\text{Bias}}(\hat{m}_{Y_+}),\widehat{\text{Bias}}(\hat{m}_{T_+}))\nonumber\\
	&\quad -2 \text{Cov}(\hat{m}_{Y_+},\widehat{\text{Bias}}(\hat{m}_{Y_+})) + 2 \tau_0 \text{Cov}(\hat{m}_{Y_+},\widehat{\text{Bias}}(\hat{m}_{T_+}))\nonumber\\
	&\quad + 2 \tau_0 \text{Cov}(\hat{m}_{T_+},\widehat{\text{Bias}}(\hat{m}_{Y_+})) - 2 \tau_0^2 \text{Cov}(\hat{m}_{T_+},\widehat{\text{Bias}}(\hat{m}_{T_+})).
\end{align}

To operationalize the variance adjustment process, we derive all the required covariance terms in the Online Supplement. 

\setcounter{theorem}{3}
\begin{theorem} \label{thm:fuzzy t adjusted dist}
	Under \cref{assumption:m continuity,assumption:sigma continuity,assumption:kernel,assumption:density x,assumption:error distribution,assumption:bandwidth}, as both $n_+ \text{ and }n_- \rightarrow \infty$, the adjusted $t$-statistic for the fuzzy RD follows an asymptotic normal distribution
	\begin{equation} \label{eq:fuzzy t adjusted dist}
		t_{\text{fuzzy}}^{\text{adj.}} = \frac{\tilde{\tau}_{\text{fuzzy}} - \widehat{\text{Bias}}(\tilde{\tau}_{\text{fuzzy}})}{\sqrt{\text{Var}(\tilde{\tau}_{\text{fuzzy}} - \widehat{\text{Bias}}(\tilde{\tau}_{\text{fuzzy}}))}} \overset{d}{\to} N(0,1).
	\end{equation}
The exact expression for $\text{Var}(\tilde{\tau}_{\text{fuzzy}} - \widehat{\text{Bias}}(\tilde{\tau}_{\text{fuzzy}}))$ is provided  in the proof; see the proof of \Cref{thm:fuzzy t adjusted dist} in the Online Supplement Section S.2.
\end{theorem}

\subsection{A revised MSE-optimal bandwidth} \label{sec:bandwidth}

In this subsection, we propose a method to revise the MSE-optimal bandwidth by taking into consideration both bias correction and s.e. adjustment. It is well-known in the nonparametric literature that the MSE-optimal bandwidth, when used in inference, often induces undercoverage of conventional confidence intervals. The root cause is that the MSE-optimal bandwidth, $h_{\text{MSE}}=O(n^{-1/5})$, leads to $n h_{\text{MSE}}^5 \rightarrow \text{a constant}$,  while we need $nh^5 \rightarrow 0$ in order to ignore bias and use undersmoothing without bias correction. Hence using $h_{\text{MSE}}$ in \cref{eq:sharp t} could lead to poor coverage of confidence intervals. Using a bias-corrected \textit{t}-statistic without adjusting the variance, $\frac{\hat{\tau}_{\text{sharp}} - \widehat{\text{Bias}}(\hat{\tau}_{\text{sharp}})-\tau_0}{\sqrt{\text{Var}(\hat{\tau}_{\text{sharp}})}}$, does not solve this problem; see the discussion in Section 2 of \cite{calonico2014robust}.

It is helpful to revisit the bandwidth selection process in an MSE-optimal setting in order to find a solution. We use a single bandwidth $h$ for illustration. Consider the usual MSE when estimating a conditional mean function $m(\cdot)$ by $\hat{m}(\cdot)$:
\begin{equation}
	\text{MSE} = \text{Bias}^2(\hat{m}) + \text{Var}(\hat{m}) +  o_p(h^4 + \frac{1}{nh}).
\end{equation}
It is clear that the MSE-optimal bandwidth aims to balance two terms,  $\text{Bias}^2$ and $\text{Variance}$. After correcting the bias for the numerator of the \textit{t}-statistic, an optimal bandwidth should balance $\text{Bias}^2(\hat{m}-\widehat{\text{Bias}}(\hat{m}))$, not $\text{Bias}^2(\hat{m})$, with the variance term. Given $\text{Bias}(\hat{m}) = O(h^2)$, we need to expand the bias expression up to $O(h^3)$ in order to compute $\text{Bias}(\hat{m}-\widehat{\text{Bias}}(\hat{m}))$. In addition, after incorporating variance adjustment due to bias correction, the adjusted MSE takes the following form
 \begin{equation} \label{eq:adjmse}
 \text{adj. MSE} = \text{Bias}^2(\hat{m}-\widehat{\text{Bias}}(\hat{m}))+ \text{adj. Var}(\hat{m}) +  o_p(h^6 + \frac{1}{nh}),
 \end{equation}
 where $\text{adj. Var}(\hat{m})$ can take the form in \cref{eq:sharp var adjusted} in a sharp RD design. The following theorem presents the bandwidth result associated with the adjusted MSE in a sharp RD design,  assuming we work with data above the cutoff.
 
 \begin{theorem} \label{thm:adjusted bandwidth}
 	Under \cref{assumption:m continuity,assumption:sigma continuity,assumption:kernel,assumption:density x,assumption:error distribution} and $h_{Y_+} \rightarrow 0$ as $n_+ \rightarrow \infty$, the bandwidth that minimizes the adjusted MSE is given by
 	\begin{equation} \label{eq:bw adj.}
 		h_{Y_+} = \left(\frac{C_3}{6C_2^2}\right)^{1/7}n_+^{-1/7},
 	\end{equation}
 	where $C_2$ and $C_3$ are constants  provided  in the proof; see the Online Supplement Section S.2. 
 \end{theorem}
 
The result for the fuzzy design case can be obtained similarly. \Cref{thm:adjusted bandwidth} assumes that two separate bandwidths are used for the data above and below the cutoff. Similar result holds for $h_{Y_-}$. On each side of the cutoff, a single bandwidth is used in both estimation and bias correction. \Cref{thm:adjusted bandwidth} provides the expression of this bandwidth.

\setcounter{theorem}{3}

\begin{remark}
	The bandwidth in \cref{eq:bw adj.} is of order $O(n_+^{-1/7})$, while \Cref{assumption:bandwidth} requires $h_{Y_+} = o(n_+^{-1/7})$ to establish the asymptotic normality of \cref{eq:sharp t adjusted dist}. Notice that the numerator in \cref{eq:sharp t adjusted dist} corrects the bias up to $O_p(h_+^2)$ (and $O_p(h_-^2)$), and the remaining terms of the ratio can be roughly written as $O_p(\sqrt{n_+h_{Y_+}}) \times O_p(h_{Y_+}^3)$. If the  bias correction up to $O_p(h_+^2)$ is truly effective, we can: (1) treat it as if bias were completely removed and ignore any remaining $O_p(h_+^3)$ terms on the numerator of the bias-corrected \textit{t}-statistic in \cref{eq:sharp t adjusted dist} so that the assumption $n_+h_{Y_+}^7 \rightarrow 0$ becomes unnecessary; (2) use a slightly larger bandwidth such as the one in \cref{eq:bw adj.} to further reduce variance while relying on the bias correction to remove any extra bias due to the use of a larger bandwidth. Our simulation results indeed indicate that bias correction in LCQR is effective. This provides an explanation of why an $O(n_+^{-1/7})$ bandwidth could work well in our case.
\end{remark}

\begin{remark}
	If equal bandwidth on both sides of the cutoff is preferred, it can be derived in a manner similar to the optimal bandwidth choice in \cite{imbens2012optimal}. To see that, we further assume $h_{Y_+} = h_{Y_-} = h$ and let $C_{+,2}$, $C_{+,3}$, $C_{-,2}$, $C_{-,3}$ be the corresponding constants in \Cref{thm:adjusted bandwidth} for the data above and below the cutoff. The adjusted MSE becomes 
$\text{adj. MSE} = (C_{+,2} - C_{-,2})^2 h^6 + \frac{1}{nh}(C_{+,3} + C_{-,3}) + o_p(h^6 + \frac{1}{nh})$, and the optimal bandwidth is given by
\medskip
	\begin{equation} \label{eq:bw adj IK style}
		h = \left(\frac{C_{+,3} + C_{-,3}}{6(C_{+,2} - C_{-,2})^2}\right)^{1/7} n^{-1/7}.
	\end{equation}
This bandwidth is shown to have good finite sample performance in our simulation study. One caveat is that the performance of this bandwidth relies on a good estimate of the third derivative of the conditional mean function. If it is difficult to obtain an estimate of the third derivative, one can use a simple, global quintic polynomial for estimation. This is implemented with the option \texttt{ls.derivative = TRUE} in the \texttt{rdcqr} package.
\end{remark}

\begin{remark}
	The bandwidth considered in this section has order $O(n_+^{-1/7})$ or $O(n^{-1/7})$. To check the robustness (to bandwidth selection) of the adjusted \textit{t}-statistic, we also experiment with the rule-of-thumb bandwidth in Section 4.2 of \cite{fan1996local} in our simulation. This bandwidth is close to the Mean Integrated Squared Error optimal (MISE-optimal) bandwidth and has order $O(n_+^{-1/5})$ and $O(n_-^{-1/5})$ for the data above and below the cutoff. The results are reported in  the Online Supplement Table 5, and they are found very similar to those obtained from using an $O(n^{-1/7})$ bandwidth.
\end{remark}

\section{Extensions}

In this section, we extend the discussion to several related topics that are of either practical or theoretical importance. Within the framework of LLR, these topics are studied in \cite{calonico2018jasa,calonico2019covariates,calonico2020optimal, calonico2020optimalbandwidth}. Within the framework of LCQR, we develop the fixed-$n$ approximations for estimation and inference in both sharp and fuzzy RD, while we also provide a brief discussion on several other topics.

\subsection{Fixed-n approximations for small samples} \label{sec:fixed-n}

Results in previous sections are based on asymptotic approximations as $n \rightarrow \infty$, while \cite{calonico2018jasa} point out that fixed-$n$ approximations and the associated Studentization can help local polynomial regression retain the automatic boundary carpentry property in coverage errors. We therefore study the fixed-$n$ approximations of the asymptotic results introduced in early sections for LCQR. Essentially, we re-derive \Cref{thm:sharp t adjusted dist,thm:fuzzy t adjusted dist} by keeping $n$ and $h$ fixed. Details of the fixed-$n$ version of the \textit{t}-statistics in \cref{eq:sharp t adjusted dist,eq:fuzzy t adjusted dist} are given in Propositions 1 and 2 and the associated proofs in Section S.3 of the Online Supplement.

As an example, consider the fixed-$n$ variance of the estimator $\hat{m}_{Y_+}$ in Lemma 5:
\begin{equation}
	\text{Var}(\hat{m}_{Y_+}|\mathbf{X})_{\text{fixed-n}} =  \frac{1}{n_+h_{Y_+}q^2} e_q^T \left(S_{nY_+}^{-1} \Sigma_{nY_+} S_{nY_+}^{-1}\right)e_q+ o_p(\frac{1}{n_+h_{Y_+}}),
\end{equation}
where $e_q$ is a vector of ones, $S_{nY_+}$ and $\Sigma_{nY_+}$ are defined in Section S.1 Equation (A.7). Compared to the first-order asymptotic variance in Lemma 2, it is clear that the fixed-$n$ expression does not use any of the kernel moments such as $\mu_{+,1}$.  The $o_p$ term at the end is a result of approximating the non-differentiable objective function of LCQR with a differentiable, quadratic function; see Proof of Lemma 5 in Section S.3 for details.

Our simulation presented later on suggests that fixed-$n$ approximations improve the coverage of LCQR-based confidence intervals when the sample size is small, and work particularly well under heteroskedasticity. In Table \ref{tb:main coverage}, confidence intervals based on fixed-$n$ approximations are found  comparable to asymptotic confidence intervals, though fixed-$n$ results seem more conservative and give slightly higher  coverage. When the sample size is decreased, fixed-$n$ approximations give better coverage (see Table 6 in Section S.4.4 of the Online Supplement).  The \texttt{rdcqr}  package offers both asymptotic and fixed-$n$ approximations in the computation of bias, s.e., and adjusted s.e. 

\subsection{Alternative bandwidth choices} \label{sec:ce-optimal}
Although we discuss the bandwidth that results from minimizing the  adjusted MSE of the bias-corrected LCQR estimator, other bandwidth choices that satisfy \Cref{assumption:bandwidth} are also allowed. One possibility is the coverage error (CE)-optimal bandwidth suggested by \cite{calonico2020optimalbandwidth}, who propose to derive the optimal bandwidth by minimizing the coverage error of confidence intervals for the treatment effect.  \cite{calonico2020optimal} further show that this bandwidth choice can make bias-corrected confidence intervals attain the minimax bound on coverage errors.  It will thus be interesting to derive such a CE-optimal bandwidth for the LCQR-based confidence interval, which, however, would require a considerable amount of work beyond the scope of this paper. We leave it for future research. Our simulation evidence presented in Section 5, however, shows that the bandwidth in \Cref{thm:adjusted bandwidth} provides  good coverage probabilities  along with an accurate estimation of the treatment effect.

\subsection{Adding covariates} \label{sec:covariates}

Many RD applications come with covariates in addition to the treatment assignment variable $X$. Let $Z_{\text{cv}}$ denote the additional covariates. Depending on the nature of the data, adding covariates could help reduce bias and make treatment estimates more precise. One can simply add these covariates to LLR to adjust the treatment estimate; see \citeauthor{imbens2008regression} (\citeyear{imbens2008regression}) for an example. More recently, \cite{calonico2019covariates} gives an MSE expansion for the covariate-adjusted RD estimator, and provides the corresponding bandwidth selection, bias correction and s.e. adjustment results.  

Using a single bandwidth $h$, similar to Equations (1) and (2) in \cite{calonico2019covariates}, we can write the treatment equation without and with covariates as
\begin{align}
    \hat{\tau}: \quad &\hat{Y}_i = \hat{a} + T_i \hat{\tau} + X_{i} \hat{\beta}_- + T_i X_{i} \hat{\beta}_+, \label{eq:without covariates} \\
    \tilde{\tau}: \quad & \hat{Y}_i = \hat{a} + T_i \tilde{\tau} + X_{i} \hat{\beta}_- + T_i X_{i} \hat{\beta}_+ + Z_{cv,i}^T \hat{\beta}_{cv}, \label{eq:with covariates} 
\end{align}
where all estimates are LCQR estimates. \Cref{thm:sharp t adjusted dist,thm:fuzzy t adjusted dist} are derived based on \cref{eq:without covariates} while \cref{eq:with covariates} follows from the following LCQR objective function,
\begin{equation} \label{eq:lcqr obj with covariates}
	\sum_{k=1}^{q} \sum_{i=1}^{n}\rho_{\tau_k} \left(Y_{i} - a_k - b (X_{i} -  x)- \tau_{\text{sharp}} T_i - \beta T_i (X_{i} -  x) - Z_{cv,i}^T \beta_{cv}\right) K\left(\frac{X_{i} - x}{h}\right).
\end{equation}

Since \cref{eq:lcqr obj with covariates} leads to a nonlinear problem with no closed-form solution, we cannot use partitioned regression to directly express $\tilde{\tau}$ as a linear function of other parameters in the equation and study its asymptotics. Instead, one needs to retool the asymptotic methods in \cite{kai2010local_supp, kai2010local} to incorporate the covariates based on \cref{eq:lcqr obj with covariates}. As a result, all asymptotic results in previous sections need to be revised to reflect the presence of covariates. Developing the rigorous asymptotic results for \cref{eq:with covariates}, similar to \cite{calonico2019covariates}, will be a useful addition to make LCQR more appealing to applied researchers, and require future work.

For interested readers who want to try LCQR with covariates in RD, we offer an \textit{ad hoc} approach. The companion \texttt{rdcqr} package offers a function to estimate $\tilde{\tau}$ in \cref{eq:with covariates}, though no bias-correction or s.e. is currently provided. We suggest to proceed by using the bias and adjusted s.e. of $\hat{\tau}$ in \cref{eq:without covariates} for the $t$-statistic for $\tilde{\tau}$. This \textit{ad hoc} approach is not completely unwarranted. \citet[p.~626]{imbens2008regression} wrote: ``If the conditional distribution of $Z$ given $X$ is continuous at $x=c$, then including $Z$ in the regression will have little effect on the expected value of the estimator for $\tau$ ...'' and $Z$ refers to covariates. Hence we conjecture that, although $\tilde{\tau}$ and $\hat{\tau}$ are different, they could be numerically close to each other, and their variance differences could also be relatively small in practice. As a result, using the bias and adjusted s.e. of $\hat{\tau}$ for $\tilde{\tau}$ in the \textit{t}-test probably will not hurt much. Indeed, our simulation in Section S.4.6 of the Online Supplement shows that this \textit{ad hoc} approach works well for the data  generating process (DGP) used in \cite{calonico2019covariates}. When $\tilde{\tau}$ and $\hat{\tau}$ differ to a large extent, one cannot use the \textit{ad hoc} approach.

\subsection{Kink designs}

While the interest of sharp and fuzzy RD designs lies in the levels (i.e. conditional mean functions) at the cutoff, kink RD designs focus on the  derivatives of regression functions; see e.g. \cite{card2015inference}. \cite{kai2010local_supp, kai2010local} show that LCQR can also have efficiency gains when used for estimating derivatives. Therefore, it is natural to extend the LCQR method to kink RD designs. In  Table 7 of the Online Supplement, we report the simulation outcome that shows LCQR could outperform the local polynomial regression for estimating derivatives in a sharp kink design when data are non-normal. We, however, do not further explore kink RD designs in this paper, since they would involve higher-order terms in Taylor series expansions, and are more challenging in practical implementations.

\section{Monte Carlo simulation}

In this section, we conduct a Monte Carlo study to investigate the finite sample properties of the LCQR method. We focus on the sharp RD design and separate the Monte Carlo study into two parts, one for  estimation (Section \ref{sec:simu estimation}), and the other for inference (Section \ref{sec:simu inference}).

The sharp designs calibrated  to \citeauthor{lee2008randomized} (\citeyear{lee2008randomized}, Lee)  and \citeauthor{ludwig2007does} (\citeyear{ludwig2007does}, LM) are used in the DGP:
\begin{align*}
	Y_i &= m(X_i) + \sigma(X_i) \epsilon_i, \quad i = 1, \cdots, n,\\
	X_i &\sim 2 \times \text{Beta}(2,4) - 1,
\end{align*}
where the conditional means are given by
\begin{equation} \label{eq:lee}
	m_{\text{Lee}}(X_i) = \begin{cases}
	0.48 + 1.27 X_i + 7.18 X_i^2 + 20.21 X_i^3 + 21.54 X_i^4 + 7.33 X_i^5 & \text{ if } X_i < 0,\\
	0.52 + 0.84 X_i - 3.00 X_i^2 + 7.99 X_i^3 - 9.01 X_i^4 + 3.56 X_i^5 & \text{ if } X_i \geq 0,
	\end{cases}
\end{equation}
and 
\begin{equation} \label{eq:lm}
	m_{\text{LM}}(X_i) = \begin{cases}
	3.71 + 2.30 X_i + 3.28 X_i^2 + 1.45 X_i^3 + 0.23 X_i^4 + 0.03 X_i^5 & \text{ if } X_i < 0,\\
	0.26 + 18.49 X_i - 54.81 X_i^2 + 74.30 X_i^3 - 45.02 X_i^4 + 9.83 X_i^5& \text{ if } X_i \geq 0.
	\end{cases}
\end{equation}
Hence the treatment effects are $0.04$ and $-3.45$ in (\ref{eq:lee}) and (\ref{eq:lm}), respectively.

We use the same five error distributions in \cite{kai2010local} to simulate $\epsilon_i$. These five error distributions are  listed in Table \ref{tab:are} above, which lead to five DGPs: for DGP 1, $\epsilon_i \sim N(0,1)$;  ...; for DGP 5, $\epsilon_i \sim \text{mixture normal }0.95N(0, 1) + 0.05N(0,10^2)$.


In addition, the homoskedastic and heteroskedastic specifications in \cite{kai2010local} are used for simulating the standard deviation:

\begin{equation*}
	\sigma(X_i) = \begin{cases}
	0.5 & \text{ for homoskedastic errors,}\\
	2 + \cos(2\pi X_i)/10 & \text{ for heteroskedastic errors.}
	\end{cases}
\end{equation*}
We set $n = 500$ with 5000 replications. The data are \textit{i.i.d.} draws in all replications. The triangular kernel is used in all estimations.

Since we also study the coverage probability of confidence intervals, we compare the LCQR results with the robust confidence interval of \cite{calonico2014robust}. For ease of exposition, we first summarize in Table \ref{table_summary} the types of the estimators reported later on.
\begin{table}[htbp]
	\centering 
	\caption{Summary of estimators in the simulation study}
\begin{tabular}{llr}  
	\toprule
	Estimator    & Description \\
	\midrule
	$\hat{\tau}_{\text{1bw}}^{\text{cqr}}$ & \text{LCQR with equal bandwidth }\\
		$\hat{\tau}_{\text{2bw}}^{\text{cqr}}$ & \text{LCQR with unequal bandwidth }\\
			$\hat{\tau}_{\text{1bw}}^{\text{llr}}$ & \text{LLR estimator}\\
	$\hat{\tau}_{\text{1bw}}^{\text{cqr,bc}}$ & \text{LCQR with equal bandwidth, bias-corrected and s.e.-adjusted}\\

	$\hat{\tau}_{\text{2bw}}^{\text{cqr,bc}}$ & \text{LCQR with unequal bandwidth, bias-corrected and s.e.-adjusted}\\

	$\hat{\tau}_{\text{1bw}}^{\text{robust,bc}}$ & \text{the robust estimator in \cite{calonico2014robust} with equal bandwidth}\\
	\\[-1.5em]
	$\hat{\tau}_{\text{1bw, fixed-n}}^{\text{cqr,bc}}$ & \text{fixed-n LCQR with equal bandwidth, bias-corrected and s.e.-adjusted}\\	
	\bottomrule
\end{tabular}\label{table_summary}
\end{table}

We set $q = 7$ for all LCQR estimators. $\hat{\tau}_{\text{1bw}}^{\text{robust,bc}}$ is obtained by using the R package \texttt{rdrobust}. We use the option \texttt{bwselect = mserd} in estimation and the option \texttt{bwselect = cerrd} in inference so that $\hat{\tau}_{\text{1bw}}^{\text{robust,bc}}$ is MSE-optimal and CE-optimal in  \Cref{sec:simu estimation} and \Cref{sec:simu inference}, respectively.  The local linear estimator, $\hat{\tau}_{\text{1bw}}^{\text{llr}}$, is obtained by applying the main bandwidth used in $\hat{\tau}_{\text{1bw}}^{\text{robust,bc}}$. 

To calculate the bandwidth in \cref{eq:bw adj.} and \cref{eq:bw adj IK style} for LCQR, we use the rule-of-thumb bandwidth selector described in Section 4.2 of \cite{fan1996local} to compute quantities such as $C_2$ and $C_3$ in \cref{eq:bw adj.}. The bandwidth in \cref{eq:bw adj.} or \cref{eq:bw adj IK style} is then used to perform the LCQR estimation, bias correction, and s.e. adjustment. Unlike the LLR estimator, the LCQR estimator has no closed-form expression, and is obtained from the iterative MM algorithm. 


\subsection{Estimation of the treatment effect} \label{sec:simu estimation}

In this subsection, we compare LCQR  with  LLR for the treatment effect estimation. 

Without bias correction, \Cref{tb:lee and lm estimation} suggests that the two LCQR estimators, $\hat{\tau}_{\text{1bw}}^{\text{cqr}}$ and $\hat{\tau}_{\text{2bw}}^{\text{cqr}}$, are less accurate than LLR, though the numerical difference is small. This result probably is not a surprise since we use the bandwidth based on the adjusted MSE, which is not MSE-optimal for $\hat{\tau}_{\text{1bw}}^{\text{cqr}}$ and $\hat{\tau}_{\text{2bw}}^{\text{cqr}}$. Since the employed bandwidth balances the bias of the bias-corrected estimator and adjusted variance in \cref{eq:adjmse}, it will be interesting to investigate its bias-correction performance. \Cref{fig:lee bias} presents the bias of the bias-corrected LCQR and LLR estimators for estimating the treatment effect in \cref{eq:lee} with homoskedasticity. It shows that LCQR produces accurate estimates after bias-correction. The improvement can be large, depending on the error distribution; similar results for other models can be found in Figure 3 in the Online Supplement. It suggests that   bias correction will help center the confidence intervals.

\begin{figure}[htp]
	\centering
	\includegraphics[width=0.60\textwidth,keepaspectratio=TRUE]{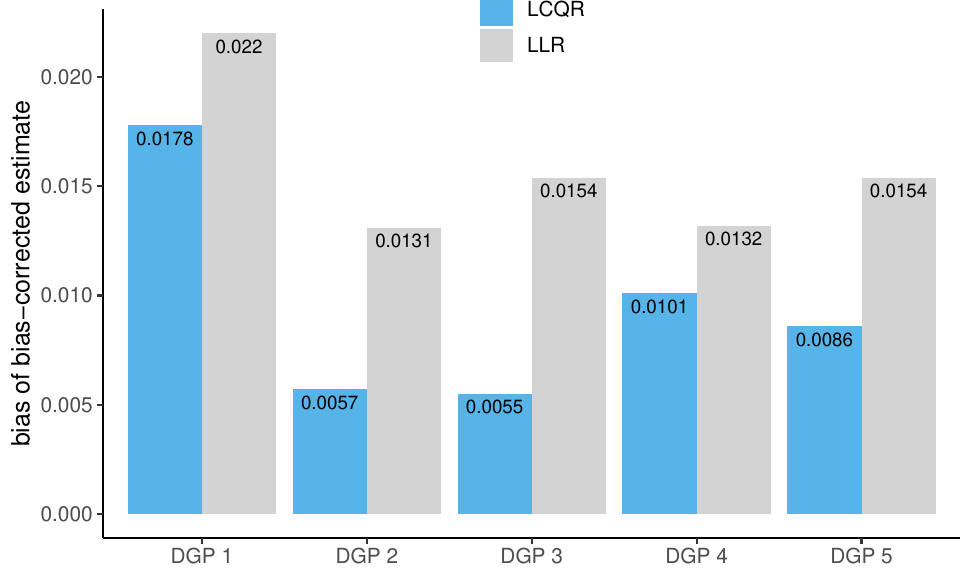}%
	\caption{Absolute value of average bias of the bias-corrected estimators, $\hat{\tau}_{\text{1bw}}^{\text{cqr,bc}}$ and $\hat{\tau}_{\text{1bw}}^{\text{robust,bc}}$ for the Lee model with homoskedasticity. $\hat{\tau}_{\text{1bw}}^{\text{cqr,bc}}$ is the bias-corrected LCQR estimator. $\hat{\tau}_{\text{1bw}}^{\text{robust,bc}}$ is the bias-corrected LLR estimator. The result is based on $5000$ replications and the true treatment effect is $0.04$.}%
	\label{fig:lee bias}%
\end{figure}

\begin{table}[htp] \centering
	\begin{center}
	\caption{Estimation of treatment effects in Lee and LM models} 
	\label{tb:lee and lm estimation} 
		\begin{threeparttable}
			\begin{tabular}{lcccccccccc} 
				\hline
		
				\multirow{2}{*}{Estimator}&\multicolumn{2}{c}{DGP 1}&\multicolumn{2}{c}{DGP 2}&\multicolumn{2}{c}{DGP 3}&\multicolumn{2}{c}{DGP 4}&\multicolumn{2}{c}{DGP 5}\\ \cmidrule(lr){2-3} \cmidrule(lr){4-5} \cmidrule(lr){6-7} \cmidrule(lr){8-9} \cmidrule(lr){10-11}  
				&coef.&s.e.&coef.&s.e.&coef.&s.e.&coef.&s.e.&coef.&s.e.\\\hline
					&			\multicolumn{10}{c}{A: Lee with homoskedastic errors and and $\tau_{\text{sharp}} = 0.04$} \\
					\cmidrule(lr){2-11}	
							$\hat{\tau}_{\text{1bw}}^{\text{cqr}}$ & 0.071 & 0.188 & 0.061 & 0.228 & 0.062 & 0.250 & 0.061 & 0.204 & 0.058 & 0.250 \\
										$\hat{\tau}_{\text{2bw}}^{\text{cqr}}$ & 0.073 & 0.186 & 0.062 & 0.227 & 0.062 & 0.248 & 0.064 & 0.202 & 0.060 & 0.247 \\
													$\hat{\tau}_{\text{1bw}}^{\text{llr}}$ & 0.068 & 0.203 & 0.060 & 0.284 & 0.060 & 0.332 & 0.060 & 0.237 & 0.062 & 0.443 \\
			$\hat{\tau}_{\text{1bw}}^{\text{cqr,bc}}$ & 0.058 & 0.347 & 0.046 & 0.411 & 0.046 & 0.454 & 0.050 & 0.375 & 0.049 & 0.453 \\

			$\hat{\tau}_{\text{2bw}}^{\text{cqr,bc}}$ & 0.059 & 0.344 & 0.046 & 0.408 & 0.048 & 0.450 & 0.053 & 0.371 & 0.054 & 0.448 \\

			$\hat{\tau}_{\text{1bw}}^{\text{robust,bc}}$ & 0.062 & 0.241 & 0.053 & 0.336 & 0.055 & 0.393 & 0.053 & 0.281 & 0.055 & 0.525 \\

				&\multicolumn{10}{c}{B: Lee  with heteroskedatic errors and and $\tau_{\text{sharp}} = 0.04$} \\
					\cmidrule(lr){2-11}			
							
				$\hat{\tau}_{\text{1bw}}^{\text{cqr}}$ & 0.071 & 0.100 & 0.064 & 0.123 & 0.064 & 0.134 & 0.065 & 0.109 & 0.060 & 0.137 \\
								$\hat{\tau}_{\text{2bw}}^{\text{cqr}}$ & 0.072 & 0.099 & 0.064 & 0.122 & 0.063 & 0.133 & 0.066 & 0.108 & 0.062 & 0.133 \\
												$\hat{\tau}_{\text{1bw}}^{\text{llr}}$ & 0.067 & 0.116 & 0.062 & 0.163 & 0.061 & 0.190 & 0.062 & 0.136 & 0.063 & 0.255 \\
				$\hat{\tau}_{\text{1bw}}^{\text{cqr,bc}}$ & 0.056 & 0.184 & 0.049 & 0.220 & 0.048 & 0.243 & 0.050 & 0.201 & 0.049 & 0.247 \\

				$\hat{\tau}_{\text{2bw}}^{\text{cqr,bc}}$ & 0.057 & 0.183 & 0.048 & 0.219 & 0.049 & 0.242 & 0.051 & 0.199 & 0.050 & 0.241 \\

				$\hat{\tau}_{\text{1bw}}^{\text{robust,bc}}$ & 0.061 & 0.138 & 0.055 & 0.193 & 0.056 & 0.226 & 0.056 & 0.161 & 0.057 & 0.303 \\
				
				 	&\multicolumn{10}{c}{C: LM  with homoskedasticity errors and $\tau_{\text{sharp}} = -3.45$} \\
			\cmidrule(lr){2-11}
				$\hat{\tau}_{\text{1bw}}^{\text{cqr}}$ & -3.244 & 0.191 & -3.265 & 0.231 & -3.269 & 0.252 & -3.257 & 0.206 & -3.296 & 0.251 \\
								$\hat{\tau}_{\text{2bw}}^{\text{cqr}}$ & -3.259 & 0.189 & -3.280 & 0.228 & -3.285 & 0.249 & -3.270 & 0.204 & -3.308 & 0.247 \\
							$\hat{\tau}_{\text{1bw}}^{\text{llr}}$ & -3.340 & 0.223 & -3.342 & 0.301 & -3.333 & 0.346 & -3.345 & 0.256 & -3.328 & 0.452 \\	
				$\hat{\tau}_{\text{1bw}}^{\text{cqr,bc}}$ & -3.439 & 0.353 & -3.453 & 0.415 & -3.450 & 0.458 & -3.445 & 0.380 & -3.443 & 0.455 \\

				$\hat{\tau}_{\text{2bw}}^{\text{cqr,bc}}$ & -3.440 & 0.348 & -3.452 & 0.411 & -3.450 & 0.452 & -3.443 & 0.375 & -3.440 & 0.447 \\

				$\hat{\tau}_{\text{1bw}}^{\text{robust,bc}}$ & -3.423 & 0.256 & -3.434 & 0.349 & -3.426 & 0.403 & -3.432 & 0.295 & -3.426 & 0.530 \\
			
				&\multicolumn{10}{c}{D: LM  with heteroskedatic errors and $\tau_{\text{sharp}} = -3.45$} \\
						\cmidrule(lr){2-11}
				$\hat{\tau}_{\text{1bw}}^{\text{cqr}}$ & -3.234 & 0.106 & -3.245 & 0.127 & -3.252 & 0.138 & -3.242 & 0.114 & -3.281 & 0.137 \\
								$\hat{\tau}_{\text{2bw}}^{\text{cqr}}$ & -3.239 & 0.104 & -3.256 & 0.125 & -3.259 & 0.136 & -3.247 & 0.112 & -3.288 & 0.134 \\
												$\hat{\tau}_{\text{1bw}}^{\text{llr}}$ & -3.364 & 0.146 & -3.359 & 0.190 & -3.349 & 0.215 & -3.364 & 0.165 & -3.339 & 0.273 \\
				$\hat{\tau}_{\text{1bw}}^{\text{cqr,bc}}$ & -3.443 & 0.195 & -3.448 & 0.228 & -3.444 & 0.250 & -3.442 & 0.209 & -3.443 & 0.248 \\

				$\hat{\tau}_{\text{2bw}}^{\text{cqr,bc}}$ & -3.443 & 0.192 & -3.447 & 0.225 & -3.446 & 0.246 & -3.444 & 0.206 & -3.445 & 0.243 \\

				$\hat{\tau}_{\text{1bw}}^{\text{robust,bc}}$ & -3.431 & 0.163 & -3.438 & 0.216 & -3.432 & 0.245 & -3.437 & 0.185 & -3.431 & 0.315 \\
				
				\hline 
			\end{tabular}
			\begin{tablenotes}[flushleft]
				\item[] \textit{Notes}: The estimates for the treatment effect and standard error are averages over 5000 replications with a sample size $n = 500$. The s.e. and adjusted s.e. for LCQR are obtained using the asymptotic expressions from \Cref{thm:sharp results} and \Cref{thm:sharp t adjusted dist}. Estimators with a superscript \texttt{bc} are both bias-corrected and s.e.-adjusted. The subscripts 1bw and 2bw refer to the use of equal and unequal bandwidths below and above the cutoff.
			\end{tablenotes}
		\end{threeparttable}
	\end{center} 
\end{table}

 Table \ref{tb:lee and lm estimation} further presents the standard errors  of the studied estimators to facilitate comparison.  The s.e. of $\hat{\tau}_{\text{1bw}}^{\text{cqr}}$ is consistently smaller than that of $\hat{\tau}_{\text{1bw}}^{\text{llr}}$, indicating the efficiency gain of LCQR against LLR. Although DGP 1 has normal errors, $\hat{\tau}_{\text{1bw}}^{\text{cqr}}$ and $\hat{\tau}_{\text{1bw}}^{\text{llr}}$ do not use the same bandwidth, so the s.e. of $\hat{\tau}_{\text{1bw}}^{\text{cqr}}$  also appears smaller than that of $\hat{\tau}_{\text{1bw}}^{\text{llr}}$.   As the DGP moves away from normal errors, the LCQR estimator achieves various levels of efficiency gains compared to the LLR estimator. For example, the s.e. of $\hat{\tau}_{\text{1bw}}^{\text{cqr}}$ in DGP 5 of  Table \ref{tb:lee and lm estimation}  Panel B is $0.137$, compared to $0.255$ for $\hat{\tau}_{\text{1bw}}^{\text{llr}}$, so the LCQR/LLR standard error ratio is close to 50\% in this example.

While reading  Table \ref{tb:lee and lm estimation}, it is important to bear in mind that the s.e.s of estimators with a superscript ``\texttt{bc}" are not suitable to assess the efficiency of the estimators. They are computed by incorporating the additional variability due to bias correction. However, they shed light on the length of confidence intervals.   Table \ref{tb:lee and lm estimation}  shows that  the adjusted s.e. of $\hat{\tau}_{\text{1bw}}^{\text{cqr,bc}}$ could be smaller or larger than that of $\hat{\tau}_{\text{1bw}}^{\text{robust,bc}}$, so the simulation results for comparing LCQR with LLR appear to be mixed after bias correction. 

In addition, using two bandwidths above and below the cutoff also gives mixed results in terms of bias correction for the LCQR estimator, though it is clear that the two-bandwidth approach leads to a small decrease in s.e. for LCQR.

\subsection{Inference on the treatment effect} \label{sec:simu inference}

This subsection studies the coverage probability of the confidence intervals based on the LCQR estimator, using the adjusted s.e. in \Cref{thm:sharp t adjusted dist} and the adj. MSE-based bandwidth selector. The results are summarized in Table \ref{tb:main coverage}, where the nominal coverage probability is set to $95\%$.

\begin{table}[htp] \centering 
	\begin{center}
	\caption{Coverage probability of 95\% confidence intervals in  Lee and LM models } 
	\label{tb:main coverage} 
		\begin{threeparttable}
			\begin{tabular}{lccccc  ccccc} 
				\hline
				& \multicolumn{5}{c}{A. Lee with homoskedastic errors} &  \multicolumn{5}{c}{B. Lee  with heteroskedatic errors}\\
				\cmidrule(lr){2-6}\cmidrule(lr){7-11}
			\footnotesize	Estimator& \footnotesize DGP 1&\footnotesize DGP 2&\footnotesize DGP 3&\footnotesize DGP 4&\footnotesize DGP 5&\footnotesize  DGP 1&\footnotesize DGP 2&\footnotesize DGP 3&\footnotesize DGP 4&\footnotesize DGP 5\\ 	
				\hline
			$\hat{\tau}_{\text{1bw}}^{\text{cqr}}$ & 0.922 & 0.925 & 0.918 & 0.921 & 0.910 & 0.893 & 0.898 & 0.895 & 0.894 & 0.895 \\
			    $\hat{\tau}_{\text{2bw}}^{\text{cqr}}$ & 0.919 & 0.916 & 0.912 & 0.916 & 0.904 & 0.883 & 0.892 & 0.882 & 0.890 & 0.893 \\
			        $\hat{\tau}_{\text{1bw}}^{\text{llr}}$ & 0.934 & 0.935 & 0.939 & 0.932 & 0.959 &  0.927 & 0.932 & 0.936 & 0.929 & 0.956 \\
    $\hat{\tau}_{\text{1bw}}^{\text{cqr,bc}}$ & 0.973 & 0.961 & 0.964 & 0.970 & 0.955 &0.959 & 0.952 & 0.954 & 0.956 & 0.948 \\

    $\hat{\tau}_{\text{2bw}}^{\text{cqr,bc}}$ & 0.971 & 0.957 & 0.963 & 0.969 & 0.956 &0.958 & 0.947 & 0.950 & 0.956 & 0.947  \\

    $\hat{\tau}_{\text{1bw}}^{\text{robust,bc}}$ & 0.936 & 0.939 & 0.941 & 0.934 & 0.959 & 0.932 & 0.935 & 0.937 & 0.932 & 0.957\\
    $\hat{\tau}_{\text{1bw,fixed-n}}^{\text{cqr,bc}}$ & 0.980 & 0.971 & 0.975 & 0.976 & 0.969  & 0.961 & 0.957 & 0.958 & 0.959 & 0.958 \\
    &\\
	& \multicolumn{5}{c}{C. LM  with homoskedastic errors} &  \multicolumn{5}{c}{D. LM  with heteroskedatic errors}\\
				\cmidrule(lr){2-6}\cmidrule(lr){7-11}
		\footnotesize	Estimator& \footnotesize DGP 1&\footnotesize DGP 2&\footnotesize DGP 3&\footnotesize DGP 4&\footnotesize DGP 5&\footnotesize  DGP 1&\footnotesize DGP 2&\footnotesize DGP 3&\footnotesize DGP 4&\footnotesize DGP 5\\ 	
				\hline
				$\hat{\tau}_{\text{1bw}}^{\text{cqr}}$ & 0.743 & 0.800 & 0.817 & 0.771 & 0.822& 0.526 & 0.604 & 0.642 & 0.564 & 0.670 \\
				    $\hat{\tau}_{\text{2bw}}^{\text{cqr}}$ & 0.766 & 0.813 & 0.833 & 0.797 & 0.836 & 0.505 & 0.609 & 0.649 & 0.555 & 0.679 \\
				        $\hat{\tau}_{\text{1bw}}^{\text{llr}}$ & 0.910 & 0.915 & 0.925 & 0.913 & 0.943 & 0.900 & 0.906 & 0.911 & 0.905 & 0.930 \\
    $\hat{\tau}_{\text{1bw}}^{\text{cqr,bc}}$ & 0.969 & 0.960 & 0.963 & 0.967 & 0.951  & 0.955 & 0.951 & 0.952 & 0.953 & 0.946\\

    $\hat{\tau}_{\text{2bw}}^{\text{cqr,bc}}$ & 0.969 & 0.960 & 0.960 & 0.968 & 0.953& 0.954 & 0.945 & 0.945 & 0.950 & 0.939 \\

    $\hat{\tau}_{\text{1bw}}^{\text{robust, bc}}$ & 0.931 & 0.932 & 0.935 & 0.929 & 0.954 & 0.928 & 0.929 & 0.937 & 0.931 & 0.951 \\
    $\hat{\tau}_{\text{1bw,fixed-n}}^{\text{cqr,bc}}$ & 0.976 & 0.970 & 0.973 & 0.974 & 0.967 & 0.957 & 0.953 & 0.956 & 0.959 & 0.958 \\
				\hline 
			\end{tabular}
			\begin{tablenotes}[flushleft]
				\item[] \textit{Notes}: The reported numbers are the simulated coverage probabilities of the 95\% confidence intervals associated with different estimators. The results are based on 5000 replications with a sample size $n = 500$. The s.e. and adjusted s.e. for the LCQR estimator are obtained based on the asymptotic expressions from \Cref{thm:sharp results} and \Cref{thm:sharp t adjusted dist}, except for $\hat{\tau}_{\text{1bw,fixed-n}}^{\text{cqr,bc}}$ where fixed-$n$ approximations are used. Estimators with superscript \texttt{bc} are both bias-corrected and s.e.-adjusted. The result of $\hat{\tau}_{\text{1bw}}^{\text{robust,bc}}$ is based on the CE-optimal bandwidth.
			\end{tablenotes}
		\end{threeparttable}
	\end{center} 
\end{table}

Without bias correction or s.e. adjustment, confidence intervals based on the LCQR estimators,	$\hat{\tau}_{\text{1bw}}^{\text{cqr}}$ and 	$\hat{\tau}_{\text{2bw}}^{\text{cqr}}$, are found to have poor  coverage probabilities in  Table \ref{tb:main coverage}, which highlights an important difference between estimation and inference: despite the excellent finite sample properties of LCQR in Table \ref{tb:lee and lm estimation}, one has to perform both bias correction and s.e. adjustment in order to produce the desired  coverage of confidence intervals. Furthermore, the coverage  probability of LCQR-based confidence intervals is lower than that of LLR-based confidence intervals. For example, it is $92.2\%$ under DGP 1 in Panel A of Table \ref{tb:main coverage}, compared to $93.4\%$ of the LLR-based confidence interval. Although $\hat{\tau}_{\text{1bw}}^{\text{cqr}}$ tends to have a smaller bias and a smaller s.e. than $\hat{\tau}_{\text{1bw}}^{\text{llr}}$, it appears that the bias in $\hat{\tau}_{\text{1bw}}^{\text{cqr}}$ is still not small enough to center the estimator close to the true value. A smaller s.e. will further contribute to the decrease in coverage probabilities. This explains the poor coverage for confidence intervals based on $\hat{\tau}_{\text{1bw}}^{\text{cqr}}$ and $\hat{\tau}_{\text{2bw}}^{\text{cqr}}$.

With bias correction and s.e. adjustment, the proposed bandwidth selectors in  \cref{eq:bw adj.} and \cref{eq:bw adj IK style}  lead to good coverage across the five DGPs in Table \ref{tb:main coverage}. It is important to recall that the presented LCQR estimators  use the bandwidths in  \cref{eq:bw adj.} and \cref{eq:bw adj IK style}, which are not designed to optimize the coverage probability  of  confidence intervals, while the confidence intervals for $\hat{\tau}_{\text{1bw}}^{\text{robust,bc}}$ in Table \ref{tb:main coverage} are optimized to minimize coverage errors. In this respect, it is reasonable to conclude that LCQR with bandwidths in   \cref{eq:bw adj.} and \cref{eq:bw adj IK style} offers very competitive results.

In the last row of each panel in \Cref{tb:main coverage}, we also report the coverage probability  based on fixed-$n$ approximations. The benefit of using fixed-$n$ approximations becomes clearer in Table 6 of the Online Supplement, where the sample size decreases to 300 and fixed-$n$ approximations often give the best coverage.

\section{Application to Lee (2008)}

In this section, we use the data from \cite{lee2008randomized} to illustrate the practical usage of  LCQR. To facilitate comparison, the findings based on LLR are also presented.

\begin{figure}[htp]
\caption{Comparison of LCQR and LLR using Lee (2008)}
\label{app_lee}\centering
\subfloat[Regression Discontinuity Design]{\includegraphics[width=220 pt,
height=160 pt]{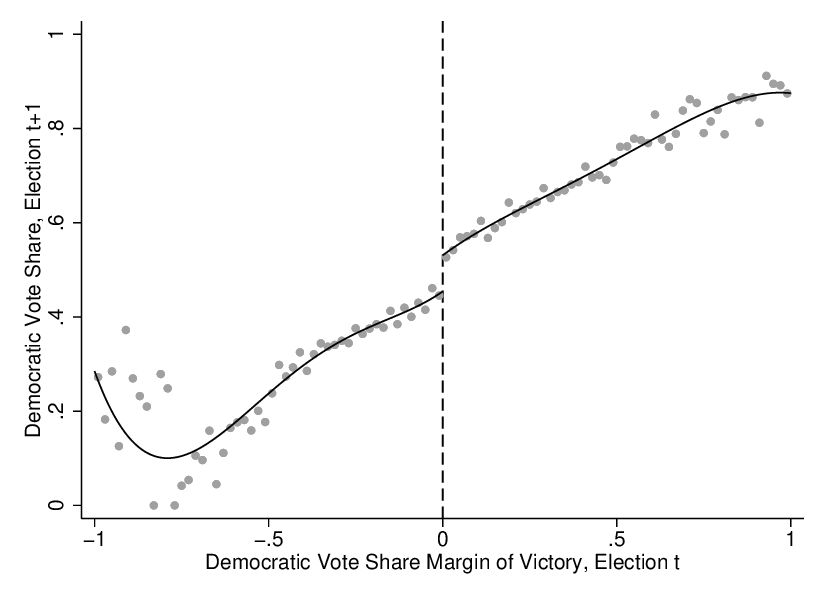}}
\subfloat[LCQR Estimate $\pm$ 2 S.E.]{\includegraphics[width=220 pt, height=160
pt]{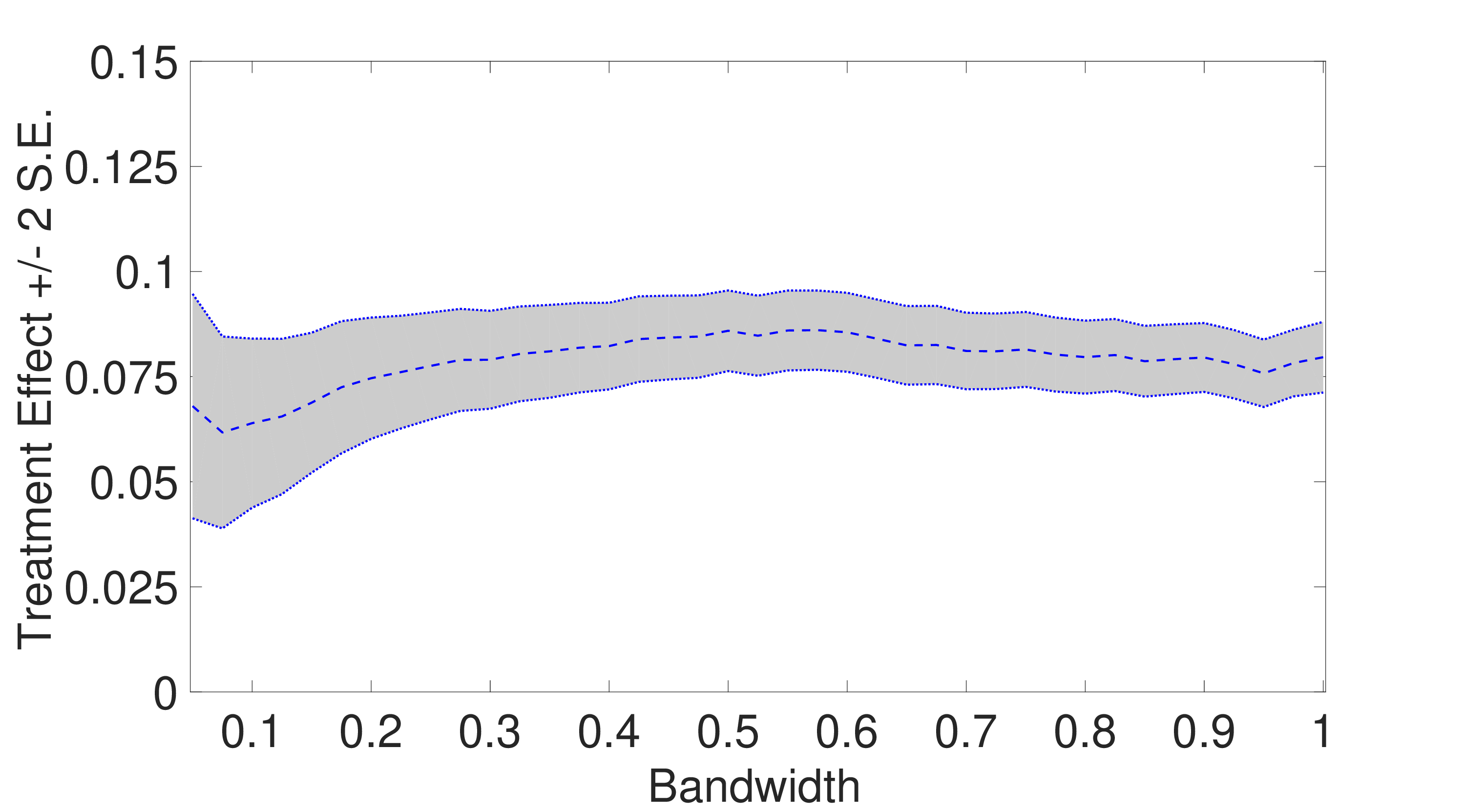}}\\
\subfloat[LLR Estimate $\pm$ 2 S.E.]{\includegraphics[width=220 pt,
height=160 pt]{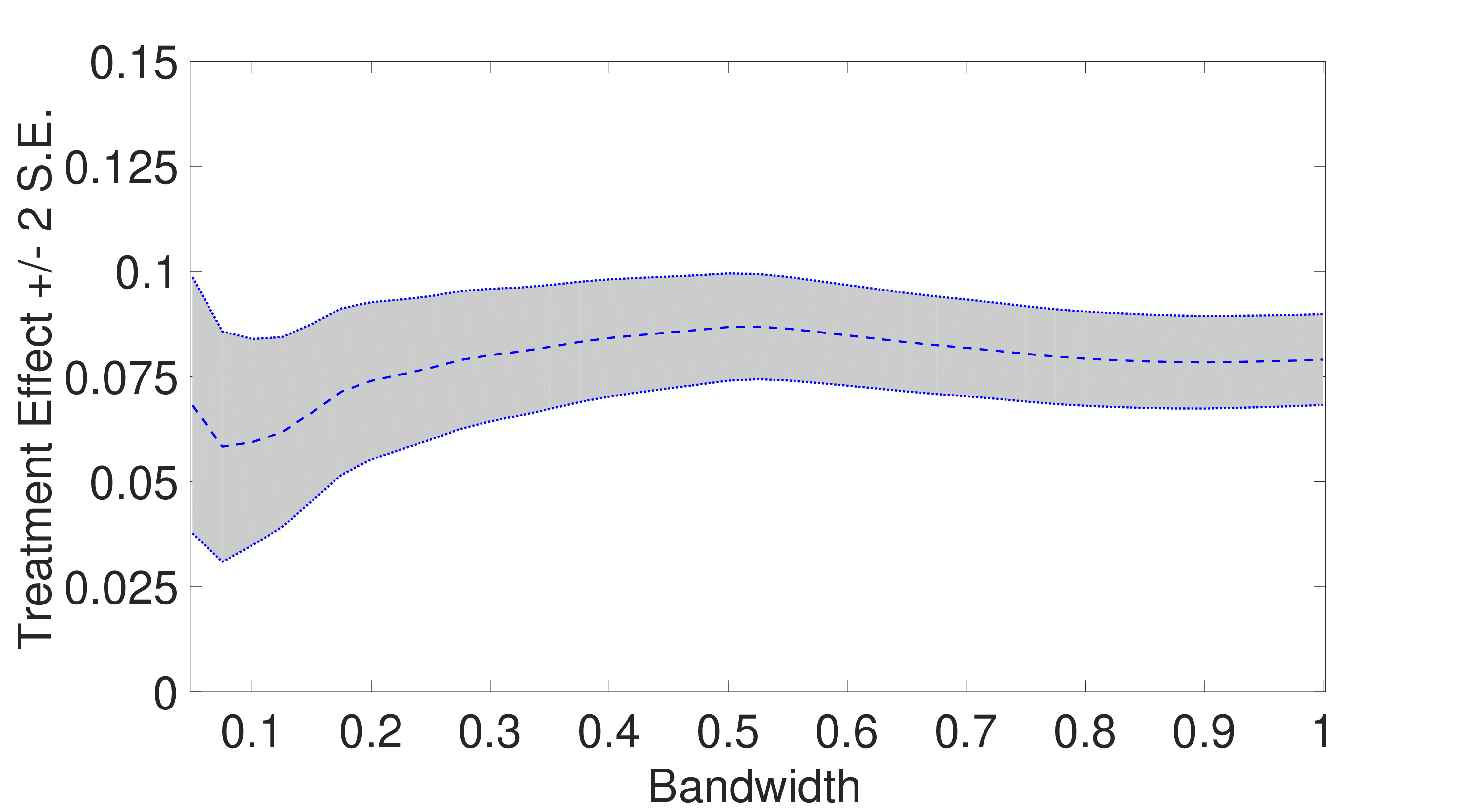}}
\subfloat[Ratio of S.E., LCQR/LLR]{\includegraphics[width=220 pt, height=160
pt]{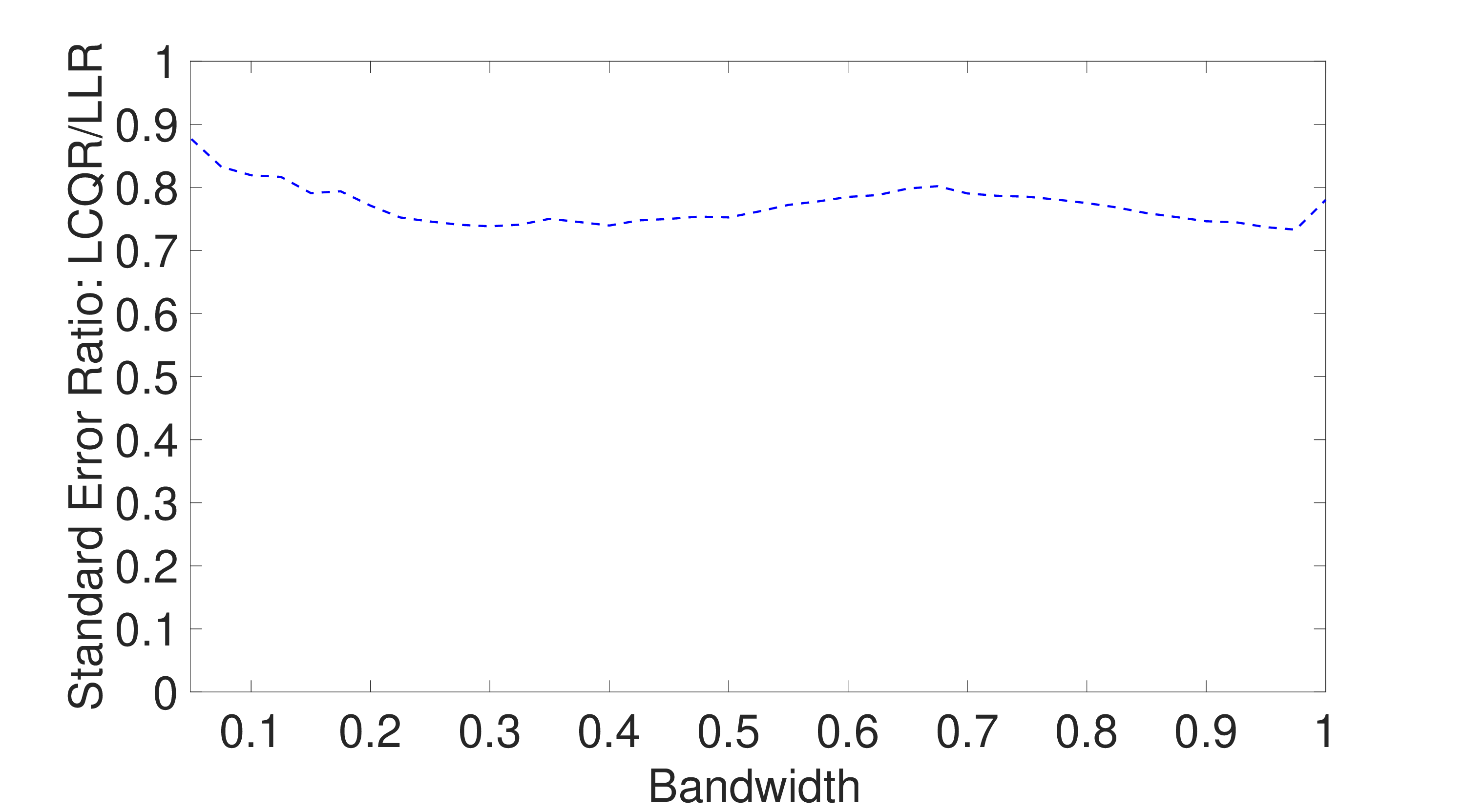}}
\par
\justify
{\textit{Notes}:  (a) x-axis, the democratic vote share (margin of victory) in Election $t$; y-axis, the democratic vote share  in Election $t+1$.  The dots show the sample mean of the vertical variable  in each bin of the horizontal variable (50 bins on each side of the cutoff). The solid line represents the fitted fourth-order polynomial.  As the bandwidth increases, (b) presents the LCQR estimate $\pm$ 2 $\times$ standard error; (c) presents the LLR estimate $\pm$ 2 $\times$ standard error; and (d) presents the ratio of the standard errors by LCQR  and LLR.  The MSE-optimal bandwidth of \cite{imbens2012optimal} for the studied data set  is about 0.3. The triangular kernel is used for both LCQR and LLR estimators.}
\end{figure}

We revisit a classic  example from \cite{lee2008randomized} with  $6,558$ observations depicted in Figure \ref{app_lee}(a); see also \cite{imbens2012optimal}. The horizontal running variable is the democratic vote share (margin of victory) in a previous election, while the vertical outcome variable is the democratic vote share  in the election afterwards.  Consistent with \cite{lee2008randomized} and several follow-up studies, Figure \ref{app_lee}(a) indicates that there is a positive impact of incumbency on re-election, i.e., a visible jump occurs at the threshold zero.

Figure \ref{app_lee}(b) and (c) present the estimated impact  of incumbency on re-election by LCQR and LLR  methods, respectively. We consider a sequence of bandwidth values ranging from 0.05 to  1 with the step size 0.025: 0.05, 0.075, 0.1, ..., 1. This bandwidth sequence thus nests many common  choices such as the MSE-optimal bandwidth of \cite{imbens2012optimal}, which is about 0.3 for the studied data. 

The comparison of Figure \ref{app_lee}(b) and (c) shows that LCQR and LLR yield similar point estimates over a wide range of bandwidths for the studied  \cite{lee2008randomized} application. As the bandwidth increases, Figure \ref{app_lee}(a) indicates that there are  data points staying further away from the fitted regression line. These data points affect  LCQR and LLR  estimators in a different manner, leading to slightly disparate point estimates. Nevertheless, all the point estimates depicted in Figure \ref{app_lee}(b) and (c) are  significantly positive, since the (vertical) zero value is excluded from the shaded regions generated by  $\pm$ 2 standard errors.

Most importantly, Figure \ref{app_lee}(d)  highlights that the standard error of the LCQR estimator is substantially smaller than that of LLR. The standard error ratio is mostly around 70\% $\sim$ 80\%, which can also be viewed by comparing the shaded regions in Figure \ref{app_lee}(b) and (c). Moreover, this standard error ratio does not change much as  the bandwidth varies. Thus, Figure \ref{app_lee}(d) indicates that a confidence interval for the impact  of incumbency on re-election based on the LCQR estimator could be considerably tighter than that by LLR. 

Consider, for example, 0.3 as the adopted bandwidth. The LCQR estimate $\pm$ 1.96 standard error leads to the 95\% confidence interval (0.068,   0.090), while the conventional 95\% confidence interval by LLR is  (0.065, 0.096). If bias correction is further accounted for at the adopted bandwidth 0.3, then the bias-corrected 95\% confidence interval based on LCQR  is   (0.048,   0.090).   This interval is comparable to the bias-corrected 95\% confidence interval  using the \cite{calonico2014robust} approach, which is  (0.046,   0.089). These empirical findings therefore lend credibility to our proposed LCQR approach.

\section{Conclusions}

In this paper, we study the application of LCQR in \cite{kai2010local} to the estimation and inference in RD. We present numerical evidence for the efficiency gain of using LCQR in RD estimation, and also propose a bias-corrected and s.e.-adjusted \textit{t}-statistic to improve the coverage of confidence intervals. Simulation results show good performance of the proposed method under several non-normal error distributions. 

The current work can be extended in several directions. For instance, throughout the paper, we focus on the local linear composite quantile regression, while the general case, local $p$-th order polynomial composite quantile regression, can be similarly adopted.   \cite{kai2010local_supp} establish the asymptotic theory for the LCQR estimator in the general case. Following  \cite{kai2010local_supp}, one could extend our results on the bias-corrected and s.e.-adjusted \textit{t}-statistic to allow for higher-order polynomials. In addition, it will  be naturally appealing to formally explore LCQR for kink RD designs as well as RD designs using covariates. It will also be interesting to revisit some of the existing applications in RD with the proposed method as data may deviate from normality. Finally, on the computation side, instead of using the same bandwidth for estimation, bias correction and s.e. adjustment, one can refine the bandwidth selection process, which may also improve the estimation and coverage probability of LCQR. We leave these topics for future research.

\section*{Acknowledgments}

We thank the Editor, the Associate Editor, and two anonymous referees
for their comments that substantially improved the paper. The authors gratefully acknowledge  the financial support of the Education Economics Center and the computational support of the Office of Research at Kennesaw State University, and have no conflicts of interest to disclose.

\newpage
\spacing{1.45}
\bibliographystyle{ecca}
\bibliography{reference}

\newpage
\setcounter{page}{1}
\spacing{1.42}

\begin{appendices}
{\centering \title{\large\MakeUppercase{Supplementary Material to ``Local Composite Quantile Regression for Regression Discontinuity"}\footnote{Email: xhuang3@kennesaw.edu and zzhan@kennesaw.edu.}\\[10pt]} }
	\begin{center}
		\large
		\author{\textsc{Xiao Huang}  \ \ \ \  \ \ \  \textsc{Zhaoguo Zhan}}\\
		\date{\today}
	\end{center}
	\maketitle
    
\bigskip    
    
This supplement contains all technical details,  lemmas and proofs for the asymptotic and fixed-$n$ results, as well as additional figures and tables.

\setstretch{2}
\bigskip

\localtableofcontents

\newpage 
\setstretch{1.35}
\setcounter{section}{19}
\setcounter{equation}{0}
\renewcommand{\theequation}{A.\arabic{equation}}
\subsection{Notation} \label{supp:notation}    

In a sharp RD, LCQR is applied separately to \cref{eq:Y plus model,eq:Y minus model}, while in a fuzzy RD, LCQR is applied separately to \cref{eq:Y plus model,eq:Y minus model,eq:T plus model,eq:T minus model}. Instead of introducing four similar sets of variables, notation and proofs, we will focus on \cref{eq:Y plus model}. Exactly the same proof holds for the results based on \cref{eq:Y minus model,eq:T plus model,eq:T minus model} with similar notation, where the subscript $Y_+$ used for \cref{eq:Y plus model} becomes $Y_-$, $T_+$, $T_-$ for \cref{eq:Y minus model,eq:T plus model,eq:T minus model}, respectively.


Consider \cref{eq:Y plus model}.  Let $f_{\epsilon_{Y_+}} = (f_{\epsilon_{Y_+}}(c_1), \cdots,f_{\epsilon_{Y_+}}(c_q))^T$ be a $q \times 1$ vector, $S_{Y_{+},11}(c)$ be a $q \times q$ diagonal matrix with diagonal elements $f_{\epsilon_{Y_+}}(c_k) \mu_{+,0}$, $S_{Y_{+},12}(c)$ be a $q \times p$ matrix with $(k,j)$ element $f_{\epsilon_{Y_+}}(c_k)\mu_{+,j}$, $S_{Y_{+},21}(c)$ be the transpose of $S_{Y_{+},12}(c)$, and $S_{Y_{+},22}(c)$ be a $p \times p$ matrix with $(j,j')$ element equal to  $\sum_{k=1}^{q}f_{\epsilon_{Y_+}}(c_k) \mu_{+,j+j'}(c)$. Let $\Sigma_{Y_{+},11}(c)$ be a $q \times q$ matrix with $(k,k')$ element $\nu_{+,0}(c)\tau_{kk'}$, $\Sigma_{Y_{+},12}(c)$ be a $q \times p$ matrix with $(k,j)$ element $\sum_{k'=1}^{q} \tau_{kk'}\nu_{+,j}$, $\Sigma_{Y_{+},21}(c)$ be the transpose of $\Sigma_{Y_{+},12}(c)$, $\Sigma_{Y_{+},22}(c)$ be a $p \times p$ matrix with $(j,j')$ element equal to $\sum_{k=1}^{q}\sum_{k'=1}^{q} \tau_{kk'} \nu_{+,j+j'}(c)$. Define
\begin{equation}
	S_{Y_+}(c) = 
	\begin{pmatrix}
	S_{Y_{+},11}(c) & S_{Y_{+},12}(c)\\
	S_{Y_{+},21}(c) & S_{Y_{+},22}(c)
	\end{pmatrix},
	\quad
	\Sigma_{Y_+}(c) =
	\begin{pmatrix}
	\Sigma_{Y_{+},11}(c) & \Sigma_{Y_{+},12}(c) \\
	\Sigma_{Y_{+},21}(c) & \Sigma_{Y_{+},22}(c)
	\end{pmatrix}, \label{notation of Thm1}
\end{equation}
and the partitioned inverse  $S_{Y_+}^{-1}(c)$:
\begin{equation}
	S_{Y_+}^{-1}(c) = 
	\begin{pmatrix}
	(S_{Y_+}^{-1}(c))_{11} & (S_{Y_+}^{-1}(c))_{12} \\
	(S_{Y_+}^{-1}(c))_{21} & (S_{Y_+}^{-1}(c))_{22}
	\end{pmatrix}.\label{notation of s inverse}
\end{equation}
Let $F_{+}(c_k,c_{k'})$ be the joint cumulative distribution function of $\epsilon_{Y_+}$ and $\epsilon_{T_+}$ at $(c_k,c_{k'})$ and assume $h_{Y_+}=h_{T_+}$. Define $\phi_{kk'} = F_+(c_k,c_{k'}) - \tau_k \tau_{k'}$. Also let $\Sigma_{YT_{+},11}(c)$ be a $q \times q$ matrix with $(k,k')$ element $\nu_{+,0}(c)\phi_{kk'}$, $\Sigma_{YT_{+},12}(c)$ be a $q \times p$ matrix with $(k,j)$ element $\sum_{k'=1}^{q}\phi_{kk'}\nu_{+,j}(c)$, $\Sigma_{YT_{+},21}(c)$ be the transpose of $\Sigma_{YT_{+},12}(c)$, $\Sigma_{YT_{+},22}(c)$ be a $p \times p$ matrix with $(j,j')$ element $\sum_{k=1}^{q}\sum_{k'=1}^{q} \phi_{kk'} \nu_{+,j+j'}(c)$. Define
\begin{equation}
	\Sigma_{YT_+}(c) =
	\begin{pmatrix}
	\Sigma_{YT_{+},11}(c) & \Sigma_{YT_{+},12}(c) \\
	\Sigma_{YT_{+},21}(c) & \Sigma_{YT_{+},22}(c)
	\end{pmatrix}.\label{notation of YT plus}
\end{equation}

Like (\ref{notation of Thm1}), (\ref{notation of s inverse}) and (\ref{notation of YT plus}) above,  a similar set of definitions can be provided to other variables on the boundary, including $S_{Y_-}(c)$, $S_{Y_-}^{-1}(c)$, $\Sigma_{Y_-}(c)$, $S_{T_+}(c)$, $S_{T_+}^{-1}(c)$, $\Sigma_{T_+}(c)$,  $S_{T_-}(c)$, $S_{T_-}^{-1}(c)$, $\Sigma_{T_-}(c)$, $F_{-}(c_k,c_{k'})$, $\Sigma_{YT_-}(c)$, and $\phi_{kk'}$ can also be redefined with $F_{-}(c_k,c_{k'})$.


Let $x_{+,i} = (X_{+,i} - x)/h_{Y_+}$ and $K_{+,i} = K(x_{+,i})$ with $x=0$. Define
\begin{align}
	u_k &= \sqrt{n_+h_{Y_+}}(a_{k} - m_{Y_+}(x) - \sigma_{\epsilon_{Y_+}} c_k), \ k = 1,\cdots,q, \nonumber\\
	v_j &= h_{Y_+}^j \sqrt{n_+h_{Y_+}} (j!b_j - m_{Y+}^{(j)}(x))/j!, \ j=1,\cdots,p, \nonumber\\
	\Delta_{i,k} &= \frac{u_k}{\sqrt{n_+ h_{Y_+}}} + \sum_{j=1}^{p}\frac{v_j x_{+,i}^j}{\sqrt{n_+ h_{Y_+}}},\nonumber\\
	r_{i,p} &= m_{Y_+}(X_{+,i}) - \sum_{j=0}^{p} m_{Y_+}^{(j)}(x)(X_{+,i} - x)^j/j!,\nonumber\\
	d_{i,k} &= c_k [\sigma_{\epsilon_{Y_+}}(X_{+,i}) - \sigma_{\epsilon_{Y_+}}(x)] + r_{i,p}.\label{notation of set uvrd}
\end{align}
Let $W_{Y_+,n_+}^* = (w_{Y_+,11}^*,\cdots,w_{Y_+,1q}^*,w_{Y_+,21}^*,\cdots,w_{Y_+,2p}^*)^T= (w_{Y_+,1n}^*,w_{Y_+,2n}^*)^T$, where 
\begin{align}
	w_{Y_+,1k}^* &= \frac{1}{\sqrt{n_+h_{Y_+}}} \sum_{i=1}^{n_+} K(x_{+,i})\eta_{Y_+,i,k}^*,\nonumber\\
	w_{Y_+,2j}^* &= \frac{1}{\sqrt{n_+h_{Y_+}}} \sum_{k=1}^{q} \sum_{i=1}^{n_+} K(x_{+,i}) x_{+,i}^j \eta_{Y_+,i,k}^*, \nonumber\\
	\eta_{Y_+,i,k}^* &= I(\epsilon_{Y_+,i} \leq c_k - \frac{d_{i,k}}{\sigma_{\epsilon_{Y_+,i}}}) - \tau_k.
\end{align}
Also let $W_{Y_+,n_+} = (w_{Y_+,11},\cdots,w_{Y_+,1q},w_{Y_+,21},\cdots,w_{Y_+,2p})^T = (w_{Y_+,1n}, w_{Y_+,2n})^T$, where
\begin{align}
w_{Y_+,1k} &= \frac{1}{\sqrt{n_+h_{Y_+}}} \sum_{i=1}^{n_+} K(x_{+,i})\eta_{Y_+,i,k},\nonumber\\
w_{Y_+,2j} &= \frac{1}{\sqrt{n_+h_{Y_+}}} \sum_{k=1}^{q} \sum_{i=1}^{n_+} K(x_{+,i}) x_{+,i}^j \eta_{Y_+,i,k}, \nonumber\\
\eta_{Y_+,i,k} &= I(\epsilon_{Y_+,i} \leq c_k) - \tau_k.
\end{align}
Similarly, we define $W_{T_+,n_+}^*,w_{T_+,1k}^*,w_{T_+,2j}^*, \eta_{T_+,i,k}^*, W_{T_+,n_+},w_{T_+,1k},w_{T_+,2j}$ and $\eta_{T_+,i,k}$.

Consider the case $p = 1$ and define $\theta = (u_1,\cdots,u_q,v_1)^T$.  Let $\hat{\theta}_{n_+} = (\hat{u}_1,\cdots,\hat{u}_q,\hat{v}_1)^T$ be the transformed minimizer of (\ref{eq:lcqr obj}).  It can be shown that minimizing (\ref{eq:lcqr obj}) is equivalent to minimizing
\begin{equation*} \label{eq:lcqr obj2}
	L_{n_+}(\theta) = \sum_{i=1}^{n_+} \left(K(x_{+,i})\sum_{k=1}^{q}(\rho_{\tau_k}(\sigma_{\epsilon_{Y_+,i}} (\epsilon_{Y_+,i} - c_k) + d_{i,k} - \Delta_{i,k}) - \rho_{\tau_k}(\sigma_{\epsilon_{Y_+,i}} (\epsilon_{Y_+,i} - c_k) + d_{i,k}) )\right).
\end{equation*}


\setstretch{1.45}
Next, in a similar fashion we introduce the notation used for fixed-$n$ approximations in \Cref{supp:fixed-n}.  The notation of $K_{+,i}$ as well as $x_{+,i}$ has been provided for (\ref{notation of set uvrd}). Let $S_{nY_{+},11}$ be a $q \times q$ diagonal matrix with diagonal elements $f_{\epsilon_{Y_+}}(c_k) \frac{1}{n_+ h_{Y+}}\sum_{i=1}^{n+} \frac{K_{+,i}}{\sigma_{\epsilon_{Y_+,i}}}  $, $S_{nY_{+},12}$ be a $q \times p$ matrix with $(k,j)$ element $f_{\epsilon_{Y_+}}(c_k) \frac{1}{n_+ h_{Y+}}\sum_{i=1}^{n+} \frac{K_{+,i}x_{+,i}^j}{\sigma_{\epsilon_{Y_+,i}}} $, $S_{nY_{+},21}$ be the transpose of $S_{nY_{+},12}$, and $S_{nY_{+},22}$ be a $p \times p$ matrix with $(j,j')$ element equal to  $\sum_{k=1}^{q}f_{\epsilon_{Y_+}}(c_k) \frac{1}{n_+ h_{Y+}}\sum_{i=1}^{n+} \frac{K_{+,i}x_{+,i}^{j+j'}}{\sigma_{\epsilon_{Y_+,i}}}$.  Let $\Sigma_{nY_{+},11}$ be a $q \times q$ matrix with $(k,k')$ element $\frac{1}{n_+ h_{Y+}}\sum_{i=1}^{n+} K_{+,i}^2 \tau_{kk'}$, $\Sigma_{nY_{+},12}$ be a $q \times p$ matrix with $(k,j)$ element $\sum_{k'=1}^{q} \tau_{kk'} \frac{1}{n_+ h_{Y+}}\sum_{i=1}^{n+} K_{+,i}^2 x_{+,i}^j $, $\Sigma_{nY_{+},21}$ be the transpose of $\Sigma_{nY_{+},12}(c)$, $\Sigma_{nY_{+},22}$ be a $p \times p$ matrix with $(j,j')$ element equal to $\sum_{k=1}^{q}\sum_{k'=1}^{q} \tau_{kk'} \frac{1}{n_+ h_{Y+}}\sum_{i=1}^{n} K_{+,i}^2 x_{+,i}^{j+j'}$.

Similar to (\ref{notation of Thm1}) and (\ref{notation of s inverse}), define
\begin{equation}
S_{nY_+} = 
\begin{pmatrix}
S_{nY_{+},11} & S_{nY_{+},12}\\
S_{nY_{+},21} & S_{nY_{+},22}
\end{pmatrix},
\quad
\Sigma_{nY_+} =
\begin{pmatrix}
\Sigma_{nY_{+},11} & \Sigma_{nY_{+},12} \\
\Sigma_{nY_{+},21} & \Sigma_{nY_{+},22}
\end{pmatrix},\label{notation of fixedn s sigma}
\end{equation}
\begin{equation}
S_{nY_+}^{-1} = 
\begin{pmatrix}
(S_{nY_+}^{-1})_{11} & (S_{nY_+}^{-1})_{12} \\
(S_{nY_+}^{-1})_{21} & (S_{nY_+}^{-1})_{22}
\end{pmatrix}.
\end{equation}

Similar to (\ref{notation of YT plus}), let $\Sigma_{nYT_{+},11}$ be a $q \times q$ matrix with $(k,k')$ element $\phi_{kk'}\frac{1}{n_+ \sqrt{h_{Y+}h_{T+}}}\sum_{i=1}^{n+} K_{+,i}^2 $, $\Sigma_{nYT_{+},12}$ be a $q \times p$ matrix with $(k,j)$ element $\sum_{k'=1}^{q}\phi_{kk'} \frac{1}{n_+ \sqrt{h_{Y+}h_{T+}}}\sum_{i=1}^{n+} K_{+,i}^2 x_{+,i}^j$, $\Sigma_{nYT_{+},21}$ be the transpose of $\Sigma_{nYT_{+},12}$, and $\Sigma_{nYT_{+},22}$ be a $p \times p$ matrix with $(j,j')$ element  equal to
$\sum_{k=1}^{q}\sum_{k'=1}^{q} \phi_{kk'} \frac{1}{n_+ \sqrt{h_{Y+}h_{T+}}}\sum_{i=1}^{n+} K_{+,i}^2 x_{+,i}^{j+j'}$.  Define
\begin{equation}
\Sigma_{nYT_+} =
\begin{pmatrix}
\Sigma_{nYT_{+},11} & \Sigma_{nYT_{+},12} \\
\Sigma_{nYT_{+},21} & \Sigma_{nYT_{+},22}
\end{pmatrix}.
\end{equation}

A similar set of definitions can be provided to other fixed-$n$ variables on the boundary, including $S_{nY_-}$, $S_{nY_-}^{-1}$, $\Sigma_{nY_-}$,  $S_{nT_+}$, $S_{nT_+}^{-1}$, $\Sigma_{nT_+}$ , $S_{nT_-}$ , $S_{nT_-}^{-1}$, $\Sigma_{nT_-}$ and $\Sigma_{nYT_-}$.

Given the above fixed-$n$ definitions and let $x=0$, it can be verified that, as $n_+ \rightarrow \infty$, 
\begin{equation*}
    S_{nY_+} \rightarrow \frac{f_{X_+}(x)}{\sigma_{\epsilon_{Y_+}}(x)} S_{Y_+}(c), \: \Sigma_{nY_+} \rightarrow f_{X_+}(x)\Sigma_{Y_+}(c) \textrm{ and } \Sigma_{nYT_+} \rightarrow f_{X_+}(x)\Sigma_{YT_+}(c).
\end{equation*}

\subsection{Lemmas and proofs for \Cref{thm:sharp results,thm:fuzzy results,thm:sharp t adjusted dist,thm:fuzzy t adjusted dist,thm:adjusted bandwidth}} \label{supp:asymptotic}


\begin{lemma} \label{lemma:theta distribution}
	Under \cref{assumption:m continuity,assumption:sigma continuity,assumption:kernel,assumption:density x,assumption:error distribution,assumption:bandwidth}, as $n_+ \rightarrow \infty$, we have
	\begin{equation}
		\hat{\theta}_{n_+} + \frac{\sigma_{\epsilon_{Y_+}}(0)}{f_{X_+}(0)} S_{Y_+}^{-1}(c) E(W_{n_+}^*|\mathbf{X}) \overset{L}{\rightarrow} MVN \left(\mathbf{0},\frac{\sigma_{\epsilon_{Y_+}}^2(0)}{f_{X_+}(0)} S_{Y_+}^{-1}(c) \Sigma_{Y_+}(c) S_{Y_+}^{-1}(c) \right).
	\end{equation}
\end{lemma}
\begin{proof}[Proof of \Cref{lemma:theta distribution}]
	See the proof of Theorem 2.1 in \cite{kai2010local_supp}.
\end{proof}

\begin{lemma} \label{lemma:bias and variance}
	Under \cref{assumption:m continuity,assumption:sigma continuity,assumption:kernel,assumption:density x,assumption:error distribution,assumption:bandwidth}, as $n_+ \rightarrow \infty$, the asymptotic bias and variance for the LCQR estimator in \cref{eq:Y plus model} are given by
	\begin{align}
		\text{Bias}(\hat{m}_{Y_+}(0)|\mathbf{X})&= \frac{1}{2} a_{Y_+}(c) m_{Y_+}^{(2)}(0) h_{Y_+}^2 + o_p(h_{Y_+}^2),\nonumber\\
		\text{Var}(\hat{m}_{Y_+}(0)|\mathbf{X}) &= \frac{1}{n_+h_{Y_+}}\frac{b_{Y_+}(c)\sigma_{\epsilon_{Y_+}}^2(0)}{f_{X_+}(0)} + o_p(\frac{1}{n_+h_{Y_+}}),\nonumber\\
		a_{Y_+}(c) &= \frac{\mu_{+,2}^2(c) - \mu_{+,1}(c)\mu_{+,3}(c)}{\mu_{+,0}(c)\mu_{+,2}(c) - \mu_{+,1}^2(c)},\nonumber\\ 
		b_{Y_+}(c) &= e_q^T (S_{Y_+}^{-1}(c) \Sigma_{Y_+}(c) S_{Y_+}^{-1}(c))_{11} e_q/q^2.  
	\end{align}
\end{lemma}
\begin{proof}[Proof of \Cref{lemma:bias and variance}]
	The bias result follows that in Theorem 2.2 in \cite{kai2010local_supp}. The variance result also largely follows that in \cite{kai2010local_supp}. Given 
	\begin{equation} \label{eq:var}
		\text{Var}(\hat{m}_{Y_+}(0)|\mathbf{X}) = \frac{1}{n_+h_{Y_+}} \frac{\sigma_{\epsilon_{Y_+}}^2}{q^2 f_{X_+}(0)} e_q^T (S_{Y_+}^{-1}(c) \Sigma_{Y_+}(c) S_{Y_+}^{-1}(c))_{11} e_q + o_p(\frac{1}{n_+h_{Y_+}}),
	\end{equation}
	It is easy to verify that when $q = 1$, \cref{eq:var} can be written as
	\begin{align} \label{eq:var q=1}
		\text{Var}(\hat{m}_{Y_+}(0)|\mathbf{X}) &= \frac{1}{n_+h_{Y_+}} \frac{\sigma_{\epsilon_{Y_+}}^2}{f_{X_+}(0)} \frac{\mu_{+,2}^2(c)\nu_{+,0}(c) - 2 \mu_{+,1}(c)\mu_{+,2}(c)\nu_{+,1}(c)+\mu_{+,1}^2(c)\nu_{+,2}(c)}{(\mu_{+,0}(c)\mu_{+,2}(c)-\mu_{+,1}^2(c))^2} R_1(q) \nonumber\\
		&\quad+ o_p(\frac{1}{n_+h_{Y_+}}),
	\end{align}
where $R_1(q) = \frac{1}{q^2} \sum_{k=1}^{q} \sum_{k'=1}^{q} \frac{\tau_{kk'}}{f_{\epsilon_{Y_+}(c_k)}f_{\epsilon_{Y_+}(c_{k'})}}$.
However, for $q \geq 2$, the result in \cref{eq:var q=1} no longer holds and we use \cref{eq:var} instead.
\end{proof}


\begin{lemma} \label{lemma:covariance}
	Under \cref{assumption:m continuity,assumption:sigma continuity,assumption:kernel,assumption:density x,assumption:error distribution,assumption:bandwidth}, as $n_+ \rightarrow \infty$, the covariance between $\hat{m}_{Y_+}(x)$ and $\hat{m}_{T_+}(x)$ at the boundary point $0$ is given by
	
	\begin{align}
\text{Cov}(\hat{m}_{Y_+}(0),\hat{m}_{T_+}(0)|\mathbf{X}) 		&= \frac{1}{n_+\sqrt{h_{Y_+}h_{T_+}}}\frac{\sigma_{\epsilon_{Y_+}}(0)\sigma_{\epsilon_{T_+}}(0)}{f_{X_+}(0)} b_{YT_+}+ o_p(\frac{1}{n_+h_{Y_+}} + \frac{1}{n_+h_{T_+}}), \label{eq:cov q=2}
	\end{align}
	where
	\begin{equation} \label{eq:b_YT+}
		b_{YT_+} = \frac{1}{q^2} e_q^T
		\left(S_{Y_+}^{-1}(c) \Sigma_{YT_+} S_{T_+}^{-1}(c)\right)_{11} e_q.
	\end{equation}
\end{lemma}
\begin{proof}[Proof of \Cref{lemma:covariance}]
	Assume $p = 1$. From \Cref{lemma:theta distribution}, we write
	\begin{align*}
	&\quad \hat{m}_{Y_+}(0) - E(\hat{m}_{Y_+}(0)|\mathbf{X}) \nonumber\\ 
	&= -\frac{1}{q\sqrt{n_+h_{Y_+}}}\frac{\sigma_{\epsilon_{Y_+}}(0)}{f_{X_+}(0)} e_q^T
	\begin{pmatrix}
	(S_{Y_+}^{-1}(c))_{11} & (S_{Y_+}^{-1}(c))_{12}
	\end{pmatrix}
	\begin{pmatrix}
	w_{Y_+,1n}^* - E(w_{Y_+,1n}^*|\mathbf{X})\\
	w_{Y_+,2n}^* - E(w_{Y_+,2n}^*|\mathbf{X})
	\end{pmatrix} + o_p(1) \nonumber\\
	&= -\frac{1}{q\sqrt{n_+h_{Y_+}}}\frac{\sigma_{\epsilon_{Y_+}}(0)}{f_{X_+}(0)} e_q^T
	\begin{pmatrix}
	(S_{Y_+}^{-1}(c))_{11} & (S_{Y_+}^{-1}(c))_{12}
	\end{pmatrix}
	\begin{pmatrix}
	w_{Y_+,1n} - E(w_{Y_+,1n}|\mathbf{X}) \\
	w_{Y_+,2n} - E(w_{Y_+,2n}|\mathbf{X})
	\end{pmatrix} + o_p(1),
	\end{align*}
	where the last equality follows by the result that $\text{Var}(w_{Y_+,1n}^* - w_{Y_+,1n}|\mathbf{X}) = o_p(1)$ and $\text{Var}(w_{Y_+,21}^* - w_{Y_+,21}|\mathbf{X}) = o_p(1)$. See \cite{kai2010local} for a proof. 
	Similarly, we have
	\begin{align*}
	&\quad \hat{m}_{T_+}(0) - E(\hat{m}_{T_+}(0)|\mathbf{X}) \nonumber\\ 
	&= -\frac{1}{q\sqrt{n_+h_{T_+}}}\frac{\sigma_{\epsilon_{T_+}}(0)}{f_{X_+}(0)} e_q^T
	\begin{pmatrix}
	(S_{T_+}^{-1}(c))_{11} & (S_{T_+}^{-1}(c))_{12}
	\end{pmatrix}
	\begin{pmatrix}
	w_{T_+,1n} - E(w_{T_+,1n}|\mathbf{X})\\
	w_{T_+,2n} - E(w_{T_+,2n}|\mathbf{X})
	\end{pmatrix} + o_p(1). \\
	&\quad\text{Cov}(\hat{m}_{Y_+}(0),\hat{m}_{T_+}(0)|\mathbf{X})\\
	&= E\left((\hat{m}_{Y_+}(0) - E(\hat{m}_{Y_+}(0)|\mathbf{X}))(\hat{m}_{T_+}(0) - E(\hat{m}_{T_+}(0)|\mathbf{X}))\right)\\
	&= \frac{1}{q^2n_+\sqrt{h_{Y_+}h_{T_+}}}\frac{\sigma_{\epsilon_{Y_+}}(0)\sigma_{\epsilon_{T_+}}(0)}{f_{X_+}^2(0)}e_q^T
	\begin{pmatrix}
	(S_{Y_+}^{-1}(c))_{11} & (S_{Y_+}^{-1}(c))_{12}
	\end{pmatrix}\\
	&\quad \times E
	\begin{bmatrix}
	\begin{pmatrix}
	w_{Y_+,1n} - E(w_{Y_+,1n}|\mathbf{X})\\
	w_{Y_+,2n} - E(w_{Y_+,2n}|\mathbf{X})
	\end{pmatrix} 
		\begin{pmatrix}
	w_{T_+,1n} - E(w_{T_+,1n}|\mathbf{X})\\
	w_{T_+,2n} - E(w_{T_+,2n}|\mathbf{X})
	\end{pmatrix}^T
	\end{bmatrix}\times 
	\begin{pmatrix}
	(S_{T_+}^{-1}(c))_{11} \\
	(S_{T_+}^{-1}(c))_{12}
	\end{pmatrix} e_q\\
 &=\frac{1}{q^2n_+\sqrt{h_{Y_+}h_{T_+}}}\frac{\sigma_{\epsilon_{Y_+}}(0)\sigma_{\epsilon_{T_+}}(0)}{f_{X_+}(0)}e_q^T
	\left(S_{Y_+}^{-1}(c) \Sigma_{YT_+} S_{T_+}^{-1}(c)\right)_{11} e_q + o_p(\frac{1}{n_+h_{Y_+}} + \frac{1}{n_+h_{T_+}}),
	\end{align*}
	where $\text{Cov}(\eta_{Y_+,i,k},\eta_{T_+,j,k'}) = \phi_{kk'}$ if $i=j$, and $\text{Cov}(\eta_{Y_+,i,k},\eta_{T_+,j,k'}) = 0$ if $i \ne j$.
\end{proof}

\begin{proof}[Proof of \Cref{thm:sharp results}]
	See \Cref{lemma:bias and variance} above.
\end{proof}

\begin{proof}[Proof of \Cref{thm:fuzzy results}]
Consider the approximation
	\begin{align*}
		\hat{\tau}_{\text{fuzzy}} - \tau_{\text{fuzzy}} &= 
		\frac{1}{m_{T_+}(0) - m_{T_-}(0)}\left[\hat{m}_{Y_+}(0) - m_{Y_+}(0) - (\hat{m}_{Y_-}(0)- m_{Y_-}(0))\right] \nonumber \\
		&\quad - \frac{m_{Y_+}(0) - m_{Y_-}(0)}{\left[m_{T_+}(0) - m_{T_-}(0)\right]^2}\left[\hat{m}_{T_+}(0) - m_{T_+}(0) - (\hat{m}_{T_-}(0)- m_{T_-}(0))\right]\nonumber\\
		&\quad + o_p(h_{Y_+}^2 + h_{Y_-}^2 + h_{T_+}^2 + h_{T_-}^2). 
	\end{align*} 
The bias expression follows from \Cref{lemma:bias and variance}: use it four times for $\hat{m}_{Y_+}(0)$,  $\hat{m}_{Y_-}(0)$, $\hat{m}_{T_+}(0)$, and $\hat{m}_{T_-}(0)$. For the variance expression, note that the approximation above leads to
\begin{eqnarray}
    &&\text{Var}(\hat{\tau}_{\text{fuzzy}}) \notag \\
    &=& \frac{\text{Var}(\hat{m}_{Y_+}(0)) + \text{Var}(\hat{m}_{Y_-}(0))}{\left(m_{T_+}(0) - m_{T_-}(0)\right)^2} + \frac{\left(m_{Y_+}(0) - m_{Y_-}(0)\right)^2}{\left(m_{T_+}(0) - m_{T_-}(0)\right)^4} \left[\text{Var}(\hat{m}_{T_+}(0)) + \text{Var}(\hat{m}_{T_-}(0))\right] \notag\\
   && - 2\frac{m_{Y_+}(0) - m_{Y_-}(0)}{\left(m_{T_+}(0) - m_{T_-}(0)\right)^3} \left[\text{Cov}(\hat{m}_{Y_+}(0),\hat{m}_{T_+}(0)) + \text{Cov}(\hat{m}_{Y_-}(0),\hat{m}_{T_-}(0))\right] + \text{\small{s.o.}} \label{fuzzy variance long form}
\end{eqnarray}	
where $s.o.$ denotes a small order term.

Plugging the variance and covariance expressions in \Cref{lemma:bias and variance,lemma:covariance} to (\ref{fuzzy variance long form}) leads to the asymptotic variance expression of $\hat{\tau}_{\text{fuzzy}}$. 
\end{proof}

For convenience we write $\hat{m}_{Y_+}(0)$ and $\hat{m}_{Y_-}(0)$ as $\hat{m}_{Y_+}$ and $\hat{m}_{Y_-}$, respectively. \Cref{eq:sharp var adjusted} suggests that we need the expressions for $\text{Var}(\text{Bias}(\hat{m}_{Y_+}))$ and $\text{Cov}(\hat{m}_{Y_+}, \text{Bias}(\hat{m}_{Y_+}))$ to adjust the variance. The next lemma provides results for computing $\text{Var}(\text{Bias}(\hat{m}_{Y_+}))$. In deriving the results, we also present the bias of $\text{Bias}(\hat{m}_{Y_+})$. Let $e_r$ be a $p \times 1$ unit vector with  the $r$-th element equal to one. Let $p = 3$ in the following proof.

\begin{lemma} \label{lemma:bias and var of bias}
	Under Assumptions 1 to 6, as  $n_+ \rightarrow \infty$, the asymptotic bias and variance of $\hat{m}_{Y_+}^{(2)}$ are given by
	\begin{align} 
		\text{Bias}(\hat{m}_{Y_+}^{(2)}|\mathbf{X}) &= \frac{1}{12} a_{Y_+}^*(c) m_{Y_+}^{(4)} h_{Y_+}^2 + o_p(h_{Y_+}^2),	\label{eq:bias of m2d }\\
		\text{Var }(\hat{m}_{Y_+}^{(2)}|\mathbf{X}) &= \frac{4}{n_+ h_{Y_+}^5} \frac{\sigma_{\epsilon_{Y_+}}^2(0) b_{Y_+}^*(c)}{f_{X_+}(0)} + o_p(\frac{1}{n_+ h_{Y_+}^5}) \label{eq:variance of m2d},
	\end{align}
	where
	\begin{align}
		a_{Y_+}^*(c) &= \mu_{+,4} e_2^T (S_{Y_+}^{-1}(c))_{21} f_{\epsilon_{Y_+}} + \sum_{k=1}^{q}f_{\epsilon_{Y_+}}(c_k) e_2^T (S_{Y_+}^{-1}(c))_{22} 
		(\mu_{+,5},
		\mu_{+,6},
		\mu_{+,7})^T\\
		b_{Y_+}^*(c) &= e_2^T(S_{Y_+}^{-1}(c) \Sigma_{Y_+}(c) S_{Y_+}^{-1})_{22} e_2.
	\end{align}
\end{lemma}

\begin{proof}[Proof of \Cref{lemma:bias and var of bias}]
	From the definition of $v_j$, we have
	\begin{equation} \label{eq:md2 and v}
		\hat{m}_{Y_+}^{(2)} = m_{Y_+}^{(2)} + \frac{2\hat{v}_2}{h_{Y_+}^2  \sqrt{n_+ h_{Y_+}}}.
	\end{equation}
	Hence the bias becomes
	\begin{align*}
		&E(\hat{m}_{Y_+}^{(2)}) - m_{Y_+}^{(2)} = -\frac{2\sigma_{\epsilon_{Y_+}}(0)}{h_{Y_+}^2  \sqrt{n_+ h_{Y_+}}f_{X_+}(0)}  e_2^T ((S_{Y_+}^{-1}(c))_{21}, (S_{Y_+}^{-1}(c))_{22}) E(W_{Y_+,n}^*)\\
		&= -\frac{2\sigma_{\epsilon_{Y_+}}(0)}{h_{Y_+}^2  \sqrt{n_+ h_{Y_+}}f_{X_+}(0)} e_2^T (S_{Y_+}^{-1}(c))_{21} E(W_{Y_+,1n}^*) -\frac{2\sigma_{\epsilon_{Y_+}}(0)}{h_{Y_+}^2  \sqrt{n_+ h_{Y_+}}f_{X_+}(0)} e_2^T (S_{Y_+}^{-1}(c))_{22} E(W_{Y_+,2n}^*)\\
		&= \textrm{I} + \textrm{II}.
	\end{align*}
	\begin{align*}
		\textrm{I} &= -\frac{2\sigma_{\epsilon_{Y_+}}(0)}{h_{Y_+}^2  \sqrt{n_+ h_{Y_+}}f_{X_+}(0)} e_2^T (S_{Y_+}^{-1}(c))_{21} \\
		&\quad \times \left[-\frac{f_{\epsilon_{Y_+}}}{\sqrt{n_+ h_{Y_+}}} \sum_{i = 1}^{n_+} K_i c_k \frac{\sigma_{\epsilon_{Y_+,i}} - \sigma_{\epsilon_{Y_+}}(0)}{\sigma_{\epsilon_{Y_+,i}}} -\frac{f_{\epsilon_{Y_+}}}{\sqrt{n_+ h_{Y_+}}} \sum_{i = 1}^{n_+} K_i \frac{r_{i,3}}{\sigma_{\epsilon_{Y_+,i}}} \right]\\
		&= \frac{1}{12} m^{(4)}_{Y_+} \mu_{+,4}(c) e_2^T (S_{Y_+}^{-1}(c))_{21} f_{\epsilon_{Y_+}} h_{Y_+}^2 + o_p(h_{Y_+}^2),\\
		\textrm{II} &=  -\frac{2\sigma_{\epsilon_{Y_+}}(0)}{h_{Y_+}^2  \sqrt{n_+ h_{Y_+}}f_{X_+}(0)} e_2^T (S_{Y_+}^{-1}(c))_{22}\\
		&\quad \times\left[-\frac{\sum_{k=1}^{q}f_{\epsilon_{Y_+}}(c_k)}{\sqrt{n_+ h_{Y_+}}} \sum_{i = 1}^{n_+} K_i c_k \frac{\sigma_{\epsilon_{Y_+,i}} - \sigma_{\epsilon_{Y_+}}(0)}{\sigma_{\epsilon_{Y_+,i}}} \begin{pmatrix}
		x_{+,i}\\
		x_{+,i}^2\\
		x_{+,i}^3
		\end{pmatrix}
		 -\frac{\sum_{k=1}^{q}f_{\epsilon_{Y_+}}(c_k)}{\sqrt{n_+ h_{Y_+}}} \sum_{i = 1}^{n_+} K_i  \frac{r_{i,3}}{\sigma_{\epsilon_{Y_+,i}}} \begin{pmatrix}
		 x_{+,i}\\
		 x_{+,i}^2\\
		 x_{+,i}^3
		 \end{pmatrix}\right] \\
		 &= \frac{1}{12} m^{(4)}_{Y_+} \sum_{k=1}^{q}f_{\epsilon_{Y_+}}(c_k) e_2^T (S_{Y_+}^{-1}(c))_{22} \begin{pmatrix}
		 \mu_{+,5}\\
		 \mu_{+,6}\\
		 \mu_{+,7}
		 \end{pmatrix} h_{Y_+}^2 + o_p(h_{Y_+}^2).
	\end{align*}
	The bias result is proved by combining the two terms \textrm{I} and \textrm{II}. One would expect a number of $4! = 24$ instead of $12$ on the denominator. This is due to the extra number $2$ in \cref{eq:md2 and v}. Because of the way $\hat{v}_2$ is defined, the ``effective" constant on the denominator is still $24$, in line with the standard results for nonparametric derivatives. Similarly, the number $4$ appearing on the numerator of the variance is also a result of the number $2$ in \cref{eq:md2 and v}.
	The variance results from \cref{eq:md2 and v} and  \Cref{lemma:theta distribution}.
\end{proof}

\begin{proof} [Proof of \Cref{thm:sharp t adjusted dist}]
	Following \Cref{thm:sharp results}, we have
	\begin{align*} 
		\text{Var}(\hat{m}_{Y_+}) &= \frac{1}{n_+h_{Y_+}} \frac{b_{Y_+}(c)\sigma_{\epsilon_{Y_+}}^2(0)}{f_{X_+}(0)} + o_p(\frac{1}{n_+h_{Y_+}}),\\
		\text{Var}(\hat{m}_{Y_-}) &= \frac{1}{n_-h_{Y_-}} \frac{b_{Y_-}(c)\sigma_{\epsilon_{Y_-}}^2(0)}{f_{X_-}(0)} + o_p(\frac{1}{n_-h_{Y_-}}).
	\end{align*}		
	Use the bias expression in \Cref{thm:sharp results} and the variance result in \Cref{lemma:bias and var of bias}, we have
	\begin{align*}
		\text{Var}(\widehat{\text{Bias}}(\hat{m}_{Y_+})) &= \frac{\sigma_{\epsilon_{Y_+}}^2(0)}{n_+ h_{Y_+} f_{X_+}(0)}a^2(c)b_{Y_+}^*(c) + o_p(\frac{1}{n_+h_{Y_+}}),\\
		\text{Var}(\widehat{\text{Bias}}(\hat{m}_{Y_-})) &= \frac{\sigma_{\epsilon_{Y_-}}^2(0)}{n_- h_{Y_-} f_{X_-}(0)}a^2(c)b_{Y_-}^*(c) + o_p(\frac{1}{n_-h_{Y_-}}).
	\end{align*}
	For the covariances, we have
	\begin{align*}
		\text{Cov}(\hat{m}_{Y_+},\widehat{\text{Bias}}(\hat{m}_{Y_+})) &= \text{Cov}(m_{Y_+} + \frac{1}{q\sqrt{n_+ h_{Y_+}}}\sum_{k=1}^{q}\hat{u}_k, \frac{1}{2}a_{Y_+}(c)h_{Y_+}^2(m_{Y_+}^{(2)} + \frac{2\hat{v}_2}{h_{Y_+}^2\sqrt{n_+ h_{Y_+}}})   \\
		&= \frac{a_{Y_+}(c)}{n_+ h_{Y_+}q}\sum_{k=1}^{q} \text{Cov}(\hat{u}_k,\hat{v}_2)\\
		&= \frac{a_{Y_+}(c)\sigma_{\epsilon_{Y_+}}^2(0)}{n_+ h_{Y_+}qf_{X_+}(0)} e_q^T (S_{Y_+}^{-1} \Sigma_{Y_+} S_{Y_+}^{-1})_{12,2} + o_p(\frac{1}{n_+h_{Y_+}}),
	\end{align*}
	where $(S_{Y_+}^{-1} \Sigma_{Y_+} S_{Y_+}^{-1})_{12,2}$	is the second column of the matrix $(S_{Y_+}^{-1} \Sigma_{Y_+} S_{Y_+}^{-1})_{12}$ and the last line follows from \Cref{lemma:theta distribution}.
	Similarly, for data below the cutoff, we have
	\begin{equation*}
			\text{Cov}(\hat{m}_{Y_-},\widehat{\text{Bias}}(\hat{m}_{Y_-})) = \frac{a_{Y_-}(c)\sigma_{\epsilon_{Y_-}}^2(0)}{n_- h_{Y_-}qf_{X_-}(0)} e_q^T (S_{Y_-}^{-1} \Sigma_{Y_-} S_{Y_-}^{-1})_{12,2} + o_p(\frac{1}{n_-h_{Y_-}}).
	\end{equation*}
	
	The expression for $\text{Var}(\hat{\tau}_{\text{sharp}} - \widehat{\text{Bias}}(\hat{\tau}_{\text{sharp}}))$ is obtained by substituting the six variance and covariance results into \cref{eq:sharp var adjusted},
	\begin{equation*}
		\text{Var}(\hat{\tau}_{\text{sharp}} - \widehat{\text{Bias}}(\hat{\tau}_{\text{sharp}})) = \frac{1}{n_+h_{Y_+}} V_{\text{sharp},+} + \frac{1}{n_-h_{Y_-}} V_{\text{sharp},-}, 
	\end{equation*}
	where
	\begin{align*}
		 V_{\text{sharp},+} &=  \frac{b_{Y_+}(c)\sigma_{\epsilon_{Y_+}}^2(0)}{f_{X_+}(0)} + \frac{\sigma_{\epsilon_{Y_+}}^2(0)}{f_{X_+}(0)}a^2(c)b_{Y_+}^*(c) - 2 \frac{a_{Y_+}(c)\sigma_{\epsilon_{Y_+}}^2(0)}{qf_{X_+}(0)} e_q^T (S_{Y_+}^{-1} \Sigma_{Y_+} S_{Y_+}^{-1})_{12,2}, \nonumber\\
		 V_{\text{sharp},-} &= \frac{b_{Y_-}(c)\sigma_{\epsilon_{Y_-}}^2(0)}{f_{X_-}(0)} +\frac{\sigma_{\epsilon_{Y_-}}^2(0)}{f_{X_-}(0)}a^2(c)b_{Y_-}^*(c) - 2 \frac{a_{Y_-}(c)\sigma_{\epsilon_{Y_-}}^2(0)}{qf_{X_-}(0)} e_q^T (S_{Y_-}^{-1} \Sigma_{Y_-} S_{Y_-}^{-1})_{12,2}.
 	\end{align*}
 	
 	Next, we establish the asymptotic normality of the adjusted \textit{t}-statistic. From \Cref{lemma:bias and variance}, we have
 	\begin{align} 
 	    \frac{\hat{m}_{Y_+} - \frac{1}{2} a_{Y_+}(c)m_{Y_+}^{(2)} h_{Y_+}^2 - m_{Y_+}}{\sqrt{\text{Var}(\hat{m}_{Y_+})}} &= \frac{\hat{m}_{Y_+} - E(\hat{m}_{Y_+})}{\sqrt{\text{Var}(\hat{m}_{Y_+})}} + \frac{E(\hat{m}_{Y_+}) - m_{Y_+} - \frac{1}{2} a_{Y_+}(c)m_{Y_+}^{(2)} h_{Y_+}^2 }{\sqrt{\text{Var}(\hat{m}_{Y_+})}} \nonumber \\ 
 	    &=\frac{\hat{m}_{Y_+} - E(\hat{m}_{Y_+})}{\sqrt{\text{Var}(\hat{m}_{Y_+})}} + \frac{O_p(h_{Y_+}^3)}{O_p(\sqrt{1/n_+h_{Y_+}})}.\label{eq:m bc t stat} \\
 	    & \overset{d}{\rightarrow} N(0,1). \label{eq:m bc t stat normality}
 	\end{align}
 	The second term in \cref{eq:m bc t stat} converges to 0 under \Cref{assumption:bandwidth}. In the first term, given the definition of $u_k$ in \cref{notation of set uvrd} and since $\hat{m}_{Y_+}$ is a linear function of $\hat{u}_k$ in \cref{eq:lcqr estimate}, \Cref{lemma:theta distribution} and the Delta method lead to the normality result in \cref{eq:m bc t stat normality}.
 	
 	Similarly, we have
 	\begin{equation} \label{eq:m bc t stat normality minus}
 	    \frac{\hat{m}_{Y_-} - \frac{1}{2} a_{Y_-}(c)m_{Y_-}^{(2)} h_{Y_-}^2 - m_{Y_-}}{\sqrt{\text{Var}(\hat{m}_{Y_-})}} \overset{d}{\rightarrow} N(0,1).
 	\end{equation}
 	
 	Let $\tau_0 = m_{Y_+} - m_{Y_-}$. Using the proof for \cref{eq:m bc t stat normality,eq:m bc t stat normality minus}, we can show
 	\begin{align} \label{eq:m bc t stat 1}
 	    \frac{\hat{\tau}_{\text{sharp}} - \left[\frac{1}{2} a_{Y_+}(c)m_{Y_+}^{(2)} h_{Y_+}^2 - \frac{1}{2} a_{Y_-}(c)m_{Y_-}^{(2)} h_{Y_-}^2\right] - \tau_0}{\sqrt{\text{Var}(\hat{\tau}_{\text{sharp}})}} &= \overset{d}{\rightarrow} N(0,1).
 	\end{align}
 	
 	Finally, we have
 	
 \begin{align}
 	    \frac{\hat{\tau}_{\text{sharp}} - \widehat{\text{Bias}}({\hat{\tau}_{\text{sharp}}}) - \tau_0}{\sqrt{\text{Var}(\hat{\tau}_{\text{sharp}}^{\text{bc}})}} &= \frac{\hat{\tau}_{\text{sharp}} - \widehat{\text{Bias}}({\hat{\tau}_{\text{sharp}}}) - E(\hat{\tau}_{\text{sharp}} - \widehat{\text{Bias}}({\hat{\tau}_{\text{sharp}}}))}{\sqrt{\text{Var}(\hat{\tau}_{\text{sharp}}^{\text{bc}})}} \nonumber \\
 	    &\quad + \frac{ E(\hat{\tau}_{\text{sharp}} - \widehat{\text{Bias}}({\hat{\tau}_{\text{sharp}}})) - \tau_0}{\sqrt{\text{Var}(\hat{\tau}_{\text{sharp}}^{\text{bc}})}} \nonumber \\
 	    &= \frac{\hat{\tau}_{\text{sharp}}^{\text{bc}} - E(\hat{\tau}_{\text{sharp}}^{\text{bc}})}{\sqrt{\text{Var}(\hat{\tau}_{\text{sharp}}^{\text{bc}})}} + O_p(\sqrt{n_+h_+^7}) + O_p(\sqrt{n_-h_-^7}) \nonumber \\
 	    &\overset{d}{\rightarrow} N(0,1),
\end{align}
where we use the proof similar to \cref{eq:m bc t stat normality,eq:m bc t stat normality minus,eq:m bc t stat 1} and the fact that $E(\hat{\tau}_{\text{sharp}} - \widehat{\text{Bias}}({\hat{\tau}_{\text{sharp}}})) - \tau_0 = O_p(h_{Y_+}^3) + O_p(h_{Y_-}^3)$.
 	
 \end{proof}

\begin{proof} [Proof of \Cref{thm:fuzzy t adjusted dist}]
	We first note that all bias terms in \cref{eq:fuzzy biases} can be obtained using \Cref{lemma:bias and variance}. For terms in the adjusted variance in \cref{eq:fuzzy adj var}, $\text{Var}(\hat{m}_{Y_+})$, $\text{Var}(\hat{m}_{T_+})$, $\text{Var}(\widehat{\text{Bias}}(\hat{m}_{Y_+}))$, $\text{Var}(\widehat{\text{Bias}}(\hat{m}_{T_+}))$, $\text{Cov}(\hat{m}_{Y_+},\widehat{\text{Bias}}(\hat{m}_{Y_+}))$, and $\text{Cov}(\hat{m}_{T_+},\widehat{\text{Bias}}(\hat{m}_{T_+}))$  can be obtained using results in the proof of \Cref{thm:sharp t adjusted dist}; $\text{Cov}(\hat{m}_{Y_+},\hat{m}_{T_+})$ is obtained using \Cref{lemma:covariance}. And we list these seven terms in the following.
	\begin{align*}
		\text{Var}(\hat{m}_{Y_+}) &= \frac{1}{n_+h_{Y_+}} \frac{b_{Y_+}(c)\sigma_{\epsilon_{Y_+}}^2(0)}{f_{X_+}(0)} + o_p(\frac{1}{n_+h_{Y_+}}),\\
		\text{Var}(\hat{m}_{T_+}) &= \frac{1}{n_+h_{T_+}} \frac{b_{T_+}(c)\sigma_{\epsilon_{T_+}}^2(0)}{f_{X_+}(0)} + o_p(\frac{1}{n_+h_{T_+}}),\\
		\text{Var}(\widehat{\text{Bias}}(\hat{m}_{Y_+})) &= \frac{\sigma_{\epsilon_{Y_+}}^2(0)}{n_+ h_{Y_+} f_{X_+}(0)}a_{Y_+}^2(c)b_{Y_+}^*(c) + o_p(\frac{1}{n_+h_{Y_+}}),\\
		\text{Var}(\widehat{\text{Bias}}(\hat{m}_{T_+})) &= \frac{\sigma_{\epsilon_{T_+}}^2(0)}{n_+ h_{T_+} f_{X_+}(0)}a_{T_+}^2(c)b_{T_+}^*(c) + o_p(\frac{1}{n_+h_{T_+}}),
			\end{align*}
				\begin{align*}
		\text{Cov}(\hat{m}_{Y_+},\hat{m}_{T_+}) &=\frac{1}{n_+\sqrt{h_{Y_+}h_{T_+}}}\frac{\sigma_{\epsilon_{Y_+}}(0)\sigma_{\epsilon_{T_+}}(0)}{f_{X_+}(0)} b_{YT_+}+ o_p(\frac{1}{n_+h_{Y_+}} + \frac{1}{n_+h_{T_+}}),\\
		\text{Cov}(\hat{m}_{Y_+},\widehat{\text{Bias}}(\hat{m}_{Y_+}))&= \frac{a_{Y_+}(c)\sigma_{\epsilon_{Y_+}}^2(0)}{n_+ h_{Y_+}qf_{X_+}(0)} e_q^T (S_{Y_+}^{-1} \Sigma_{Y_+} S_{Y_+}^{-1})_{12,2} + o_p(\frac{1}{n_+h_{Y_+}}),\\
		\text{Cov}(\hat{m}_{T_+},\widehat{\text{Bias}}(\hat{m}_{T_+}))&= \frac{a_{T_+}(c)\sigma_{\epsilon_{T_+}}^2(0)}{n_+ h_{T_+}qf_{X_+}(0)} e_q^T (S_{T_+}^{-1} \Sigma_{T_+} S_{T_+}^{-1})_{12,2} + o_p(\frac{1}{n_+h_{T_+}}).
	\end{align*}
Next, we compute the remaining three covariances.
	\begin{align*}
		\text{Cov}(\widehat{\text{Bias}}(\hat{m}_{Y_+}),\widehat{\text{Bias}}(\hat{m}_{T_+})) &= \text{Cov}(\frac{1}{2}a_{Y_+}(c)\hat{m}_{Y_+}^{(2)}h_{Y_+}^2,\frac{1}{2}a_{T_+}(c)\hat{m}_{T_+}^{(2)}h_{T_+}^2)\\
		&= \frac{a_{Y_+}(c)a_{T_+}(c)}{\sqrt{n_+ h_{Y_+}}\sqrt{n_+ h_{T_+}}}\text{Cov}(\hat{v}_{2,Y_+}, \hat{v}_{2,T_+})\\
		&=\frac{a_{Y_+}(c)a_{T_+}(c)\sigma_{\epsilon_{Y_+}}(0)\sigma_{\epsilon_{T_+}}(0)}{n_+ \sqrt{h_{Y_+}h_{T_+}}f_{X_+}(0)} e_2^T (S_{Y_+}^{-1} \Sigma_{YT_+} S_{T_+}^{-1})_{22} e_2 \\
		&\quad+ o_p(\frac{1}{n_+h_{Y_+}} + \frac{1}{n_+h_{T_+}}).
		\end{align*}	
		\begin{align*}
		\text{Cov}(\hat{m}_{Y_+},\widehat{\text{Bias}}(\hat{m}_{T_+})) &=\text{Cov}(m_{Y_+} + \frac{1}{q\sqrt{n_+ h_{Y_+}}}\sum_{k=1}^{q}\hat{u}_{k,Y}, \frac{1}{2}a_{T_+}(c)h_{T_+}^2(m_{T_+}^{(2)} + \frac{2\hat{v}_{2,T}}{h_{T_+}^2\sqrt{n_+ h_{T_+}}}))\\
		&=\frac{a_{T_+}(c)\sigma_{\epsilon_{Y_+}}(0)\sigma_{\epsilon_{T_+}}(0)}{q n_+ \sqrt{h_{Y_+}h_{T_+}}f_{X_+}(0)} e_q^T (S_{Y_+}^{-1} \Sigma_{YT_+} S_{T_+}^{-1})_{12,2} + o_p(\frac{1}{n_+h_{Y_+}} + \frac{1}{n_+h_{T_+}}).
	\end{align*}	
	\begin{align*}	
\text{Cov}(\hat{m}_{T_+},\widehat{\text{Bias}}(\hat{m}_{Y_+}))&=\text{Cov}(m_{T_+} + \frac{1}{q\sqrt{n_+ h_{T_+}}}\sum_{k=1}^{q}\hat{u}_{k,T}, \frac{1}{2}a_{Y_+}(c)h_{Y_+}^2(m_{Y_+}^{(2)} + \frac{2\hat{v}_{2,Y}}{h_{Y_+}^2\sqrt{n_+ h_{Y_+}}}))\\ 
		&=\frac{a_{Y_+}(c)\sigma_{\epsilon_{Y_+}}(0)\sigma_{\epsilon_{T_+}}(0)}{q n_+ \sqrt{h_{Y_+}h_{T_+}}f_{X_+}(0)} e_q^T (S_{T_+}^{-1} \Sigma_{TY_+} S_{Y_+}^{-1})_{12,2} + o_p(\frac{1}{n_+h_{Y_+}} + \frac{1}{n_+h_{T_+}}).
	\end{align*}
	Substituting the above results into \cref{eq:fuzzy adj var} gives the expression for $\text{Var}( (\hat{m}_{Y_+} - \tau_0 \hat{m}_{T_+}) -(\widehat{\text{Bias}}(\hat{m}_{Y_+}) - \tau_0 \widehat{\text{Bias}}(\hat{m}_{T_+})) )$. The result for $\text{Var}( (\hat{m}_{Y_-} - \tau_0 \hat{m}_{T_-}) -(\widehat{\text{Bias}}(\hat{m}_{Y_-}) - \tau_0 \widehat{\text{Bias}}(\hat{m}_{T_-})) )$ can be obtained in a similar way. Adding up the two variance results gives the adjusted variance in the fuzzy case.	
	
	To establish the asymptotic normality, note that we can use \cref{eq:fuzzy biases} to write $\tilde{\tau}_{\text{fuzzy}}^{\text{bc}}$ as
	\begin{equation*}
		\tilde{\tau}_{\text{fuzzy}}^{\text{bc}} = (\hat{m}_{Y_+} - \widehat{\text{Bias}}(\hat{m}_{Y_+})) - \tau_0 (\hat{m}_{T_+} - \widehat{\text{Bias}}(\hat{m}_{T_+})) - (\hat{m}_{Y_-}-\widehat{\text{Bias}}(\hat{m}_{Y_-})) + \tau_0 (\hat{m}_{T_-}- \widehat{\text{Bias}}(\hat{m}_{T_-})).
	\end{equation*}
	Using the similar argument in proving the asymptotic normality of $\hat{\tau}_{\text{sharp}}^{\text{bc}}$, we can establish the asymptotic distribution of $t_{\text{fuzzy}}^{\text{adj.}}$.
\end{proof}

\begin{proof}[Proof of \Cref{thm:adjusted bandwidth}]
	We first expand $\text{Bias}(\hat{m}_{Y_+})$ up to $O(h_{Y_+}^3)$ on the boundary. Recall $\hat{m}_{Y_+} = \sum_{k=1}^{q} \hat{a}_k/q$ and we have
	\begin{align*}
		\text{Bias}(\hat{m}_{Y_+}) &= \frac{\sigma_{\epsilon_{Y_+}}(0)}{q}\sum_{k=1}^{q}c_k - \frac{\sigma_{\epsilon_{Y_+}}(0)}{q\sqrt{n_+ h_{Y_+}} f_{X_+}(0)} e_q^T \left[(S_{Y_+}^{-1}(c))_{11} E(w_{Y_+,1n}^*) +(S_{Y_+}^{-1}(c))_{12} E(w_{Y_+,2n}^*) \right]\\
		&=- \frac{\sigma_{\epsilon_{Y_+}}(0)}{q\sqrt{n_+ h_{Y_+}} f_{X_+}(0)} e_q^T  (S_{Y_+}^{-1}(c))_{11} E(w_{Y_+,1n}^*) - \frac{\sigma_{\epsilon_{Y_+}}(0)}{q\sqrt{n_+ h_{Y_+}} f_{X_+}(0)} e_q^T (S_{Y_+}^{-1}(c))_{12} E(w_{Y_+,2n}^*) \\
		&= \textrm{I} + \textrm{II}.
	\end{align*}
	Consider term $\textrm{I}$.
	\begin{align*}
		\textrm{I}&=  \frac{-\sigma_{\epsilon_{Y_+}}(0)}{q\sqrt{n_+ h_{Y_+}} f_{X_+}(0)} e_q^T  (S_{Y_+}^{-1}(c))_{11} 
		\begin{pmatrix}
		\frac{1}{\sqrt{n_+ h_{Y_+}}} \sum_{i=1}^{n_+}K_i E(\eta_{Y_+,i,1}^*)\\
		\vdots\\
		\frac{1}{\sqrt{n_+ h_{Y_+}}} \sum_{i=1}^{n_+}K_i E(\eta_{Y_+,i,q}^*)
		\end{pmatrix} 
		\\
		&= \frac{\sigma_{\epsilon_{Y_+}}(0)}{q\sqrt{n_+ h_{Y_+}} f_{X_+}(0)} e_q^T  (S_{Y_+}^{-1}(c))_{11} 
		\begin{pmatrix}
		\frac{f_{\epsilon_{Y_+}}(c_1)}{\sqrt{n_+ h_{Y_+}}} \sum_{i=1}^{n_+}K_i \frac{ d_{i,1}}{\sigma_{\epsilon_{Y_+,i}}}\\
		\vdots\\
		\frac{f_{\epsilon_{Y_+}}(c_q)}{\sqrt{n_+ h_{Y_+}}} \sum_{i=1}^{n_+}K_i \frac{ d_{i,q}}{\sigma_{\epsilon_{Y_+,i}}}
		\end{pmatrix} + o_p(1)
		\\
		&=	\frac{1}{q\sqrt{n_+ h_{Y_+}} f_{X_+}(0)} e_q^T  (S_{Y_+}^{-1}(c))_{11} 
		\begin{pmatrix}
		\frac{f_{\epsilon_{Y_+}}(c_1)}{\sqrt{n_+ h_{Y_+}}} \sum_{i=1}^{n_+}K_i r_{i,1}\\
		\vdots\\
		\frac{f_{\epsilon_{Y_+}}(c_q)}{\sqrt{n_+ h_{Y_+}}} \sum_{i=1}^{n_+}K_i r_{i,1}
		\end{pmatrix} + o_p(1)
		\\	
		&= \frac{1}{2q} e_q^T  (S_{Y_+}^{-1}(c))_{11} f_{\epsilon_{Y_+}} m_{Y_+}^{(2)} \mu_{+,2}  h_{Y_+}^2 + \frac{1}{6q} e_q^T  (S_{Y_+}^{-1}(c))_{11} f_{\epsilon_{Y_+}} m_{Y_+}^{(3)} \mu_{+,3}  h_{Y_+}^3 \\
		&\quad+ \frac{f_{X_+}^{(1)}(0)}{2qf_{X_+}(0)} e_q^T  (S_{Y_+}^{-1}(c))_{11} f_{\epsilon_{Y_+}} m_{Y_+}^{(2)} \mu_{+,3} h_{Y_+}^3 + o_p(h_{Y_+}^3),
	\end{align*}
	where the second equality follows by expanding the cumulative distribution  of $\epsilon_{Y_+,i}$ around $c_k$ and the third equality follows by noticing that all terms containing $c_k$, after multiplied by the coefficient $\frac{\sigma_{\epsilon_{Y_+}}(0)}{q\sqrt{n_+ h_{Y_+}} f_{X_+}(0)} e_q^T  (S_{Y_+}^{-1}(c))_{11}$, become zero after a summation. The last equality is obtained by a Taylor series expansion of $m_{Y_+}$ at $0$ up to order $3$ in $r_{i,1}$, similar to the expansion in the definition of $r_{i,3}$.

	Consider term $\textrm{II}$. Note that $p = 1$ in the following proof when we estimate the conditional mean using degree one local polynomial.
	\begin{align*}
		\textrm{II} &= \frac{-\sigma_{\epsilon_{Y_+}}(0)}{q\sqrt{n_+ h_{Y_+}} f_{X_+}(0)} e_q^T  (S_{Y_+}^{-1}(c))_{12} 
		\begin{pmatrix}
		\frac{1}{\sqrt{n_+ h_{Y_+}}} \sum_{k=1}^{q} \sum_{i=1}^{n_+}K_i X_{+,i} E(\eta_{Y_+,i,1}^*)\\
		\vdots\\
		\frac{1}{\sqrt{n_+ h_{Y_+}}} \sum_{k=1}^{q}\sum_{i=1}^{n_+}K_i X_{+,i}^p E(\eta_{Y_+,i,q}^*)
		\end{pmatrix} 
		\\
		&= \frac{\sigma_{\epsilon_{Y_+}}(0)}{q\sqrt{n_+ h_{Y_+}} f_{X_+}(0)} e_q^T  (S_{Y_+}^{-1}(c))_{12} 
		\begin{pmatrix}
		\frac{1}{\sqrt{n_+ h_{Y_+}}} \sum_{k=1}^{q} f_{\epsilon_{Y_+}}(c_1) \sum_{i=1}^{n_+}K_i X_{+,i} \frac{ d_{i,1}}{\sigma_{\epsilon_{Y_+,i}}}\\
		\vdots\\
		\frac{1}{\sqrt{n_+ h_{Y_+}}} \sum_{k=1}^{q} f_{\epsilon_{Y_+}}(c_q) \sum_{i=1}^{n_+}K_i X_{+,i}^p \frac{ d_{i,q}}{\sigma_{\epsilon_{Y_+,i}}}
		\end{pmatrix} + o_p(1)\\
		&= \frac{\sum_{k=1}^{q}f_{\epsilon_{Y_+}}(c_k)}{2q} e_q^T  (S_{Y_+}^{-1}(c))_{12} \begin{pmatrix}
		\mu_{+,3} \\
		\vdots    \\
		\mu_{+,p+2}
		\end{pmatrix} m_{Y_+}^{(2)} h_{Y_+}^2 \\
		&\quad + \frac{\sum_{k=1}^{q}f_{\epsilon_{Y_+}}(c_k)}{6q} e_q^T  (S_{Y_+}^{-1}(c))_{12} \begin{pmatrix}
		\mu_{+,4} \\
		\vdots    \\
		\mu_{+,p+3}
		\end{pmatrix} m_{Y_+}^{(3)} h_{Y_+}^3 \\
		&\quad + \frac{f_{X_+}^{(1)}(0) \sum_{k=1}^{q}f_{\epsilon_{Y_+}}(c_k)}{2qf_{X_+}(0)} e_q^T  (S_{Y_+}^{-1}(c))_{12} \begin{pmatrix}
		\mu_{+,4} \\
		\vdots    \\
		\mu_{+,p+3}
		\end{pmatrix} m_{Y_+}^{(2)} h_{Y_+}^3 + o_p(h_{Y_+}^3).
	\end{align*}
	Combining \textrm{I} and \textrm{II} yields
	\begin{align*}
		\text{Bias}(\hat{m}_{Y_+}) &= \frac{1}{2} a_{Y_+}(c) m_{Y_+}^{(2)} h_{Y_+}^2 + \frac{1}{6}\check{a}_{Y_+}(c) m_{Y_+}^{(3)} h_{Y_+}^3 +\frac{1}{2}\frac{\tilde{a}_{Y_+}(c) f_{X_+}^{(1)}(0)}{f_{X_+}(0)}m_{Y_+}^{(2)} h_{Y_+}^3 + o_p(h_{Y_+}^3),
	\end{align*}
	where $a_{Y_+}(c)$ is  in \Cref{lemma:bias and variance}, $\check{a}_{Y_+}(c) = \frac{\mu_{+,2}(c)\mu_{+,3}(c) - \mu_{+,1}(c)\mu_{+,4}(c)}{\mu_{+,0}(c)\mu_{+,2}(c) - \mu_{+,1}^2(c)}$, and $
		\tilde{a}_{Y_+}(c) = \frac{\mu_{+,2}^2(c) - \mu_{+,1}(c)\mu_{+,4}(c)}{\mu_{+,0}(c)\mu_{+,2}(c) - \mu_{+,1}^2(c)}$. Hence the leading term in $\text{Bias}(\hat{m}_{Y_+}-\text{Bias}(\hat{m}_{Y_+}))$ is  $\frac{1}{6}\check{a}_{Y_+}(c) m_{Y_+}^{(3)} h_{Y_+}^3 +  \frac{1}{2}\frac{\tilde{a}_{Y_+}(c) f_{X_+}^{(1)}(0)}{f_{X_+}(0)}m_{Y_+}^{(2)} h_{Y_+}^3$.
	
	Since we work with data above the cutoff in this proof, the adjusted variance is given by $\frac{1}{n_+ h_{Y_+}}V_{\text{sharp}}^{\text{adj.}}$ in the proof of \Cref{thm:sharp t adjusted dist}. Thus, the adjusted MSE can be written as
    \begin{align*}
    	\text{adj. MSE} &= \left[\frac{1}{6}\check{a}_{Y_+}(c) m_{Y_+}^{(3)} h_{Y_+}^3+\frac{1}{2}\frac{\tilde{a}_{Y_+}(c) f_{X_+}^{(1)}(0)}{f_{X_+}(0)}m_{Y_+}^{(2)} h_{Y_+}^3\right]^2 + \frac{1}{n_+ h_{Y_+}}V_{\text{sharp}}^{\text{adj.}} + o_p(h_{Y_+}^6 + \frac{1}{n_+ h_{Y_+}})\\
    	&= C_2^2 h_{Y_+}^6 + \frac{1}{n_+ h_{Y_+}} C_3 + o_p(h_{Y_+}^6 + \frac{1}{n_+ h_{Y_+}}),
    \end{align*}
    where $C_2 = \frac{1}{6}\check{a}_{Y_+}(c) m_{Y_+}^{(3)} +\frac{1}{2}\frac{\tilde{a}_{Y_+}(c) f_{X_+}^{(1)}(0)}{f_{X_+}(0)}m_{Y_+}^{(2)}$ and $C_3 = V_{\text{sharp}}^{\text{adj.}}$.  The bandwidth that minimizes the adjusted MSE is given by $h = \left(\frac{C_3}{6C_2^2}\right)^{1/7}n_+^{-1/7}$.
\end{proof}

\subsection{Lemmas and propositions for fixed-n results} \label{supp:fixed-n}

This section first collects several lemmas for the development of fixed-$n$ approximations. We then present two propositions that are the fixed-$n$ counterparts of \Cref{thm:sharp t adjusted dist,thm:fuzzy t adjusted dist}. Assume $p = 1$ in the following lemmas.

\begin{lemma} \label{lemma:fixed-n bias and variance}
	Under  \cref{assumption:m continuity,assumption:sigma continuity,assumption:kernel,assumption:density x,assumption:error distribution,assumption:bandwidth}, the fixed-$n$ bias and variance are given by
	\begin{align}
	\text{Bias}(\hat{m}_{Y_+}|\mathbf{X})_{\text{fixed-n}}&= \frac{1}{q} e_q^T \Bigg[(S_{nY_+}^{-1})_{11} f_{\epsilon_{Y_+}} \frac{1}{2 n_+ h_{Y_+}} \sum_{i=1}^{n_+} \frac{K_{+,i}x_{+,i}^2}{\sigma_{\epsilon_{Y_+,i}}} \nonumber \\
	&\quad + (S_{nY_+}^{-1})_{12} \sum_{k=1}^{q}f_{\epsilon_{Y_+}}(c_k) \frac{1}{2 n_+ h_{Y_+}} \sum_{i=1}^{n_+} \frac{K_{+,i}x_{+,i}^3}{\sigma_{\epsilon_{Y_+,i}}}   \Bigg] m_{Y_+}^{(2)} h_{Y_+}^2 + o_p(h_{Y_+}^2),\nonumber \\
	\text{Var}(\hat{m}_{Y_+}|\mathbf{X})_{\text{fixed-n}} &=  \frac{1}{n_+h_{Y_+}q^2} e_q^T \left(S_{nY_+}^{-1} \Sigma_{nY_+} S_{nY_+}^{-1}\right)e_q + o_p(\frac{1}{n_+h_{Y_+}}). \nonumber	
	\end{align}
\end{lemma}

\begin{proof}[Proof of \Cref{lemma:fixed-n bias and variance}]
	We first state some results for $E(w_{Y_+,1n}^*)$ and $E(w_{Y_+,2n}^*)$ that are used in the proof of the asymptotic results in \Cref{thm:adjusted bandwidth} in \Cref{supp:asymptotic}. $E(w_{Y_+,1n}^*)$ is a $q \times 1$ vector while $E(w_{Y_+,2n}^*)$ is a $p \times 1$ vector with $p=1$ in this case. By not letting $n\rightarrow \infty$, we have
	\begin{equation} \label{eq:Ew}
		\begin{aligned} 
		E(w_{Y_+,1n}^*) &= -f_{\epsilon_{Y_+}} \frac{1}{2 \sqrt{n_+ h_{Y_+}}} \sum_{i=1}^{n_+} \frac{K_{+,i}x_{+,i}^2}{\sigma_{\epsilon_{Y_+,i}}} m_{Y_+}^{(2)}(0) h_{Y_+}^2 + o(1),\\
		E(w_{Y_+,2n}^*) &=-\sum_{k=1}^{q}f_{\epsilon_{Y_+}}(c_k) \frac{1}{2 \sqrt{n_+ h_{Y_+}}} \sum_{i=1}^{n_+} \frac{K_{+,i}x_{+,i}^3}{\sigma_{\epsilon_{Y_+,i}}} m_{Y_+}^{(2)}(0) h_{Y_+}^2 + o(1).
		\end{aligned}
	\end{equation}
	From the proof of Theorem 5, Lemmas 2 and 3 in \cite{kai2010local}, the loss function becomes
	\begin{equation*}
		L_{n_+}(\theta) = \frac{1}{2} \theta^T S_{nY_+} \theta + W_{Y_+,n_+}^* + o_p(1),
	\end{equation*} 
	the solution of which is,
	\begin{equation*}
		\hat{\theta}_{n_+} = - S_{nY_+}^{-1} W_{Y_+,n_+}^* + o_p(1),
	\end{equation*}
	Rewrite the above equation as
	\begin{equation} \label{eq:fixed-n theta}
		\hat{\theta}_{n_+} + S_{nY_+}^{-1} E(W_{Y_+,n_+}^*|\textbf{X}) = - S_{nY_+}^{-1} \left[W_{Y_+,n_+}^* - E(W_{Y_+,n_+}^*|\textbf{X})\right] + o_p(1),
	\end{equation}
	which is the base to prove \Cref{lemma:theta distribution} in the asymptotic case. Combining \cref{eq:lcqr estimate} and \cref{eq:fixed-n theta}, we obtain the following expression for pre-asymptotic bias
	\begin{align}
		\hat{m}_{Y_+} - m_{Y_+} &= \frac{1}{q \sqrt{n_+ h_{Y_+}}} \sum_{k=1}^{q} \hat{u}_k \label{eq:fixed_n mY bias form 1} \\
		&= -\frac{1}{q \sqrt{n_+ h_{Y_+}}} e_q^T \begin{bmatrix}
		(S_{nY_+}^{-1})_{11} & (S_{nY_+}^{-1})_{12}
		\end{bmatrix} W_{Y_+,n_+}^*. \label{eq:fixed-n mY bias}
	\end{align}
	Plug the result in \cref{eq:Ew} into \cref{eq:fixed-n mY bias}, and we prove the fixed-$n$ bias result.
	From \cref{eq:fixed-n theta}, the variance of $\hat{\theta}_{n_+}$ becomes 
	\begin{align*}
		\text{Var}(\hat{\theta}_{n_+}) &= S_{nY_+}^{-1} \text{Var}(W_{Y_+,n_+}^* - E(W_{Y_+,n_+}^*|\textbf{X})) S_{nY_+}^{-1} \\
		&\rightarrow S_{nY_+}^{-1} \text{Var}(W_{Y_+,n_+} - E(W_{Y_+,n_+}|\textbf{X})) S_{nY_+}^{-1}\\
		&= S_{nY_+}^{-1} \text{Var}(W_{Y_+,n_+}) S_{nY_+}^{-1},
	\end{align*}
	where we use the result $\text{Var}(W_{Y_+,n_+}^* - W_{Y_+,n_+}|\textbf{X}) = o_p(1)$ from the proof of Theorem 5 in \cite{kai2010local}. Similar to the proof of \Cref{lemma:covariance}, we can show $\text{Var}(W_{Y_+,n_+}) = \Sigma_{nY_+}$, which, together with \cref{eq:lcqr estimate}, proves the variance result in this lemma.
\end{proof} 

\begin{lemma} \label{lemma:fixed-n covariance}
Under  \cref{assumption:m continuity,assumption:sigma continuity,assumption:kernel,assumption:density x,assumption:error distribution,assumption:bandwidth}, the fixed-$n$ covariance between $\hat{m}_{Y_+}(0)$ and $\hat{m}_{T_+}(0)$ is given by
\begin{equation*}
	\text{Cov}(\hat{m}_{Y_+},\hat{m}_{T_+}|\textbf{X})_{\text{fixed-n}} = \frac{1}{q^2 n_+ \sqrt{ h_{Y_+} h_{T_+}}} e_q^T\begin{bmatrix}
	(S_{nY_+}^{-1})_{11} & (S_{nY_+}^{-1})_{12}
	\end{bmatrix} \Sigma_{nYT_+} \begin{bmatrix}
	(S_{nT_+}^{-1})_{11} & (S_{nT_+ }^{-1})_{12}
	\end{bmatrix}^T e_q.
\end{equation*}
\end{lemma}

\begin{proof}[Proof of \Cref{lemma:fixed-n covariance}]
	Similar to \cref{eq:fixed-n mY bias}, we have
	\begin{equation} \label{eq:fixed-n mT bias}
		\hat{m}_{T_+} - m_{T_+} = -\frac{1}{q \sqrt{n_+ h_{T_+}}} e_q^T \begin{bmatrix}
		(S_{nT_+}^{-1})_{11} & (S_{nT_+}^{-1})_{12}
		\end{bmatrix} W_{T_+,n_+}^*.
	\end{equation}
	Using the proof similar to that in \Cref{lemma:covariance} and the result
	\begin{equation*}
		\Sigma_{nYT_+} = E \left[\begin{pmatrix}
		w_{Y_+,1n}^* - E(w_{Y_+,1n}^*|\mathbf{X})\\
		w_{Y_+,2n}^* - E(w_{Y_+,2n}^*|\mathbf{X})
		\end{pmatrix} 
		\begin{pmatrix}
		w_{T_+,1n}^* - E(w_{T_+,1n}^*|\mathbf{X})\\
		w_{T_+,2n}^* - E(w_{T_+,2n}^*|\mathbf{X})
		\end{pmatrix}^T
		\right],
	\end{equation*} we have that \Cref{lemma:fixed-n covariance} holds.
\end{proof}

\begin{lemma} \label{lemma:fixed-n sharp results}
	Under \cref{assumption:m continuity,assumption:sigma continuity,assumption:kernel,assumption:density x,assumption:error distribution,assumption:bandwidth}, we have
	\begin{align*}
		\text{Bias}(\hat{\tau}_{\text{sharp}}|\mathbf{X})_{\text{fixed-n}} &= \text{Bias}(\hat{m}_{Y_+})_{\text{fixed-n}}  - \text{Bias}(\hat{m}_{Y_-})_{\text{fixed-n}} + o_p(h_{Y_+}^2 + h_{Y_-}^2),\\
		\text{Var}(\hat{\tau}_{\text{sharp}}|\mathbf{X})_{\text{fixed-n}} &=\text{Var}(\hat{m}_{Y_+})_{\text{fixed-n}} + \text{Var}(\hat{m}_{Y_-})_{\text{fixed-n}} + o_p(\frac{1}{n_+ h_{Y_+}} + \frac{1}{n_- h_{Y_-}}),
	\end{align*} 
	where $\text{Bias}(\hat{m}_{Y_+})_{\text{fixed-n}}$ and $\text{Var}(\hat{m}_{Y_+})_{\text{fixed-n}}$ are given in \Cref{lemma:fixed-n bias and variance}, and $\text{Bias}(\hat{m}_{Y_-})_{\text{fixed-n}}$ and $\text{Var}(\hat{m}_{Y_-})_{\text{fixed-n}}$ are defined analogously.
\end{lemma}

\begin{proof}[Proof of \Cref{lemma:fixed-n sharp results}]
	The results hold by applying \Cref{lemma:fixed-n bias and variance} to \cref{eq:sharp treatment}.
\end{proof}

\begin{lemma} \label{lemma:fixed-n fuzzy results}
	Under \cref{assumption:m continuity,assumption:sigma continuity,assumption:kernel,assumption:density x,assumption:error distribution,assumption:bandwidth}, we have
	\begin{align}
	\text{Bias}(\hat{\tau}_{\text{fuzzy}}|\mathbf{X})_{\text{fixed-n}} &= \frac{1}{m_{T_+} - m_{T_-}}\left[\text{Bias}(\hat{m}_{Y_+})_{\text{fixed-n}}  - \text{Bias}(\hat{m}_{Y_-})_{\text{fixed-n}}\right] \nonumber \\
	&\quad - \frac{m_{Y_+} - m_{Y_-}}{\left[m_{T_+} - m_{T_-}\right]^2}\left[\text{Bias}(\hat{m}_{T_+})_{\text{fixed-n}}  - \text{Bias}(\hat{m}_{T_-})_{\text{fixed-n}}\right]\nonumber\\
	&\quad + o_p(h_{Y_+}^2 + h_{Y_-}^2 + h_{T_+}^2 + h_{T_-}^2).
	\end{align}
	The variance expression is given in (\ref{fuzzy variance long form}) by substituting the results for $\text{Var}(\hat{m}_{Y_+})$, $\text{Var}(\hat{m}_{Y_-})$, $\text{Var}(\hat{m}_{T_+})$, $\text{Var}(\hat{m}_{T_-})$, $\text{Cov}(\hat{m}_{Y_+}, \hat{m}_{T_+})$ and $\text{Cov}(\hat{m}_{Y_-}, \hat{m}_{T_-})$ in \Cref{lemma:fixed-n bias and variance,lemma:fixed-n covariance}.
\end{lemma}

\begin{proof}[Proof of \Cref{lemma:fixed-n fuzzy results}]
	The proof follows from  (\ref{fuzzy variance long form}) and \Cref{lemma:fixed-n bias and variance,lemma:fixed-n covariance}.
\end{proof}

Let $p=2$ in the following lemma.
\begin{lemma} \label{lemma:fixed-n bias and var of bias}
	Under \cref{assumption:m continuity,assumption:sigma continuity,assumption:kernel,assumption:density x,assumption:error distribution,assumption:bandwidth}, the fixed-$n$ variance of $\hat{m}_{Y_+}^{(2)}$ is given by
	\begin{equation*}
	\text{Var }(\hat{m}_{Y_+}^{(2)}|\mathbf{X}) = \frac{4}{n_+ h_{Y_+}^5} e_2^T(S_{nY_+}^{-1} \Sigma_{nY_+} S_{nY_+}^{-1})_{22} e_2 + o_p(\frac{1}{n_+ h_{Y_+}^5}). \label{eq:fixed-n variance of m2d}
	\end{equation*}
\end{lemma}

\begin{proof}[Proof of \Cref{lemma:fixed-n bias and var of bias}]
	It reslts from combining $\text{Var}(\hat{v}_2) = e_2^T(S_{nY_+}^{-1} \Sigma_{nY_+} S_{nY_+}^{-1})_{22} e_2$ and \cref{eq:md2 and v}.
\end{proof}

\setcounter{theorem}{0}

\begin{proposition} \label{prop:fixed-n sharp t}
	Under \cref{assumption:m continuity,assumption:sigma continuity,assumption:kernel,assumption:density x,assumption:error distribution,assumption:bandwidth}, the fixed-$n$ adjusted \textit{t}-statistic for the sharp RD is given by
	\begin{equation} \label{eq:fixed-n sharp t adjusted dist}
	t_{\text{sharp, fixed-n}}^{\text{adj.}} = \frac{\hat{\tau}_{\text{sharp}} - \widehat{\text{Bias}}(\hat{\tau}_{\text{sharp}})_{\text{fixed-n}}-\tau_0}{\sqrt{\text{Var}(\hat{\tau}_{\text{sharp}} - \widehat{\text{Bias}}(\hat{\tau}_{\text{sharp}})_{\text{fixed-n}})_{\text{fixed-n}}}},
	\end{equation} 
	where the expression for fixed-$n$ terms, $\widehat{\text{Bias}}(\hat{\tau}_{\text{sharp}})_{\text{fixed-n}}$, $\text{Var}(\hat{\tau}_{\text{sharp}} - \widehat{\text{Bias}}(\hat{\tau}_{\text{sharp}})_{\text{fixed-n}})_{\text{fixed-n}}$, are given in the proof of this proposition.
\end{proposition}

\begin{proof}[Proof of \Cref{prop:fixed-n sharp t}]
	The fixed-$n$ bias term on the numerator of \cref{eq:fixed-n sharp t adjusted dist} is given in \Cref{lemma:fixed-n sharp results}. For the denominator of \cref{eq:fixed-n sharp t adjusted dist}, recall
	\begin{align}
		\text{Var}(\hat{\tau}_{\text{sharp}} - \widehat{\text{Bias}}(\hat{\tau}_{\text{sharp}})_{\text{fixed-n}})_{\text{fixed-n}} &= \text{Var}(\hat{\tau}_{\text{sharp}}) + \text{Var}( \widehat{\text{Bias}}(\hat{\tau}_{\text{sharp}})_{\text{fixed-n}}) \nonumber \\
		&\quad - 2 \text{Cov}(\hat{\tau}_{\text{sharp}}, \widehat{\text{Bias}}(\hat{\tau}_{\text{sharp}})_{\text{fixed-n}}) \nonumber\\
		&= \textrm{I} + \textrm{II} + \textrm{III}.
	\end{align}
	Term $\textrm{I}$ is given in \Cref{lemma:fixed-n bias and variance}. Consider the second term $\textrm{II}$.
	\begin{align} \label{eq:fixed-n var of sharp treatment}
		\text{Var}( \widehat{\text{Bias}}(\hat{\tau}_{\text{sharp}})_{\text{fixed-n}}) &=\text{Var}( \widehat{\text{Bias}}(\hat{m}_{Y_+})_{\text{fixed-n}}) + \text{Var}( \widehat{\text{Bias}}(\hat{m}_{Y_-})_{\text{fixed-n}}).
	\end{align}
	Using \Cref{lemma:fixed-n bias and variance} and omitting the small-order terms, we have
	\begin{align}
		\widehat{\text{Bias}}(\hat{m}_{Y_+})_{\text{fixed-n}} &= D_{nY_+,1} \hat{m}_{Y_+}^{(2)} h_{Y_+}^2,\label{eq:fixed-n bias for mY+}\\
		\widehat{\text{Bias}}(\hat{m}_{Y_-})_{\text{fixed-n}} &= D_{nY_-,1} \hat{m}_{Y_-}^{(2)} h_{Y_-}^2,\label{eq:fixed-n bias for mY-}
	\end{align}
	where
	
	\begin{align}
		&D_{nY_+,1} \label{eq:DnY+} \\
		=& \frac{1}{q} e_q^T \Bigg[(S_{nY_+}^{-1})_{11} f_{\epsilon_{Y_+}} \frac{1}{2 n_+ h_{Y_+}} \sum_{i=1}^{n_+} \frac{K_{+,i}x_{+,i}^2}{\sigma_{\epsilon_{Y_+,i}}} + (S_{nY_+}^{-1})_{12} \sum_{k=1}^{q}f_{\epsilon_{Y_+}}(c_k) \frac{1}{2 n_+ h_{Y_+}} \sum_{i=1}^{n_+} \frac{K_{+,i}x_{+,i}^3}{\sigma_{\epsilon_{Y_+,i}}}   \Bigg],\nonumber \\
		&D_{nY_-,1} \label{eq:DnY-}\\
		=& \frac{1}{q} e_q^T \Bigg[(S_{nY_-}^{-1})_{11} f_{\epsilon_{Y_-}} \frac{1}{2 n_- h_{Y_-}} \sum_{i=1}^{n_-} \frac{K_{-,i}x_{-,i}^2}{\sigma_{\epsilon_{Y_-,i}}}  + (S_{nY_-}^{-1})_{12} \sum_{k=1}^{q}f_{\epsilon_{Y_-}}(c_k) \frac{1}{2 n_- h_{Y_-}} \sum_{i=1}^{n_-} \frac{K_{-,i}x_{-,i}^3}{\sigma_{\epsilon_{Y_-,i}}}   \Bigg].\nonumber
	\end{align}
	Applying \Cref{lemma:fixed-n bias and var of bias} to  $\widehat{\text{Bias}}(\hat{m}_{Y_+})_{\text{fixed-n}}$, $\widehat{\text{Bias}}(\hat{m}_{Y_-})_{\text{fixed-n}}$ for \cref{eq:fixed-n var of sharp treatment}:
	\begin{equation*}
		\textrm{II} = \frac{4}{n_+ h_{Y_+}} D_{nY_+,1}^2 e_2^T (S_{nY_+}^{-1} \Sigma_{nY_+} S_{nY_+}^{-1})_{22} e_2 + \frac{4}{n_- h_{Y_-}} D_{nY_-,1}^2 e_2^T (S_{nY_-}^{-1} \Sigma_{nY_-} S_{nY_-}^{-1})_{22} e_2.
	\end{equation*}
	For term $\textrm{III}$ with \textit{i.i.d.} errors, we have
	\begin{equation*}
		\text{Cov}(\hat{\tau}_{\text{sharp}}, \widehat{\text{Bias}}(\hat{\tau}_{\text{sharp}})_{\text{fixed-n}}) = \text{Cov}(\hat{m}_{Y_+}, \widehat{\text{Bias}}(\hat{m}_{Y_+})_{\text{fixed-n}}) + \text{Cov}(\hat{m}_{Y_-}, \widehat{\text{Bias}}(\hat{m}_{Y_-})_{\text{fixed-n}}).
	\end{equation*}
	Using \cref{eq:md2 and v} and \cref{eq:fixed_n mY bias form 1}, we can show
	\begin{align*}
		\text{Cov}(\hat{m}_{Y_+}, \widehat{\text{Bias}}(\hat{m}_{Y_+})_{\text{fixed-n}}) &= \frac{2 D_{nY_+,1}}{q n_+ h_{Y_+}} e_2^T (S_{nY_+}^{-1} \Sigma_{nY_+} S_{nY_+}^{-1})_{12,2}, \\
		\text{Cov}(\hat{m}_{Y_-}, \widehat{\text{Bias}}(\hat{m}_{Y_-})_{\text{fixed-n}}) &= \frac{2 D_{nY_-,1}}{q n_- h_{Y_-}} e_2^T (S_{nY_-}^{-1} \Sigma_{nY_-} S_{nY_-}^{-1})_{12,2}.
	\end{align*}
	Putting everything together, the numerator in \cref{eq:fixed-n sharp t adjusted dist} becomes
	\begin{align*}
		\hat{\tau}_{\text{sharp}} - \widehat{\text{Bias}}(\hat{\tau}_{\text{sharp}})_{\text{fixed-n}} &=\hat{\tau}_{\text{sharp}} - \Bigg\{ \frac{1}{q} e_q^T \Bigg[(S_{nY_+}^{-1})_{11} f_{\epsilon_{Y_+}} \frac{1}{2 n_+ h_{Y_+}} \sum_{i=1}^{n_+} \frac{K_{+,i}x_{+,i}^2}{\sigma_{\epsilon_{Y_+,i}}} \nonumber \\
		&\quad + (S_{nY_+}^{-1})_{12} \sum_{k=1}^{q}f_{\epsilon_{Y_+}}(c_k) \frac{1}{2 n_+ h_{Y_+}} \sum_{i=1}^{n_+} \frac{K_{+,i}x_{+,i}^3}{\sigma_{\epsilon_{Y_+,i}}}   \Bigg] m_{Y_+}^{(2)} h_{Y_+}^2 \\
		&\quad - \frac{1}{q} e_q^T \Bigg[(S_{nY_-}^{-1})_{11} f_{\epsilon_{Y_-}} \frac{1}{2 n_- h_{Y_-}} \sum_{i=1}^{n_-} \frac{K_{-,i}x_{-,i}^2}{\sigma_{\epsilon_{Y_-,i}}} \nonumber \\
		&\quad + (S_{nY_-}^{-1})_{12} \sum_{k=1}^{q}f_{\epsilon_{Y_-}}(c_k) \frac{1}{2 n_- h_{Y_-}} \sum_{i=1}^{n_-} \frac{K_{-,i}x_{-,i}^3}{\sigma_{\epsilon_{Y_-,i}}}   \Bigg] m_{Y_-}^{(2)} h_{Y_-}^2 \Bigg \}.
	\end{align*}
	The variance on the denominator is given by
	\begin{align*}
		&\quad \text{Var}(\hat{\tau}_{\text{sharp}} - \widehat{\text{Bias}}(\hat{\tau}_{\text{sharp}})_{\text{fixed-n}}) \\
		&=\frac{1}{n_+h_{Y_+}q^2} e_q^T \left(S_{nY_+}^{-1} \Sigma_{nY_+} S_{nY_+}^{-1}\right)e_q + \frac{1}{n_-h_{Y_-}q^2} e_q^T \left(S_{nY_-}^{-1} \Sigma_{nY_-} S_{nY_-}^{-1}\right)e_q \\
		&\quad +\frac{4}{n_+ h_{Y_+}} D_{nY_+,1}^2 e_2^T (S_{nY_+}^{-1} \Sigma_{nY_+} S_{nY_+}^{-1})_{22} e_2 + \frac{4}{n_- h_{Y_-}} D_{nY_-,1}^2 e_2^T (S_{nY_-}^{-1} \Sigma_{nY_-} S_{nY_-}^{-1})_{22} e_2\\
		&\quad - 2 \left[\frac{2 D_{nY_+,1}}{q n_+ h_{Y_+}} e_2^T (S_{nY_+}^{-1} \Sigma_{nY_+} S_{nY_+}^{-1})_{12,2} + \frac{2 D_{nY_-,1}}{q n_- h_{Y_-}} e_2^T (S_{nY_-}^{-1} \Sigma_{nY_-} S_{nY_-}^{-1})_{12,2}\right].
	\end{align*}
\end{proof}

\begin{proposition} \label{prop:fixed-n fuzzy t adjusted dist}
	Under \cref{assumption:m continuity,assumption:sigma continuity,assumption:kernel,assumption:density x,assumption:error distribution,assumption:bandwidth}, the fixed-$n$ adjusted \textit{t}-statistic for the fuzzy RD is given by
	\begin{equation} \label{eq:fixed-n fuzzy t adjusted dist}
	t_{\text{fuzzy, fixed-n}}^{\text{adj.}} = \frac{\tilde{\tau}_{\text{fuzzy}} - \widehat{\text{Bias}}(\tilde{\tau}_{\text{fuzzy}})_{\text{fixed-n}}}{\sqrt{\text{Var}(\tilde{\tau}_{\text{fuzzy}} - \widehat{\text{Bias}}(\tilde{\tau}_{\text{fuzzy}})_{\text{fixed-n}})}}.
	\end{equation}
\end{proposition}

\begin{proof} [Proof of \Cref{prop:fixed-n fuzzy t adjusted dist}]
	The numerator of \cref{eq:fixed-n fuzzy t adjusted dist} can be obtained by applying the fixed-$n$ bias result in \Cref{lemma:fixed-n bias and variance} to \cref{eq:fuzzy biases}. To compute the denominator, we again start with \cref{eq:fuzzy adj var}. For the fixed-$n$ result in \cref{eq:fuzzy adj var}, expressions for $\text{Var}(\hat{m}_{Y_+})_{\text{fixed-n}}$ and $\text{Var}(\hat{m}_{T_+})_{\text{fixed-n}}$ are given by \Cref{lemma:fixed-n bias and variance}, $\text{Var}( \widehat{\text{Bias}}(\hat{m}_{Y_+})_{\text{fixed-n}})$, $\text{Var}( \widehat{\text{Bias}}(\hat{m}_{T_+})_{\text{fixed-n}})$, $\text{Cov}(\hat{m}_{Y_+}, \widehat{\text{Bias}}(\hat{m}_{Y_+})_{\text{fixed-n}})$ and $\text{Cov}(\hat{m}_{T_+}, \widehat{\text{Bias}}(\hat{m}_{T_+})_{\text{fixed-n}})$ are derived in the proof of \Cref{prop:fixed-n sharp t}. We list the seven terms below and omit the small-order terms.
	\begin{align*}
		\text{Var}(\hat{m}_{Y_+})_{\text{fixed-n}} &= \frac{1}{n_+h_{Y_+}q^2} e_q^T \left(S_{nY_+}^{-1} \Sigma_{nY_+} S_{nY_+}^{-1}\right)e_q,\\
		\text{Var}(\hat{m}_{T_+})_{\text{fixed-n}} &= \frac{1}{n_+h_{T_+}q^2} e_q^T \left(S_{nT_+}^{-1} \Sigma_{nT_+} S_{nT_+}^{-1}\right)e_q,\\
		\text{Var}(\widehat{\text{Bias}}(\hat{m}_{Y_+})_{\text{fixed-n}}) &= \frac{4}{n_+ h_{Y_+}} D_{nY_+,1}^2 e_2^T (S_{nY_+}^{-1} \Sigma_{nY_+} S_{nY_+}^{-1})_{22} e_2,\\
		\text{Var}(\widehat{\text{Bias}}(\hat{m}_{T_+})_{\text{fixed-n}}) &= \frac{4}{n_+ h_{T_+}} D_{nT_+,1}^2 e_2^T (S_{nT_+}^{-1} \Sigma_{nT_+} S_{nT_+}^{-1})_{22} e_2,\\
		\text{Cov}(\hat{m}_{Y_+},\hat{m}_{T_+})_{\text{fixed-n}} &= \frac{1}{q^2 n_+ \sqrt{ h_{Y_+} h_{T_+}}} e_q^T
		\begin{bmatrix}
		(S_{nY_+}^{-1})_{11} \; (S_{nY_+}^{-1})_{12}
		\end{bmatrix} \Sigma_{nYT_+} 
		\begin{bmatrix}
		(S_{nT_+}^{-1})_{11} \; (S_{nT_+ }^{-1})_{12}
		\end{bmatrix}^T e_q,\\
		\text{Cov}(\hat{m}_{Y_+}, \widehat{\text{Bias}}(\hat{m}_{Y_+})_{\text{fixed-n}}) &= \frac{2 D_{nY_+,1}}{q n_+ h_{Y_+}} e_2^T (S_{nY_+}^{-1} \Sigma_{nY_+} S_{nY_+}^{-1})_{12,2},\\
		\text{Cov}(\hat{m}_{T_+}, \widehat{\text{Bias}}(\hat{m}_{T_+})_{\text{fixed-n}}) &= \frac{2 D_{nT_+,1}}{q n_+ h_{T_+}} e_2^T (S_{nT_+}^{-1} \Sigma_{nT_+} S_{nT_+}^{-1})_{12,2}.
	\end{align*}
	
	 We only need to compute the remaining three terms in \cref{eq:fuzzy adj var}. Consider the term $\text{Cov}(\widehat{\text{Bias}}(\hat{m}_{Y_+})_{\text{fixed-n}},\widehat{\text{Bias}}(\hat{m}_{T_+})_{\text{fixed-n}})$. Using the result in \cref{eq:fixed-n bias for mY+} and a similar result for $\widehat{\text{Bias}}(\hat{m}_{T_+})_{\text{fixed-n}}$, together with the result in \cref{eq:md2 and v} and a similar result for $\hat{m}_{T_+}^{(2)}$, it can be shown that
	\begin{equation*}
		\text{Cov}(\widehat{\text{Bias}}(\hat{m}_{Y_+})_{\text{fixed-n}},\widehat{\text{Bias}}(\hat{m}_{T_+})_{\text{fixed-n}}) = \frac{4 D_{nY_+,1} D_{nT_+,1}}{n_+ \sqrt{h_{Y_+} h_{T_+}}} e_2^T (S_{nY_+}^{-1} \Sigma_{nYT_+} S_{nT_+}^{-1})_{22} e_2,
	\end{equation*}
	similar to the proof in \Cref{thm:adjusted bandwidth}.
	
	Again, similar to the proof in \Cref{thm:adjusted bandwidth}, using \cref{eq:md2 and v}, \cref{eq:fixed-n bias for mY+} and \cref{eq:DnY+}, we have
	\begin{align*}
		\text{Cov}(\hat{m}_{Y_+},\widehat{\text{Bias}}(\hat{m}_{T_+})_{\text{fixed-n}}) &= \frac{2 D_{nT_+,1}}{q n_+ \sqrt{h_{Y_+} h_{T_+}}} e_2^T (S_{nY_+}^{-1} \Sigma_{nYT_+} S_{nT_+}^{-1})_{12,2}, \\
		\text{Cov}(\hat{m}_{T_+},\widehat{\text{Bias}}(\hat{m}_{Y_+})_{\text{fixed-n}}) &= \frac{2 D_{nY_+,1}}{q n_+ \sqrt{h_{Y_+} h_{T_+}}} e_2^T (S_{nT_+}^{-1} \Sigma_{nTY_+} S_{nY_+}^{-1})_{12,2}.
	\end{align*}
	
	Substitute the above ten results into \cref{eq:fuzzy adj var} to obtain a fixed-$n$ version of $\text{Var}( (\hat{m}_{Y_+} - \tau_0 \hat{m}_{T_+}) -(\widehat{\text{Bias}}(\hat{m}_{Y_+}) - \tau_0 \widehat{\text{Bias}}(\hat{m}_{T_+})) )$. The fixed-$n$ result for $\text{Var}( (\hat{m}_{Y_-} - \tau_0 \hat{m}_{T_-}) -(\widehat{\text{Bias}}(\hat{m}_{Y_-}) - \tau_0 \widehat{\text{Bias}}(\hat{m}_{T_-})) )$ can be obtained in a similar way. Adding up the two results gives the variance on the denominator of \cref{eq:fixed-n fuzzy t adjusted dist}.
\end{proof}

\newpage

\subsection{Additional figures and tables}

\subsubsection{Figure: the bias of bias-corrected estimators} \label{supp:fig bias of bias-corrected estimator}

To accompany \Cref{fig:lee bias} in the main text, this subsection contains additional  \Cref{fig:lee lm more} that compares the finite sample performance of LCQR and LLR in estimating the treatment effect.  

\begin{figure}[htp]
	\centering
\subfloat[Lee with heteroskedatic errors ]{\includegraphics[width=0.5\textwidth,keepaspectratio=TRUE]{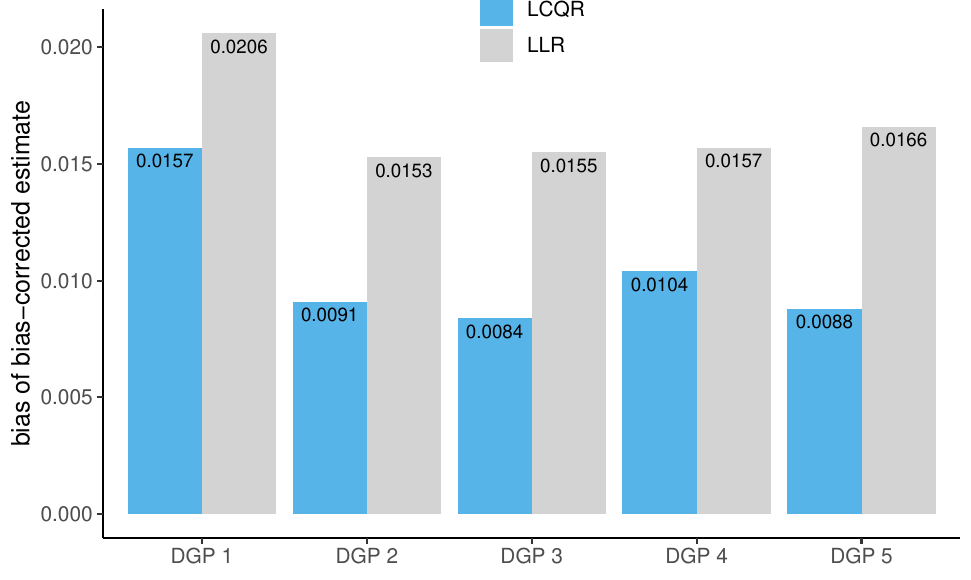}}\\
\subfloat[LM with homoskedastic errors]{	\includegraphics[width=0.5\textwidth,keepaspectratio=TRUE]{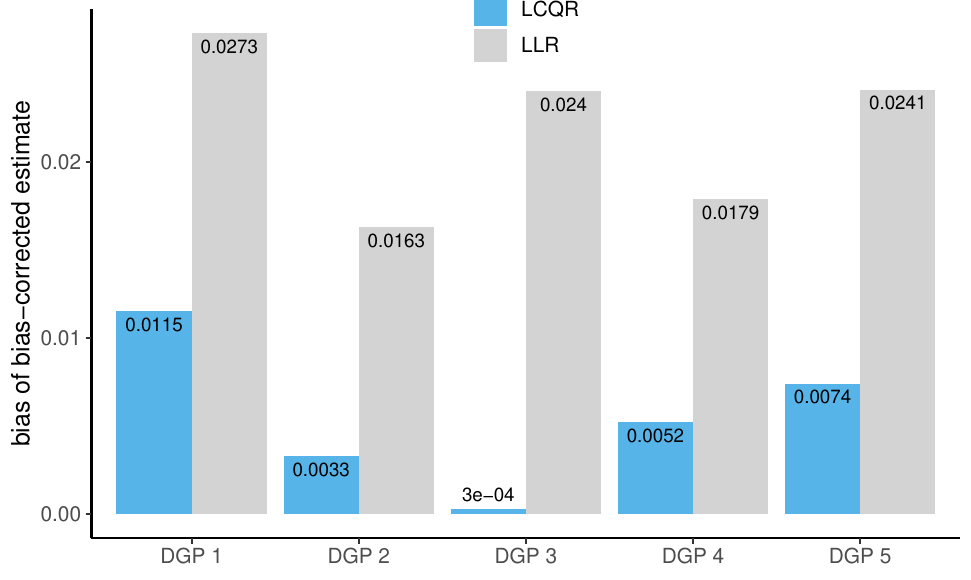}}\\
\subfloat[LM with heteroskedatic errors ]{	\includegraphics[width=0.5\textwidth,keepaspectratio=TRUE]{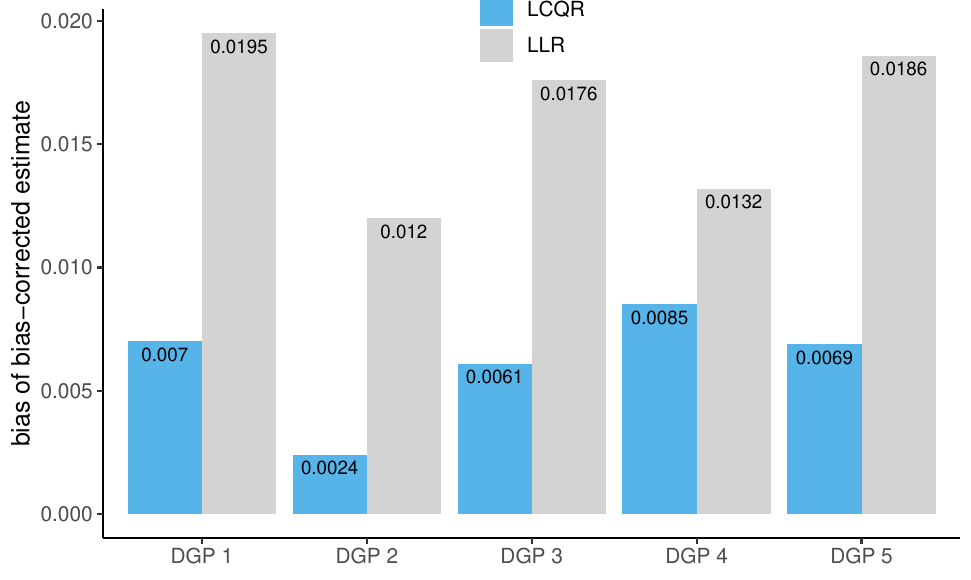}}
	\caption{Absolute value of average bias of the bias-corrected estimators, $\hat{\tau}_{\text{1bw}}^{\text{cqr,bc}}$ and $\hat{\tau}_{\text{1bw}}^{\text{robust,bc}}$ for the Lee and LM models. $\hat{\tau}_{\text{1bw}}^{\text{cqr,bc}}$ is the bias-corrected LCQR estimator. $\hat{\tau}_{\text{1bw}}^{\text{robust,bc}}$ is the bias-corrected LLR estimator. The result is based on $5000$ replications and the true treatment effect is $0.04$ for Lee and $-3.45$ for LM. The DGPs are described in the paper. }%
	\label{fig:lee lm more}%
\end{figure}

\subsubsection{Figure: LCQR and LLR at interior and boundary points} 

To motivate the use of LCQR, consider the nonlinear model in \cite{ruppertsheather1995JASA}, $Y = \sin(5\pi X) + 0.5\epsilon$, where $\epsilon$ follows a mixture normal distribution, $0.95N(0,1) + 0.05N(0,10^2)$, and $X$ follows a uniform distribution on $[0,1]$. It is clear from \Cref{fig:intro_plot} that LCQR exhibits less ``flapping" for both interior and boundary points. The relative stable behavior of LCQR on the boundary when data move away from normality is of particular importance to the estimation and inference in RD.

\begin{figure}[htp]
	\centering
	\subfloat[LLR versus LCQR]{{\includegraphics[width=0.48\linewidth,keepaspectratio]{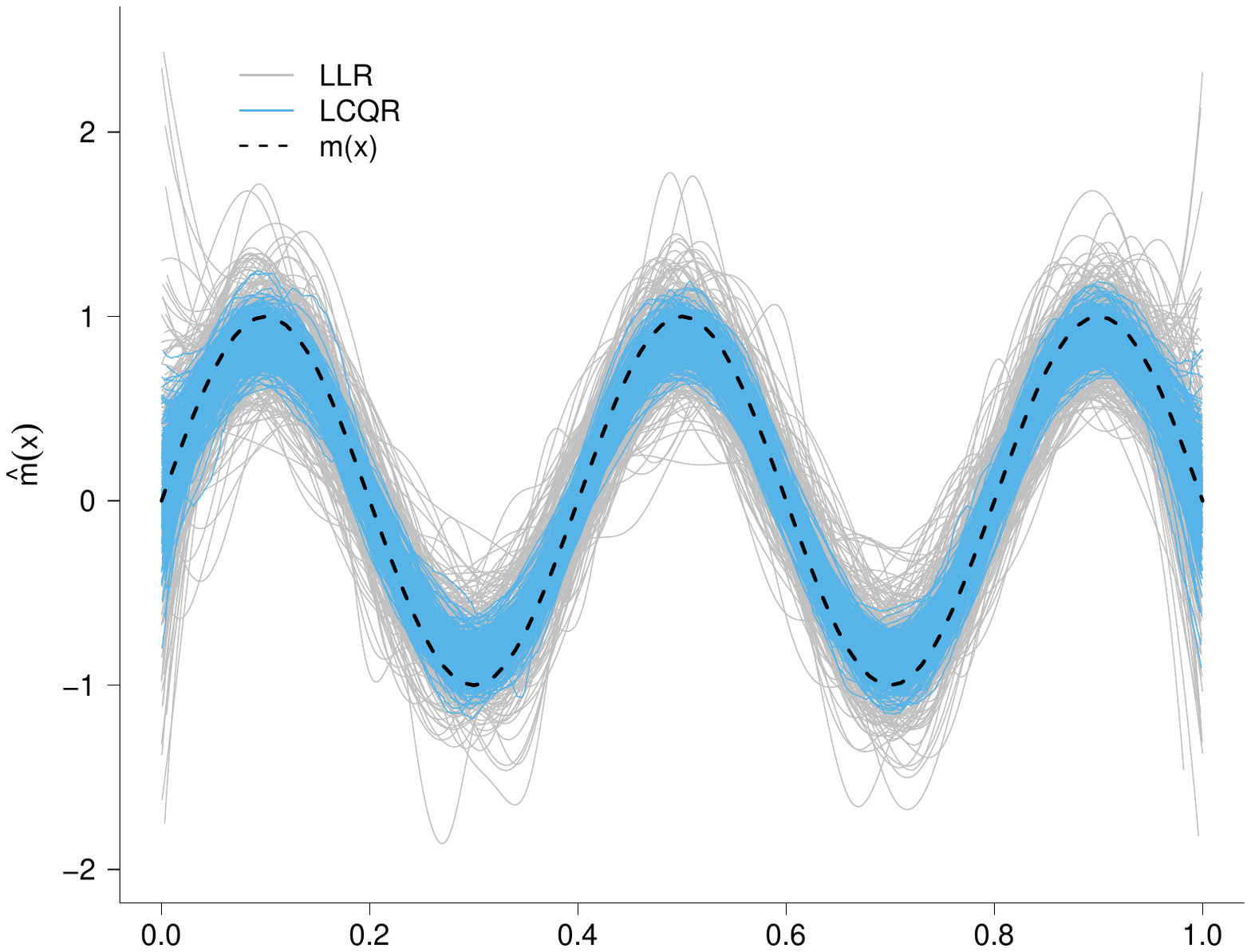} }}%
	\subfloat[Box plots of the estimates on the boundaries]{{\includegraphics[width=0.48\linewidth,keepaspectratio]{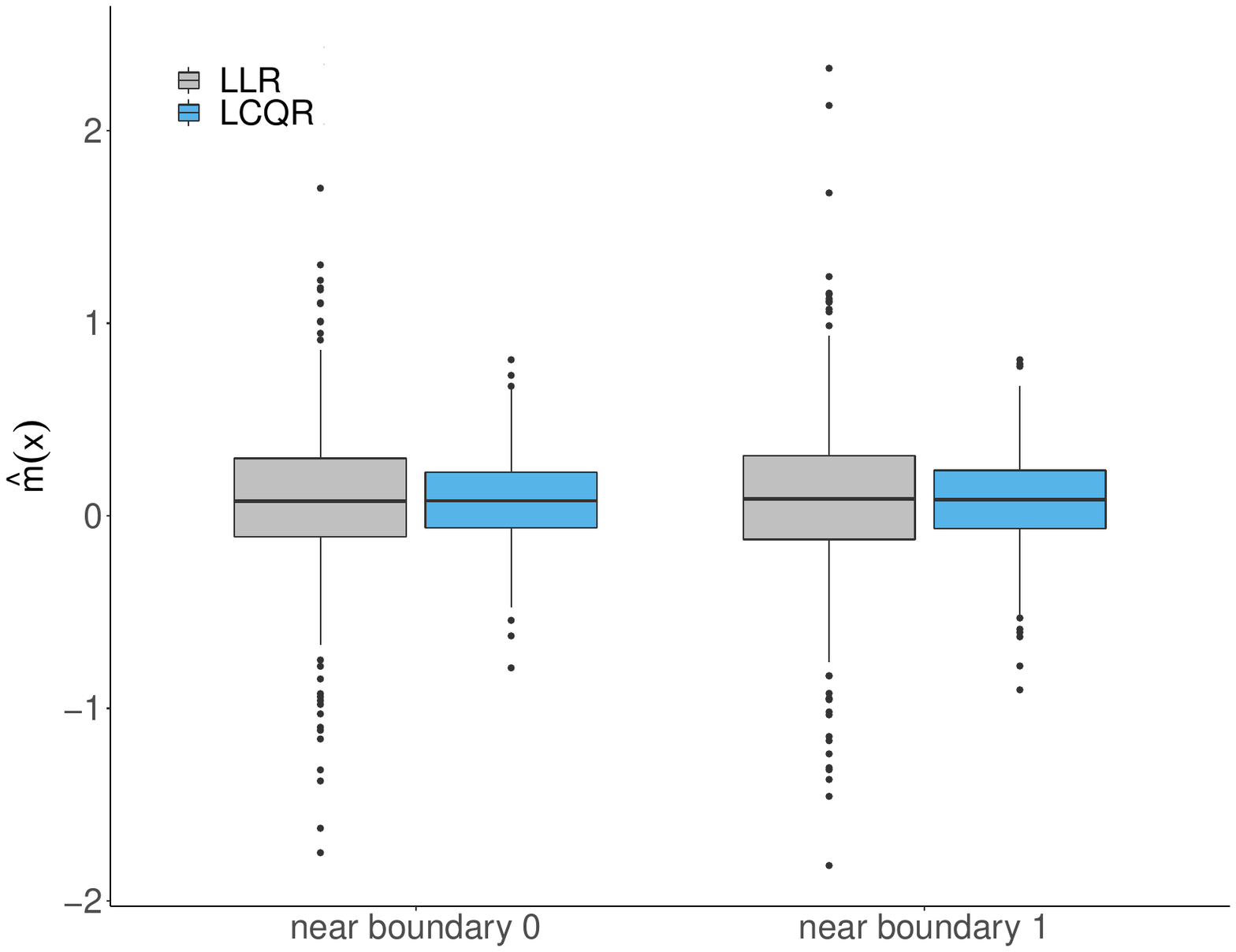} }}%
	\caption{Estimates of LLR and LCQR with a sample size of 400 and 400 replications. Both methods use the same direct plug-in bandwidth in \cite{ruppertsheather1995JASA}. $m(X)=\sin(5\pi X)$. }%
	\label{fig:intro_plot}%
\end{figure}




\subsubsection{Table: coverage probability with the rule-of-thumb bandwidth} \label{supp:lcqr with rot bandwidth}

This subsection presents \Cref{tb: coverage rot} that is similar to Table \ref{tb:main coverage} except that $\hat{\tau}_{\text{2bw}}^{\text{cqr}}$ and $\hat{\tau}_{\text{2bw}}^{\text{cqr,bc}}$ use the rule-of-thumb bandwidth described by Equation (4.3) in \cite{fan1996local}. \Cref{tb: coverage rot}  indicates that the proposed LCQR method has some robustness to the choice of bandwidth.

\begin{table}[htp] \centering 
	\begin{center}
	\caption{Coverage probability of 95\% confidence intervals in  Lee and LM models using the rule-of-thumb bandwidth for LCQR} 
	\label{tb: coverage rot} 
		\begin{threeparttable}
			\begin{tabular}{lccccc  ccccc} 
				\hline
				& \multicolumn{5}{c}{A. Lee with homoskedastic errors} &  \multicolumn{5}{c}{B. Lee  with heteroskedatic errors}\\
				\cmidrule(lr){2-6}\cmidrule(lr){7-11}
			\footnotesize	Estimator& \footnotesize DGP 1&\footnotesize DGP 2&\footnotesize DGP 3&\footnotesize DGP 4&\footnotesize DGP 5&\footnotesize  DGP 1&\footnotesize DGP 2&\footnotesize DGP 3&\footnotesize DGP 4&\footnotesize DGP 5\\ 	
				\hline
				$\hat{\tau}_{\text{2bw}}^{\text{cqr}}$ & 0.915 & 0.917 & 0.909 & 0.916 & 0.917 & 0.901 & 0.896 & 0.887 & 0.897 & 0.895\\
				$\hat{\tau}_{\text{2bw}}^{\text{cqr,bc}}$ & 0.976 & 0.963 & 0.968 & 0.976 & 0.965 & 0.969 & 0.956 & 0.958 & 0.965 & 0.951 \\
    &\\
	& \multicolumn{5}{c}{C. LM  with homoskedastic errors} &  \multicolumn{5}{c}{D. LM  with heteroskedatic errors}\\
				\cmidrule(lr){2-6}\cmidrule(lr){7-11}
		\footnotesize	Estimator& \footnotesize DGP 1&\footnotesize DGP 2&\footnotesize DGP 3&\footnotesize DGP 4&\footnotesize DGP 5&\footnotesize  DGP 1&\footnotesize DGP 2&\footnotesize DGP 3&\footnotesize DGP 4&\footnotesize DGP 5\\ 	
				\hline
				$\hat{\tau}_{\text{2bw}}^{\text{cqr}}$ & 0.888 & 0.891 & 0.880 & 0.892 & 0.859  & 0.876 & 0.876 & 0.870 & 0.878 & 0.845\\
				$\hat{\tau}_{\text{2bw}}^{\text{cqr,bc}}$ & 0.967 & 0.956 & 0.962 & 0.968 & 0.960 & 0.958 & 0.946 & 0.952 & 0.952 & 0.943\\	
				\hline 
				\end{tabular}
			\begin{tablenotes}[flushleft]
				\item[] \textit{Notes}: The reported numbers are the simulated coverage probabilities of the 95\% confidence intervals associated with different estimators. The results are based on 5000 replications with a sample size $n = 500$. The s.e. and adjusted s.e. for the LCQR estimator are obtained based on the asymptotic expressions from \Cref{thm:sharp results} and \Cref{thm:sharp t adjusted dist}. Estimators with superscript \texttt{bc} are both bias-corrected and s.e.-adjusted. The DGPs are described in the paper. 
			\end{tablenotes}
		\end{threeparttable}
	\end{center} 
\end{table}

\subsubsection{Table: coverage probability of fixed-n LCQR with small sample} \label{supp:fixed-n coverage probability small sample}

In this subsection, we decrease the sample size from $n=500$ to $300$ in the simulation study. We show that the fixed-$n$ approach indeed can improve the coverage when the sample size is relatively small, as reported in the last row of each panel of \Cref{tb: coverage small sample}.


\begin{table}[htp] \centering 
	\begin{center}
	\caption{Coverage probability of 95\% confidence intervals in  Lee and LM models, $n=300$ } 
	\label{tb: coverage small sample} 
		\begin{threeparttable}
			\begin{tabular}{lccccc  ccccc} 
				\hline
				& \multicolumn{5}{c}{A. Lee with homoskedastic errors} &  \multicolumn{5}{c}{B. Lee  with heteroskedastic errors}\\
				\cmidrule(lr){2-6}\cmidrule(lr){7-11}
			\footnotesize	Estimator& \footnotesize DGP 1&\footnotesize DGP 2&\footnotesize DGP 3&\footnotesize DGP 4&\footnotesize DGP 5&\footnotesize  DGP 1&\footnotesize DGP 2&\footnotesize DGP 3&\footnotesize DGP 4&\footnotesize DGP 5\\ 	
				\hline
			$\hat{\tau}_{\text{1bw}}^{\text{cqr}}$ & 0.906 & 0.892 & 0.900 & 0.903 & 0.886 & 0.882 & 0.872 & 0.877 & 0.873 & 0.871  \\
			    $\hat{\tau}_{\text{2bw}}^{\text{cqr}}$ & 0.895 & 0.887 & 0.890 & 0.890 & 0.875  & 0.869 & 0.859 & 0.868 & 0.859 & 0.857 \\
			        $\hat{\tau}_{\text{1bw}}^{\text{llr}}$ &0.926 & 0.926 & 0.929 & 0.917 & 0.943  & 0.923 & 0.922 & 0.928 & 0.918 & 0.945  \\
    $\hat{\tau}_{\text{1bw}}^{\text{cqr,bc}}$ & 0.958 & 0.941 & 0.940 & 0.948 & 0.923  &0.947 & 0.926 & 0.928 & 0.933 & 0.914 \\

    $\hat{\tau}_{\text{2bw}}^{\text{cqr,bc}}$ & 0.954 & 0.937 & 0.937 & 0.945 & 0.918& 0.939 & 0.924 & 0.929 & 0.929 & 0.906  \\

    $\hat{\tau}_{\text{1bw}}^{\text{robust,bc}}$ &0.927 & 0.928 & 0.930 & 0.921 & 0.946  & 0.923 & 0.925 & 0.928 & 0.920 & 0.946\\
    $\hat{\tau}_{\text{1bw,fixed-n}}^{\text{cqr,bc}}$ & 0.977 & 0.959 & 0.960 & 0.969 & 0.950  &0.958 & 0.939 & 0.944 & 0.945 & 0.936 \\
    &\\
	& \multicolumn{5}{c}{C. LM  with homoskedastic errors} &  \multicolumn{5}{c}{D. LM  with heteroskedastic errors}\\
				\cmidrule(lr){2-6}\cmidrule(lr){7-11}
		\footnotesize	Estimator& \footnotesize DGP 1&\footnotesize DGP 2&\footnotesize DGP 3&\footnotesize DGP 4&\footnotesize DGP 5&\footnotesize  DGP 1&\footnotesize DGP 2&\footnotesize DGP 3&\footnotesize DGP 4&\footnotesize DGP 5\\ 	
				\hline
				$\hat{\tau}_{\text{1bw}}^{\text{cqr}}$ & 0.791 & 0.820 & 0.826 & 0.803 & 0.832 & 0.613 & 0.668 & 0.700 & 0.640 & 0.728 \\
				    $\hat{\tau}_{\text{2bw}}^{\text{cqr}}$ & 0.805 & 0.839 & 0.841 & 0.815 & 0.837 & 0.622 & 0.689 & 0.706 & 0.657 & 0.732  \\
				        $\hat{\tau}_{\text{1bw}}^{\text{llr}}$ &0.907 & 0.917 & 0.921 & 0.905 & 0.936 & 0.899 & 0.905 & 0.914 & 0.896 & 0.926 \\
    $\hat{\tau}_{\text{1bw}}^{\text{cqr,bc}}$ &0.959 & 0.941 & 0.939 & 0.946 & 0.923   & 0.941 & 0.928 & 0.927 & 0.931 & 0.910\\

    $\hat{\tau}_{\text{2bw}}^{\text{cqr,bc}}$ & 0.953 & 0.936 & 0.936 & 0.944 & 0.915& 0.935 & 0.924 & 0.924 & 0.928 & 0.908\\

    $\hat{\tau}_{\text{1bw}}^{\text{robust, bc}}$ & 0.926 & 0.930 & 0.934 & 0.923 & 0.946  &0.927 & 0.930 & 0.933 & 0.923 & 0.946  \\
    $\hat{\tau}_{\text{1bw,fixed-n}}^{\text{cqr,bc}}$ & 0.975 & 0.956 & 0.959 & 0.966 & 0.951 & 0.956 & 0.939 & 0.942 & 0.946 & 0.938 \\
				\hline 
			\end{tabular}
			\begin{tablenotes}[flushleft]
				\item[] \textit{Notes}: The reported numbers are the simulated coverage probabilities of the 95\% confidence intervals associated with different estimators. The results are based on 5000 replications with a sample size $n = 300$. The s.e. and adjusted s.e. for the LCQR estimator are obtained based on the asymptotic expressions from \Cref{thm:sharp results} and \Cref{thm:sharp t adjusted dist}, except for $\hat{\tau}_{\text{1bw,fixed-n}}^{\text{cqr,bc}}$ where fixed-$n$ approximations are used. Estimators with superscript \texttt{bc} are both bias-corrected and s.e.-adjusted. The result of $\hat{\tau}_{\text{1bw}}^{\text{robust,bc}}$ is based on the CE-optimal bandwidth. The DGPs are described in the paper.
			\end{tablenotes}
		\end{threeparttable}
	\end{center} 
\end{table}

\subsubsection{Table: LCQR for sharp kink RD} \label{supp:lcqr with kink}

We consider the LM model used for the simulation study, but now focus on the difference in derivatives around the cutoff: $18.49-2.3=16.19$, as in a sharp kink RD design. \Cref{tb:lcqr kink} shows that LCQR could  outperform the local polynomial regression for estimating derivatives when data are non-normal; see e.g. DGP 2 - 5.

\begin{table}[htp] \centering 
	\begin{center}
	\caption{LCQR for sharp kink RD} 
	\label{tb:lcqr kink} 
		\begin{threeparttable}
			\begin{tabular}{lccccc  ccccc} 
				\hline
	
    				 & \multicolumn{10}{c}{LM  with $\tau_{sharp \  kink}=16.19$}\\
				 \cline{2-11}
				& \multicolumn{5}{c}{Homoskedastic errors } &  \multicolumn{5}{c}{Heteroskedastic errors}\\
				\cmidrule(lr){2-6}\cmidrule(lr){7-11}
		\footnotesize		Estimator& 	\footnotesize DGP 1&	\footnotesize DGP 2&	\footnotesize DGP 3&	\footnotesize DGP 4&	\footnotesize DGP 5& 	\footnotesize DGP 1&	\footnotesize DGP 2&	\footnotesize DGP 3&	\footnotesize DGP 4&	\footnotesize DGP 5\\ 	
				\hline
				LCQR & 15.91 & 16.14   & 15.89  & 16.12  & 16.13   
				& 15.93 & 16.09  & 15.91  & 16.07  & 16.08 \\
				(s.e.)    & \footnotesize (10.94)& \footnotesize (12.90)& \footnotesize (13.83)& \footnotesize (11.47)& \footnotesize (11.87) &\footnotesize (6.27)& \footnotesize (7.35)& \footnotesize (7.88)& \footnotesize (6.57)& \footnotesize (6.79)\\
				
    LPR & 15.93 & 16.13 & 15.69 & 16.06 & 15.84  
    & 15.96 & 16.07 & 15.82 & 16.03 & 15.90\\
     (s.e.)   &\footnotesize (10.47)& \footnotesize (14.77)&\footnotesize  (18.11)& \footnotesize (12.25)& \footnotesize (25.42) 
        &\footnotesize (6.10)& \footnotesize (8.61)& \footnotesize (10.53)& \footnotesize (7.13)& \footnotesize (14.80)\\
				\hline 
			\end{tabular}
			\begin{tablenotes}[flushleft]
				\item[] \textit{Notes}: The reported numbers are the simulated averages and  standard errors (in brackets) of the associated  estimators. The results are based on 5000 replications with a sample size $n = 500$. The DGPs are as described in the paper for the  LM model, yet the focus here is on the difference in first derivatives. For both LCQR and LPR (local polynomial regression), we consider the 3rd-order polynomial with the fixed bandwidth $=0.3$ and the triangular kernel. The R code to replicate this table can be downloaded from \url{https://xhuang.netlify.app/post/r-code-to-replicate-rd-tables/}.
			\end{tablenotes}
		\end{threeparttable}
	\end{center} 
\end{table}

\subsubsection{Table: simulation results for sharp RD with covariates} \label{supp:covariates}

In this subsection we use the same DGP as for Table SA-1 in \cite{calonico2019covariates}. We briefly describe the DGP below. Let $Z_i$ be the  covariate. Consider a sample size of $n=1000$ and $5000$ replications. For each $i=1,\cdots,n$, we have
\begin{equation*}
    Y_i = m_{y,j}(X_i,Z_i)+\epsilon_{y,i}, \quad Z_i = m_{z}(X_i)+\epsilon_{z,i}, \quad X_i \sim 2\times \text{Beta}(2,4)-1
\end{equation*}
with
\begin{equation*}
    \begin{pmatrix}
    \epsilon_{y,i}\\
    \epsilon_{z,i}
    \end{pmatrix}
    \sim N(0,\Sigma_j),\quad
    \Sigma_j = 
    \begin{pmatrix}
    \sigma_y^2 & \rho_j \sigma_y \sigma_z\\
    \rho_j \sigma_y \sigma_z & \sigma_z^2
    \end{pmatrix}
\end{equation*}
and $j=1,2,3,4$, corresponding to the following four models.
\begin{itemize}
    \item \textbf{Model 1} has no covariate and is the same as \cref{eq:lee}
    \begin{equation*} 
	m_{y,1}(X_i,Z_i) = \begin{cases}
	0.48 + 1.27 X_i + 7.18 X_i^2 + 20.21 X_i^3 + 21.54 X_i^4 + 7.33 X_i^5 & \text{ if } X_i < 0,\\
	0.52 + 0.84 X_i - 3.00 X_i^2 + 7.99 X_i^3 - 9.01 X_i^4 + 3.56 X_i^5 & \text{ if } X_i \geq 0,
	\end{cases}
\end{equation*}
and let $\sigma_y=0.1295$ and $\sigma_z = 0.1353$.
\item \textbf{Model 2} adds one covariate, and let $\rho = 0.2692$,
\end{itemize}
\begin{equation*} 
	m_{y,2}(X_i,Z_i) = \begin{cases}
	0.36 + 0.96 X_i + 5.47 X_i^2 + 15.28 X_i^3 + 15.87 X_i^4 + 5.14 X_i^5 + 0.22Z_i & \text{ if } X_i < 0,\\
	0.38 + 0.62 X_i - 2.84 X_i^2 + 8.42 X_i^3 - 10.24 X_i^4 + 4.31 X_i^5 + 0.28Z_i & \text{ if } X_i \geq 0,
	\end{cases}
\end{equation*}
\begin{equation*} 
	m_{z}(X_i) = \begin{cases}
	0.49 + 1.06 X_i + 5.74 X_i^2 + 17.14 X_i^3 + 19.75 X_i^4 + 7.47 X_i^5 & \text{ if } X_i < 0,\\
	0.49 + 0.61 X_i + 0.23 X_i^2 - 3.46 X_i^3 + 6.43 X_i^4 - 3.48 X_i^5 & \text{ if } X_i \geq 0.
	\end{cases}
\end{equation*}
\begin{itemize}
\item \textbf{Model 3} is the same as Model 2 except for $\rho = 0$.
\item \textbf{Model 4} is the same as Model 2 except for $\rho = 2 \times 0.2692$.
\end{itemize}

The true value for $\tau$ is $0.04$ in Model 1 and approximately $0.05$ in Models 2-4. \Cref{tb:covariate} reports $\sqrt{\text{MSE}}$, bias as a percentage of $\tau$ and empirical coverage (EC) for the confidence intervals based on $\hat{\tau}$ and $\tilde{\tau}$. The EC for $\hat{\tau}$ is obtained using bias-corrected, s.e.-adjusted \textit{t}-statistic in \Cref{thm:sharp t adjusted dist}; the EC for $\tilde{\tau}$ is obtained using the same \textit{t}-statistic for $\hat{\tau}$ but replacing $\hat{\tau}$ with $\tilde{\tau}$ on the numerator. See also \Cref{sec:covariates} for a discussion of this \textit{ad hoc} method for $\tilde{\tau}$. The last column in \Cref{tb:covariate} gives reasonably good coverage probabilities, suggesting the \textit{ad hoc} approach described in \Cref{sec:covariates} works well under the considered DGP. However, more simulation studies are needed to investigate its performance. 
\begin{table}[htp] \centering 
	\begin{center}
	\caption{Simulation results using a single bandwidth} 
	\label{tb:covariate} 

		\begin{threeparttable}
			\begin{tabular}{lccc cc ccc} 
				\hline
				& &\multicolumn{3}{c}{$\hat{\tau}$ in \cref{eq:without covariates}}& &\multicolumn{3}{c}{$\tilde{\tau}$ in \cref{eq:with covariates}}\\ 
				\cmidrule(lr){3-5} \cmidrule(lr){7-9}
				&&$\sqrt{\text{MSE}}$&\text{Bias (\%)}&\text{EC}&&$\sqrt{\text{MSE}}$&\text{Bias (\%)}&\text{EC}  \\ 
				\hline
				\text{Model 1} & \ \ \ \ \ &0.046 & 0.369 & 0.953 &\ \ \ \ \ &0.046 & 0.368 & 0.952 \\
                \text{Model 2} & &0.049 & 0.256 & 0.942 && 0.043 & 0.178 & 0.968 \\
                \text{Model 3} & &0.047 & 0.275 & 0.938 & &0.046 & 0.231 & 0.944 \\
                \text{Model 4} & &0.053 & 0.275 & 0.949 & &0.038 & 0.139 & 0.980 \\
				\hline
			\end{tabular}
		    \begin{tablenotes}[flushleft]
		    	\item[] \textit{Notes}: We use a single bandwidth $h=0.15$ for estimation, bias-correction, s.e.-adjustment in all four models. This number is chosen to mimic the bandwidth used in the simulation section in \cite{calonico2019covariates}. All numbers in the table are based on $5000$ replications. The R code to replicate this table can be downloaded from \url{https://xhuang.netlify.app/post/r-code-to-replicate-rd-tables/}.
		    \end{tablenotes}
		\end{threeparttable}
	\end{center}
\end{table}

\end{appendices}
\end{document}